\documentclass[12pt]{article}

\usepackage[left=1in,right=1in,top=1in,bottom=1in]{geometry}
\usepackage{graphicx}
\usepackage{natbib}
\usepackage{setspace}
\usepackage{soul}               

\usepackage[]{appendix}    

\usepackage[shortlabels]{enumitem}
\setlist[description]{leftmargin=0em,labelindent=\parindent}

\usepackage{amsmath,amsthm}

\newcommand{\E}{\operatorname{E}}

\theoremstyle{definition}
\newtheorem{lemma}{Lemma}

\newtheorem{prop}{Proposition}

\theoremstyle{remark}

\usepackage[T1]{fontenc}

\usepackage[scaled=.85]{beramono}
\usepackage[type1]{cabin}
\usepackage{amsmath,amsthm}
\usepackage[libertine]{newtxmath}
\usepackage[scr=rsfso]{mathalfa}
\usepackage{bm}
\usepackage[semibold]{libertine}

\usepackage{titlesec}
\titleformat{\section}[block]{\bf\Large}{\thesection\quad}{0pt}{}
\titleformat{\subsection}[block]{\bf\large}{\thesubsection\quad}{0pt}{}

\newenvironment{figurenotes}[1][Note]{\begin{minipage}[t]{\linewidth}\footnotesize{\itshape#1: }}{\end{minipage}}

\usepackage[xcolor=dvipsnames]{changes}
\definecolor{csub-blue}{RGB}{0, 53, 148}
\usepackage[unicode,hyperfootnotes=true,linktocpage,pagebackref]{hyperref}
 \hypersetup{%
   colorlinks = true,
   pdfborder = {0 0 0},
   citecolor = csub-blue,
   linkcolor = csub-blue,
   urlcolor  = csub-blue,
 }%
\makeatletter
\def\@fnsymbol#1{\ensuremath{\ifcase#1\or \mathparagraph\or \|\or **\or \dagger\dagger \or \ddagger\ddagger \else\@ctrerr\fi}}
\makeatother

\usepackage{pdfpages}

\title{\textsf{\large Discretionary Extensions to Unemployment-Insurance Compensation\\and Some Potential Costs for a McCall Worker}}
\author{Rich Ryan\thanks{Email: \href{mailto:richryan@csub.edu}{richryan@csub.edu}. Department of Economics, California State University, Bakersfield,
  9001 Stockdale Hwy, Mail Stop: 20 BDC, Bakersfield, CA 93311, USA.}} 
\date{December 20, 2023}

\begin{document}
\maketitle


\begin{abstract}  
Unemployment insurance provides temporary cash benefits to eligible unemployed workers.
Benefits are sometimes extended by discretion during economic slumps.
In a model that features temporary benefits and sequential job opportunities,
a worker's reservation wages are studied
when policymakers can make discretionary extensions to benefits.
A worker's optimal labor-supply choice is characterized by a sequence of reservation wages that increases with weeks of remaining benefits.
The possibility of an extension raises the entire sequence of reservation wages,
meaning a worker is more selective when accepting job offers throughout their spell of unemployment.
The welfare consequences of misperceiving the probability and length of an extension are investigated.
Properties of the model can help policymakers interpret data on reservation wages,
  which may be important if extended benefits are used more often in response to economic slumps, virus pandemics, extreme heat, and natural disasters.
\end{abstract}

\vspace{3em}

\textbf{Keywords}: extended benefits, job search, reservation wages, unemployment, unemployment benefits, unemployment compensation, unemployment insurance

\vspace{1em}

\textbf{JEL Codes}: J22, J29, J31, J64, J65

\vfill


\clearpage
\pagebreak

\section{Introduction}


Each month millions of workers lose their jobs.
Unemployed workers search for jobs by
answering online ads,
contacting employers,
using the services of employment agencies, and
asking friends and relatives about employment opportunities.
Workers understand the the types and frequency of job offers they expect to receive. 
Depending on how long they are unemployed,
workers are entitled to unemployment-insurance compensation, typically 26 weeks.
Upon receiving a job offer,
a worker has the option to reject the offer and continue their search or accept the job.
A sequence of reservation wages describes optimal choices:
any wage offer above the reservation level should be accepted.
As more weeks of unemployment are endured and fewer weeks of unemployment insurance remain,
the reservation wage will fall,
indicating that workers are less selective as benefits expire.
Less is known about optimal decisions, however, 
when policymakers can extend benefits by discretion.

\begin{description}
\item[Extensions to UI benefits.] In brief, here is the main issue:
Unemployed workers are often entitled to claim 26 weeks of UI compensation.
In periods of economic distress,
policymakers sometimes
extend the number of weeks an individual might claim UI compensation by discretion.
When a worker misperceives the probability and length of an extension,
what are the costs?
\end{description}

I answer the question from a worker's perspective.
Imagine that Sonia is unemployed and submitting her resume to online ads.
Each week Sonia receives a job offer
that she can accept or reject.
Meanwhile,
her unemployment compensation benefits are expiring.
Between job offers, however,
there is a chance that benefits are extended.
An extension would allow Sonia to claim additional weeks of UI compensation.
If an extension is likely,
then Sonia can reject low-wage offers,
knowing that her search will likely be supported by extended UI compensation.
Sonia's reservation wage increases.
Sonia, though,
might not know
the true probability that benefits are extended and
the length of the extension.

I am interested in Sonia's welfare costs 
when she misperceives the probability and length of an extension.
If Sonia believes an extension is unlikely to happen,
so the true probability of an extension is higher than what Sonia believes,
then Sonia risks accepting job offers she would like to reject.
In this scenario,
if Sonia knew the true probability,
then on average she would like to spend more time searching
in order to find a higher-paying job.
To understand these costs,
I develop a formal dynamic model of sequential job search with expiring UI benefits, 
which allows for the possibility that benefits are extended.
\citet{burdett_1979} studied expiring benefits in \citeauthor{mccall_1970}'s \citeyearpar{mccall_1970} environment of sequential job search.
I add to this environment the possibility that benefits are extended.
I show how perceptions about the possibility of an extension affect reservation wages,
which determine a worker's job-acceptance criteria.

The decision-making environment I consider 
captures some features of UI extensions made during and after
the Great Recession and the COVID-19 pandemic.
For example,
extensions to UI compensation benefits created by Congress during the Great Recession
under the Emergency Unemployment Compensation program
expired three times and
each time Congress had to reauthorize the program to the previous expiration date \citep{rothstein_2011}.
In fact,
temporary additional UI benefits have been created by Congress 9 times:
in 1958, 1961, 1971, 1974, 1982, 1991, 2002, 2008, and 2020.
These policies have extended the number of weeks a worker could claim UI compensation
anywhere from 6 to 53 weeks \citep{whittaker_isaacs_2022}.
As I document in section \ref{sec:descr-decis-making-ext},
there is much uncertainty and risk from a worker's perspective
when extensions are made by discretion.

Because discretionary extensions made to UI compensation benefits are endogenous by design,
making econometric work challenging,
there is scope for investigating a theoretical model.
I present the model in sections \ref{sec:job-search-basic} and \ref{sec:job-search-benefits}.
Section \ref{sec:welfare} presents a numerical example.
The exercise is not meant to be definitive.
Instead, 
it is meant to demonstrate the theoretical properties
established in sections \ref{sec:job-search-basic} and \ref{sec:job-search-benefits}.
Nevertheless,
it is worth noting that the welfare calculations
imply small costs to misperceiving benefit extensions.
Finding small costs is discussed in the context of the literature in section \ref{sec:discussion}.
Section \ref{sec:conclusion} concludes.

\section{Two Episodes of Extended Benefits}
\label{sec:descr-decis-making-ext}

There are two ways for UI compensation benefits to be extended:
automatically and by discretion.
Benefits are automatically extended through
the Extended Benefits (EB) program of the UI system.
The UI system is a joint federal--state partnership
with essentially 53 different systems
(the 50 states, Puerto Rico, the District of Columbia, and the US Virgin Islands).
Regular UI in most states provides 26 weeks of UI compensation.
The EB program extends the number of weeks a worker can claim UI compensation 
by 13 or 20 weeks 
when state-specific unemployment-rate triggers are reached.
But
``in practice,
the required EB trigger is set to such a high level of unemployment
that the majority of states do not trigger onto EB in most recessions''
\citep[][1--2]{whittaker_isaacs_2022}.
In response,
Congress often temporarily extends UI benefits
like they recently did under
the  Coronavirus Aid, Relief, and Economic Security Act (CARES Act) during the COVID-19 pandemic and
the Emergency Unemployment Compensation (EUC) program during the Great Recession \citep{fujita_2010}.
From a worker's perspective,
both types of extensions involve uncertainty and risk:
Automatic extensions require a forecast of unemployment-rate statistics and
discretionary extensions require a forecast of legislative action.\footnote{Triggers for EB weeks may involve thresholds based on
  (1) a moving average of a state's unemployment rate or
  (2) the insured unemployment rate,
  which is the ratio of unemployment-compensation claimants divided by individuals in jobs covered by unemployment compensation (so, for example,
  not the self-employed and gig-economy workers) \citep{whittaker_isaacs_2022}.

  \citet[][158, figure 2]{chodorow-reich_coglianese_2019} show data on
  the number UI recipients who claim regular, EB, and emergency weeks
  created by temporary extension.
  Extended benefits are relied upon less frequently than emergency benefits.}

Details of the EUC program and CARES Act
justify investigating the costs associated with
misperceiving the probability and length of an extension.
Section \ref{sec:ext-thought-experiment} considers a thought experiment
about how difficult it was for a worker to perceive that
they would be entitled to claim up to 99 weeks of UI benefits after the Great Recession.
Section \ref{sec:ext-google-data} shares an index of Google searches for ``unemployment benefits extension,''
which spikes around potential cutoffs to extended benefits after the COVID-19 pandemic.
Both episodes document the significant uncertainty
a worker must manage when making decisions about accepting job offers when benefits can be extended.


\subsection{Extensions Created in Response to the Great Recession}
\label{sec:ext-thought-experiment}

Imagine that Vernon lived in California and became unemployed in early January 2010
when the United States was dealing with the fallout from the Great Recession.
The EB program in California had been triggered on since February 22, 2009,
so Vernon could expect to receive
26 regular weeks of UI coverage and
20 EB weeks for a total of 46 weeks---as long as California did not trigger off its extended-benefits thresholds,
requiring Vernon to forecast unemployment-rate statistics.
After the 46 weeks,
Vernon could reasonably expect to receive no temporary benefits
because the EUC program was
set to expire in late February 2010 \citep[][150, table 1]{rothstein_2011}.
Yet, around June 27, 2010, 25 weeks or so after Vernon became unemployed,
California's EB program had triggered off and
the federal EUC program had also expired.\footnote{``When a state triggers off of an EB period, all EB benefit payments in the state cease immediately'' \citep[][7]{whittaker_isaacs_2022}.}

What Vernon expected then is unknowable,
but by July 25, 2010,
the EB program had triggered on in California and the EUC program had been reauthorized.
Vernon at that time could reasonably expect---in addition
to the 26 weeks of regular UI compensation just received---20 EB weeks plus
$20 + 14 + 13 + 6 = 53$ weeks of UI compensation
under the 4 tiers of the reauthorized EUC program
(where some of the compensation would be paid retroactively) \citep[][153]{rothstein_2011}.
In total,
Vernon could have received up to $26+20+53 = 99$  weeks of UI compensation.
Yet,
the $99$ weeks belies the uncertainty and risk faced by Vernon.\footnote{The dates for this thought experiment come from
  data in \citep[][150, table 1]{rothstein_2011} and
  data available at \url{https://oui.doleta.gov/unemploy/claims_arch.asp}.
  Trigger dates for extended benefits and the EUC program are included in appendix \ref{sec:documentation-eb-euc}.}

Vernon's predicament was pointed out by \citet{kahn_2011}
in the context of UI extensions in the Great Recession aftermath.
After the EUC program was authorized in June 2008,
temporary extensions were enacted and reauthorized by Congress in
``fits and starts'' \citep[][149]{rothstein_2011}.
``For much of the program's history,
the expiration date was quite close.
Indeed, on three occasions\dots Congress allowed the program to expire.
Each time,
Congress eventually reauthorized it retroactive to the previous expiration date'' \citep[][150]{rothstein_2011}.

\subsection{Extensions Created in Response to the COVID-19 Pandemic}
\label{sec:ext-google-data}

Uncertainty similarly surrounded CARES Act extensions.
The CARES Act was created by Congress
in response to the recession caused by the COVID-19 pandemic.
Two CARES Act programs extended UI compensation benefits.
The Pandemic Emergency Unemployment Compensation (PEUC) program
provided additional weeks of federally funded benefits
similar to the EUC program in the Great Recession.
And the Pandemic Unemployment Assistance (PUA) program expanded coverage
to individuals who would be ineligible for UI benefits
(the self-employed, independent contractors, gig-economy workers) and
to individuals unemployed due to COVID-19-related reasons.
Both programs were reauthorized multiple times.
Data on internet searches suggest that
workers were concerned about access to extended benefits throughout 2020 and 2021.

Figure \ref{fig:index} depicts a weekly index of Google searches for ``unemployment benefits extension.''
In early 2020, before the COVID-19 pandemic, the index was near zero.
Google searches began rising by mid March and
peaked locally in the week after March 27, 2020,
which coincided with the signing of the CARES Act into law,
creating
``several temporary, now-expired UI programs'' \citep{whittaker_isaacs_2022}.

\begin{figure}[htbp]
  \centerline{\includegraphics[width=1.0\textwidth]{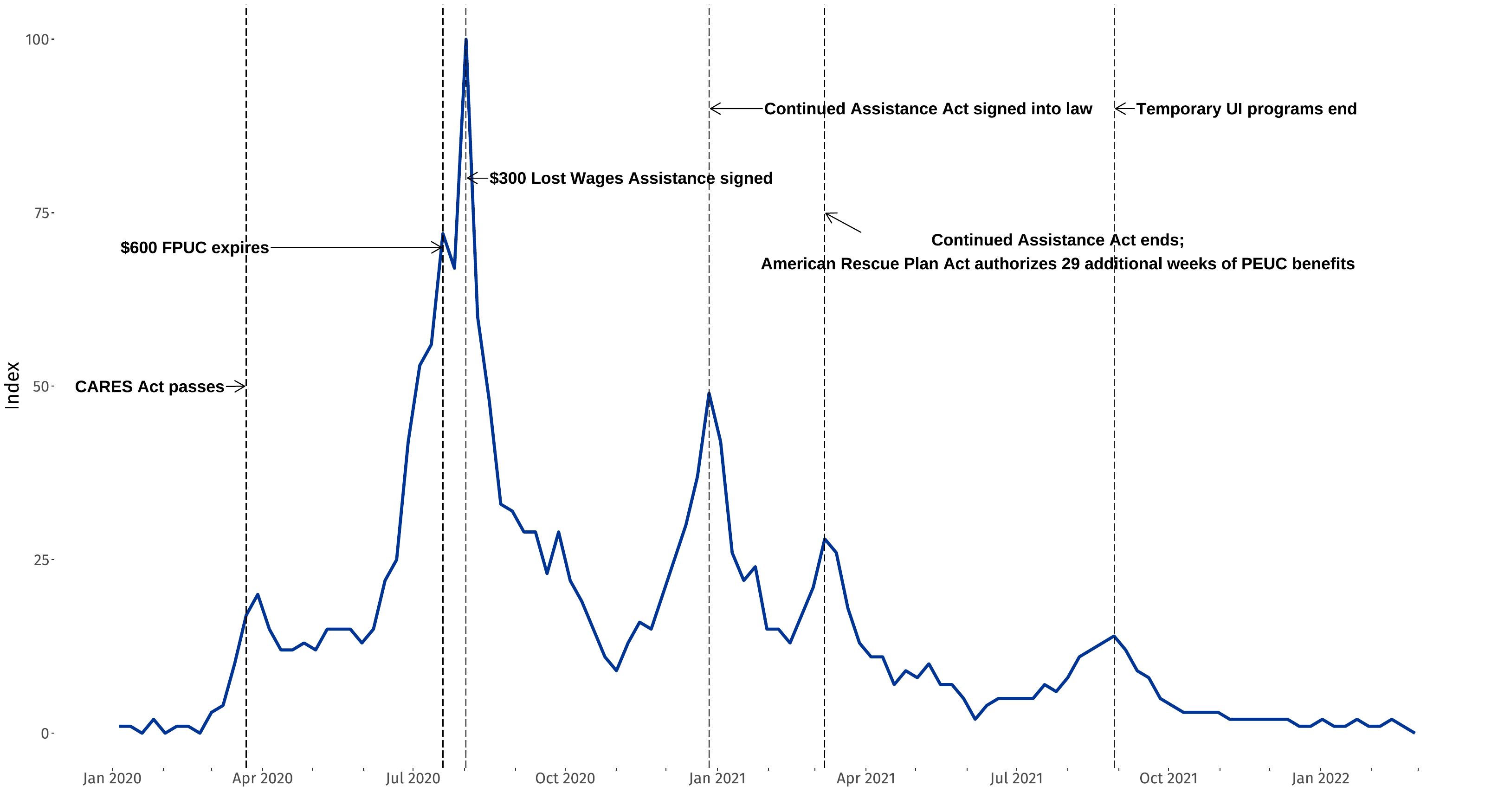}}
\caption[]{\label{fig:index} Index of Google searches for ``unemployment benefits extension.''}
  \begin{figurenotes}[Note]
    These data represent relative search frequency from the week containing January 5, 2020 to the week containing February 27, 2022.
    The series is normalized so that the highest-intensity week is set to 100.
  \end{figurenotes}
  \begin{figurenotes}[Source]
    Google Trends.
\end{figurenotes}  
\end{figure}

During the third week of June 2020 the index was around 25.
But when the additional \$600-per-week benefit
created by the Federal Pandemic Unemployment Compensation (FPUC) program
expired in the week ending July 25, 2020,
the indexed neared its global peak.
The 100 peak was reached in the week of August 8, 2020
when President Donald Trump
authorized the \$300-per-week Lost Wages Assistance benefit.

A little over 21 weeks passed between
March 13, 2020 
(when a nationwide emergency was declared) and
August 8, 2020.\footnote{As of August 11, 2023, the March 13 date
  is reported on the Centers for Disease Control and Prevention's
  COVID-19 \href{https://www.cdc.gov/museum/timeline/covid19.html}{timeline}.}
So up to this point,
many workers would have been covered by 26 weeks of regular UI benefits.
The local peaks that follow have to do with benefit extensions.
In the week of December 27, 2020,
the Continued Assistance Act was signed into law,
which
increased the maximum number of PEUC weeks
from 13 to 24 and
increased the maximum number of PUA weeks from 39 to 50.
In the week of March 11, 2021,
the American Rescue Plan
increased the maximum number of PEUC weeks to 53 and
increased the maximum number of PUA weeks to 79.
Both the PEUC and PUA programs were extended through September 4, 2021,
when temporary UI programs ended \citep[][4]{spadafora_2023}.
As documented in figure \ref{fig:index},
each of these policy milestones coincided with increased Google searches
for ``unemployment benefits extension.''\footnote{In addition,
  workers in 18 states faced a cutoff in compensation-benefit generosity
  when their states opted out of the FPUC and PUA progams in June 2021
  before the programs were set to expire in September 2021,
  citing labor-supply issues \citep{holzer_hubbard_strain_2021}.

  There are many other parts to the UI system that are beyond the scope of this paper.
  \citet{whittaker_isaacs_2022,whittaker_isaacs_2014} and \citet{spadafora_2023} provide excellent, detailed coverage of the UI system.
  \citet{fujita_2010}, \citet{rothstein_2011}, and \citet{chodorow-reich_coglianese_2019} provide excellet discussion of the many temporary programs.}

\section{Job Search with Expiring Benefits}
\label{sec:job-search-basic}


Before turning to a decision-making environment where UI benefits can be extended,
I start with a canonical model of sequential search.
In the environment,
a worker is entitled to $N$ periods of UI compensation benefits.
The worker searches for a job,
taking market conditions as given.
Market conditions are summarized by a wage-offer distribution.
The worker receives a wage offer each period.
Upon receiving an offer,
the worker decides whether to accept the job or continue search.
After $N$ rejected job offers,
the worker will no longer receive UI benefits.

This sequential-search problem,
where remaining periods of UI compensation is a state variable,
was studied by \citet{burdett_1979}.
The optimal solution is a sequence of reservation wages
that is increasing in the remaining periods of UI compensation.
To maximize their expected income,
a worker 
rejects any job offer with a wage below the reservation wage and
accepts any job offer with a wage above the reservation wage.
\citet{burdett_1979} proves the result using
Bellman equations for employment and unemployment.
In contrast,
I provide a proof that expresses the model in terms of reservation wages,
using techniques described by \citet{rogerson_shimer_wright_2005}.
I believe the optimal policy is more transparent.
In addition, my proof provides an algorithm for directly computing the sequence of reservation wages.

\citet{krueger_mueller_2016}, using a model discussed by \citet{mortensen_1977}, consider a similar environment---where
workers can claim UI compensation for a finite number of weeks---but allow for jobs to end.
A worker who holds a job that does not last 6 months does not qualify for UI compensation.
The added complexity, however, requires them to rely on numerical results.
Unlike the environment I study,
they do not allow for the possibility that benefits are extended.
\citet{boar_mongey_2020} consider a related problem
where workers consider expiring \textit{additional} benefits added by the CARES Act.
They are interested in whether workers will reject return-to-work offers in order to claim the \$600-per-week additional UI benefit.
Benefits expire with a fixed probability, though, so
optimal decisions are not characterized by a sequence of reservation wages, a feature that is essential to the search environment I consider.
Thus,
the results of \citet{krueger_mueller_2016} and \citet{boar_mongey_2020} are complementary to my analysis.


The environment considered in this section is generalized in section \ref{sec:job-search-benefits},
where extensions to UI benefits are considered.
Additional references to the literature on unemployment insurance are shared in appendix \ref{sec:lit-review}.

\subsection{The Sequential Search Environment with Expiring Benefits}

Time is indexed by $t$.
A worker searches for a job in discrete time.
They seek to maximize
\begin{equation}
\label{eq:text:objective}
\E_{0} \sum_{t=0}^{\infty}\beta^{t}x_{t},
\end{equation}
where $\beta \in \left( 0,1 \right)$ is the discount factor such that $\beta = \left( 1+r \right)^{-1}$,
$x_{t}$ is income at time $t$, and
$\E_{0}$ denotes the expectation taken with respect to income given available information at time $0$.
A worker's income is
their wage if employed and
the value of nonwork if unemployed.
The value of nonwork includes
unemployment-insurance compensation,
the value of leisure, and
the value of home production.
Justification
for modeling a risk-neutral worker with no savings includes
the fact that
``a majority of unemployed workers have only a trivial amount of savings''
and
``extending the model to include savings
is an unnecessary complication for a large portion of the unemployed'' \citep[][146--147]{krueger_mueller_2016}.\footnote{In \citeauthor{krueger_mueller_2016}'s \citeyearpar{krueger_mueller_2016}
  survey on reservation wages, which I discuss below,
  over half of respondents with less than 3 months of unemployment duration
  report have no savings.}


A job is fully characterized by the wage the job pays, but
a worker has imperfect information about job possibilities.
The imperfect information is characterized as a distribution of possible wage offers.
While unemployed,
a worker samples one independent and identically distributed offer each period
from a known cumulative distribution function $F$.
I assume $F$ is continuous, has finite expectation, and
has support
$\left[ \underline{w}, \overline{w} \right]$ or
$\left[  \underline{w}, \infty \right)$.
While I refer to $w$ as the wage,
``more generally it could capture some measure of the desirability of the job,
depending on benefits, location, prestige'' \citep[][962]{rogerson_shimer_wright_2005}.

If a worker rejects a wage offer, they remain unemployed.
If a worker accepts a wage offer, they keep the job forever.
An accepted job entitles the worker to a wage payment each period.
I let $W \left( w \right)$ denote the payoff from accepting a job offering wage $w$.
The Bellman equation for $W$ satisfies
\begin{equation}
\label{eq:text:W}
W \left(w \right) = w+\beta W\left(w\right) \quad \text{ or } \quad
W\left( w\right) = \frac{w}{1-\beta}.
\end{equation}

While unemployed,
a worker's income includes the value of leisure and home production, $z >0$.
In addition, a worker may be entitled to UI compensation benefits, $c > 0$.
I let $n$ denote the remaining periods of benefits, typically weeks,
where $n \in \left\{ 0,\dots,N \right\}$ and
$N$ is the maximum number of periods a worker is entitled to claim UI,
typically 26 weeks.
While unemployed,
a worker's income is $x = z+c$ if $n > 0$ and $x = z$ if $n = 0$.
For positive $n$,
the payoff of rejecting a wage offer, earning $z + c$, and
sampling a wage offer the following period is
\begin{equation}
\label{eq:text:U-n}
U \left( n \right) = z+b
+ \beta \E \left[ \max\left\{ U \left(n-1\right),W\left(w\right)\right\} \right].
\end{equation}
The expectation is taken with respect to potential wage draws.
When $n = 0$,
the payoff of rejecting a wage offer, earning $z$, and
sampling a wage offer the following period is
\begin{equation}
\label{eq:text:U-0}
U \left( 0 \right) = z + \beta \E \left[ \max\left\{ U \left(0\right), W\left(w\right)\right\} \right].
\end{equation}

The Bellman equations in \eqref{eq:text:W}, \eqref{eq:text:U-n}, and \eqref{eq:text:U-0}
will be used to write the problem in terms of reservation wages in the next section. 

\subsection{Optimal Search}

Optimal search implies a reservation-wage policy.\footnote{\citet{ljungqvist_sargent_2018} provide an informative textbook exposition.
  \citet{rogerson_shimer_wright_2005} show how a similar
  decision-making environment fits into larger models of the macroeconomy that feature search.} 
The payoff of accepting a job, $W \left( w \right) = w / \left( 1-\beta \right)$,
starts from $0$ and strictly increases.
This implies there is a unique reservation wage, $w_{R} \left( n \right)$,
which satisfies $W \left( w_{R} \left( n \right) \right) = U \left( n \right)$ and
depends on the remaining periods of UI compensation.
With $n \in \left\{ 0,\dots,N \right\}$ remaining periods of UI compensation,
under the convention that the worker accepts a job when they are indifferent between the job and unemployment,
any wage offer $w < w_{R} \left( n \right)$ should be rejected and
any wage offer $w \geq w_{R} \left( n \right)$ should be accepted.

Using
$W \left( w \right) = w / \left( 1-\beta \right)$ and
$U \left( w_{R} \left( n \right) \right) = w_{R} \left( n \right) / \left( 1-\beta \right)$
in the expressions for $W$ and $U$ in equations
\eqref{eq:text:W}, \eqref{eq:text:U-n}, and \eqref{eq:text:U-0}
implies, for $n \in  \left\{ 1,\dots,N \right\}$,
\begin{equation}
\label{eq:text:wRn}
w_{R}\left( n \right) = \left(z+c\right) \left(1-\beta\right) + \beta\left\{ \int_{0}^{w_{R}\left(n-1\right)}w_{R}\left(n-1\right)dF\left(w\right)+\int_{w_{R}\left(n-1\right)}^{\overline{w}}wdF\left(w\right)\right\}
\end{equation}
and, for $n = 0$,
\begin{equation}
\label{eq:text:wR0}
w_{R}\left(0\right) = z\left(1-\beta\right)+\beta\left\{ \int_{0}^{w_{R}\left(0\right)}w_{R}\left(0\right)dF\left(w\right)+\int_{w_{R}\left(0\right)}^{\overline{w}}wdF\left(w\right)\right\} .
\end{equation}
Additional details are provided in appendix \ref{sec:app-basic-job-search}.

Looking at \eqref{eq:text:wR0},
the expression can be written as $\mathcal{T} \left( w_{R} \left( 0 \right) \right) = w_{R} \left( 0 \right)$,
where
\begin{align*}
\mathcal{T}\left(x\right) \equiv z\left(1-\beta\right)+\beta\left\{ \int_{0}^{x}xdF\left(w\right)+\int_{x}^{\bar{w}}wdF\left(w\right)\right\}.
\end{align*}
In other words, $w_{R} \left( 0 \right)$ is a fixed point of the function $\mathcal{T}$.

Appendix \ref{sec:proof-basic} establishes that
$\mathcal{T}$ is a self-map and $0< \mathcal{T}^{\prime}\left(x\right) = \beta F \left( x \right) < 1$ for $x \in \left( \underline{w}, \overline{w} \right)$.
Thus $\mathcal{T}$ is a contraction.
The contraction mapping theorem implies $\mathcal{T}$ admits one and only one fixed point
$w_{R}\left(0\right) \in \left[ \underline{w}, \overline{w} \right]$.\footnote{Appendix \ref{sec:proof-basic} considers the case
  where the support of wages is $\left[\underline{w}, \infty \right)$.
  In that case, I establish that $\mathcal{T}$ contracts through direct verification,
  which requires more algebra.}
Given $w_{R} \left( 0 \right)$,
an induction argument implies that the
reservation wages $w_{R} \left( n \right)$ are increasing in $n$.
A worker is more selective when there is a single remaining week of UI benefits than when there are no remaining weeks, $w_{R} \left( 1 \right) > w_{R} \left( 0 \right)$.
And, using an induction argument,
because $\mathcal{T}$ is increasing, $w_{R} \left( n \right) > w_{R} \left( n-1 \right)$ implies $w_{R} \left( n+1 \right) > w_{R} \left( n \right)$.
Proposition \ref{prop:text:basic-incr-reservation-wages} summarizes these results and
appendix \ref{sec:proof-basic} provides the details.

\begin{prop}
\label{prop:text:basic-incr-reservation-wages}
Assume a risk-neutral, infinitely lived
worker searching for a job.
The worker perceives no possibility that benefits are extended.
The worker is entitled to $N$ periods
of UI benefits and regularly receives wage offers
from the known offer distribution $F$ with support $\left[\underline{w},\overline{w}\right]$
and mean $\mu_{w}$.
A single independently and identically distributed offer is received each period.
The solution to the worker's sequential-search problem
is a sequence of reservation wages that increases
in the number of remaining periods of UI compensation benefits:
\begin{equation}
\overline{w}>w_{R}\left(N\right)>\cdots>w_{R}\left(n+1\right)>w_{R}\left(n\right)>\cdots>w_{R}\left(0\right)>\underline{w}.\label{eq:text:optimal-seq-res-wages}
\end{equation}
In state $n$, the worker accepts any wage $w \geq w_{R}\left(n\right)$.
For offer distributions where $\underline{w}<\left(1-\beta\right)z+\beta\mu_{w}$,
the reservation wages exist and are unique.
The support of wages can extend to $\left[0,\infty\right)$.
\end{prop}

When the future means more to workers currently,
then they are more selective throughout their entire unemployment spells.
In other words,
a higher $\beta$ shifts the entire sequence of reservation wages upward.
Workers are also more selective throughout their unemployment spells when $z$ or $c$ is higher.
A higher value of nonwork means they can reject some jobs they wouldn't have otherwise.
Appendix \ref{sec:prop-reserv-wages} provides more details.

\section{Extending Benefits by Discretion}
\label{sec:job-search-benefits}


In many situations
a worker must make decisions about
accepting a job or continuing their search when
there is a chance that UI benefits are extended.
Policymakers extend UI compensation benefits by discretion.
For example,
temporary additional UI benefits were created
by Congress after the Great Recession and COVID-19 pandemic.
Or extended benefits can be triggered automatically,
but this feature of the UI system
requires a worker to forecast a state triggering on and off.
From the worker's perspective,
there is uncertainty over
\begin{enumerate}[(i)]
\item\label{item:whether-ext} whether benefits are extended and 
\item\label{item:length-ext}  the length of the extension.
\end{enumerate}

I investigate theoretically how the
the possibility of extending UI benefits affects job-acceptance decisions.
I allow workers to form a fixed belief about an extension---summarized
by two parameters for items \ref{item:whether-ext} and \ref{item:length-ext}---and
compute their welfare conditional on that belief.
My formulation abstracts from
how beliefs are formed and
how beliefs are updated.
While central to decision-making,
there is no agreement on how to specify beliefs
\citep[][provide a recent take]{caplin_leahy_2019}.
In addition,
my formulation avoids specifying how legislators enact extensions.
These features allow me to develop a formal, tractable dynamic model
that investigates the costs of misperceiving the
probability and length of an extension.

The economic environment is described in section \ref{sec:environment}.
The environment generalizes the environment described in section \ref{sec:job-search-basic}.
Section \ref{sec:welfare} then considers the costs to a worker
who cannot know
the true probability that benefits are extended and
the length of the extension.

\subsection{Environment}
\label{sec:environment}

I analyze optimal sequential search when
UI compensation benefits can possibly be extended.
To make the model tractable,
once benefits are extended,
there is no chance that benefits are extended subsequently.

An extension affects how the state variable $n$ evolves.
When there are $n$ remaining periods a worker is entitled to collect UI compensation and
the worker is deciding to accept or reject a job offer,
the worker understands that
benefits will be extended by $\Delta$ periods with probability $\delta$.
If benefits are extended,
then the following period there are
$n-1 + \Delta$ remaining periods of UI compensation.
If benefits are not extended,
then the following period there are $n-1$ remaining periods of benefits.
From the worker's perspective,
the probability that benefits are extended, $\delta$, and
the length of the extension, $\Delta$, remain constant but potentially unknown. 

Figure \ref{fig:days-left} shows
a particular instance of how the state variable $n$ evolves.
Starting from the left side of the figure,
a worker is at the node with $n$ remaining periods of UI compensation.
Benefits are not extended,
which occurs with probability $1-\delta$, and the worker does not accept a job.
As indicated by the dark-cyan path,
the following period
the worker makes a decision about accepting or rejecting a job offer
when there are $n - 1$ periods of UI compensation and the possibility of extension.
Again, as indicated by the dark-cyan path,
benefits are not extended and the job offer is rejected,
which places the worker at the node labeled $n-2$.
Then benefits are extended,
which occurs with probability $\delta$.
As indicated by the dark-cyan path, 
the following period after the job offer is rejected
the worker makes a decision about accepting or rejecting a job offer
when there are $n - 3 + \Delta$ periods of UI compensation.
There is no longer a chance of extension.


\begin{figure}
\begin{centering}
\includegraphics{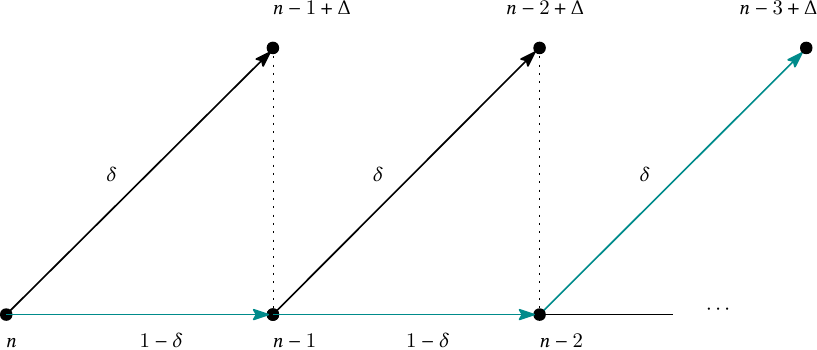}
\par\end{centering}
\caption{\label{fig:days-left} Evolution of remaining periods of UI benefits.}
  \begin{figurenotes}[Note]
    The worker believes benefits will be extended with probability $\delta$ and the extension will add $\Delta$ periods they are allowed to claim UI compensation benefits.
\end{figurenotes} 
\end{figure}

The extension of benefits rules out benefits being extended again.
The one-time occurrence of an extension implies that, upon extension,
the worker solves a standard McCall model
that includes the remaining periods of compensation as a state variable.
In other words, once benefits are extended,
the worker solves the problem described in section \ref{sec:job-search-basic}.

Like in the sequential-search model with finite benefits considered in section \ref{sec:job-search-basic},
preferences are the same as in \eqref{eq:text:objective}.
The value of a job is $W\left(w\right)=w/\left(1-\beta\right)$.
In this environment, I need
to distinguish between the value of unemployment
in the upper and lower halves of figure \ref{fig:days-left}.
In the upper half, there
is no chance of an extension.
The value of unemployment when there is no chance of extension is denoted by $U$.
In the lower half of figure \ref{fig:days-left},
there is a chance---or a perceived chance---that benefits will be extended.
The value of unemployment in this case is denoted by
$U^{\delta}$.
Both $U$ and $U^{\delta}$ depend on the number of remaining periods of UI compensation.

The value of $U$ corresponds to the basic model of job search in section \ref{sec:job-search-basic}.
The value of $U^{\delta}$ requires characterization.
When there is a perceived chance that benefits can be extended,
the value of unemployment is, for  $n \in \left\{ 1,\dots, N \right\}$,
\begin{align}
  \label{eq:text:Udelta-n}
  \begin{split}
  U^{\delta} \left( n \right) &= z + c + \beta \delta \E \left[\max\left\{ U\left(n-1+\Delta\right),W\left(w\right)\right\} \right] \\
  & \quad + \beta \left(1-\delta\right) \E \left[\max\left\{ U^{\delta}\left(n-1\right),W\left(w\right)\right\} \right].
  \end{split}
\end{align}
The first component of the expression on the right side,
$z+c$,
corresponds to the value of nonwork plus the flow UI compensation benefit.
The following period, discounted by $\beta$, corresponds to choosing to accept or reject a job when benefits have or have not been extended,
which occurs with probability $\delta$ and $1-\delta$.
Likewise,
\begin{equation}
\label{eq:text:Udelta-0}  
U^{\delta} \left( 0 \right) = z
+ \beta \left\{  
\delta \E\left[\max\left\{ U\left(\Delta\right),W\left(w\right)\right\} \right] \\
\left(1-\delta\right) \E\left[\max\left\{ U^{\delta}\left(0\right),W\left(w\right)\right\} \right] \right\}.
\end{equation}

When there is a perceived chance that benefits will be extended,
the reservation wage differs from the case when there is no chance.
I need to distinguish between the two cases.
The reservation wage
when there is a chance of extension is denoted by $w_{R}^{\delta}$ and
the reservation wage
upon an extension is denoted by $w_{R}$.
Both depend on the remaining periods of UI benefits.

Using
the techniques in \ref{sec:job-search-basic},
the value of a job, and
the expressions for $U^{\delta}$ in \eqref{eq:text:Udelta-n} and \eqref{eq:text:Udelta-0},
the reservation wages satisfy, for $n \in \left\{ 1,\dots,N \right\}$,
\begin{align}
\label{eq:text:wRdelta}
  \begin{split}
  w_{R}^{\delta} \left( n \right) &= \left(z+c\right) \left(1-\beta\right) + \beta \delta
  \left[ \int_{\underline{w}}^{w_{R}\left(n-1+\Delta\right)} w_{R}\left(n-1+\Delta\right) dF\left(w\right) + \int_{w_{R}\left(n-1+\Delta\right)}^{\overline{w}}wdF\left(w\right) \right] \\
  &\quad +\beta\left(1-\delta\right) \left[ \int_{\underline{w}}^{w_{R}^{\delta}\left(n-1\right)}w_{R}^{\delta}\left(d-1\right)dF\left(w\right) + \int_{w_{R}^{\delta}\left(d-1\right)}^{\overline{w}}wdF\left(w\right) \right].
  \end{split}
\end{align}
and, when $n = 0$,
\begin{align}
  \label{eq:text:wRdelta0}
  \begin{split}
  w_{R}^{\delta}\left(0\right) & = z \left(1-\beta\right) + \beta\delta
                                 \left[ \int_{\underline{w}}^{w_{R}\left(\Delta\right)}w_{R}\left(\Delta\right)dF\left(w\right) + \int_{w_{R}\left(\Delta\right)}^{\overline{w}}wdF\left(w\right) \right] \\
                                 &\quad +\beta \left(1-\delta\right)
                                 \left[ \int_{0}^{w_{R}^{\delta}\left(0\right)}w_{R}^{\delta}\left(0\right)dF\left(w\right) + \int_{w_{R}^{\delta}\left(0\right)}^{\overline{w}}wdF\left(w\right) \right].
  \end{split}
\end{align}

The expressions in \eqref{eq:text:wRdelta} and \eqref{eq:text:wRdelta0}
can be interpreted as
the benefit of search when an offer $w_{R}^{\delta}$ is in hand.
For example,
subtracting
$\left( 1-\beta \right) z$ and $\left( 1-\beta \right) w_{R}^{\delta} \left( 0 \right)$
from both sides of equation \eqref{eq:text:wRdelta0} yields
\begin{align}
\label{eq:search-reservation-interpretation}
  \begin{split}
      w_{R}^{\delta} \left( 0 \right) - z &= \frac{\delta}{r}
                                    \left\{ \int\limits_{\underline{w}}^{w_{R}\left( \Delta \right)} \left[ w_{R}\left( \Delta \right) - w_{R}^{\delta} \left( 0 \right) \right] dF \left( w \right) 
                                    + \int\limits_{w_{R}\left( \Delta \right)}^{\overline{w}} \left[ w^{\prime} - w_{R}^{\delta} \left( 0 \right) \right] dF \left( w^{\prime} \right) \right\} \\
  &\quad + \frac{1-\delta}{r} 
  \left\{ \int\limits_{w_{R}^{\delta} \left( 0 \right)}^{\overline{w}} \left[ w^{\prime\prime} - w_{R}^{\delta} \left( 0 \right) \right] dF \left( w^{\prime\prime} \right) \right\}.
  \end{split}
\end{align}
The left side is the cost of searching another period when offer $w_{R}^{\delta} \left( 0 \right)$ is available.
The right side is the expected benefit.
When benefits are extended,
the worker gains the expected net benefit of adopting reservation wage $w_{R} \left( \Delta \right)$.
Adopting $w_{R} \left( \Delta \right)$ as opposed to $w_{R}^{\delta} \left( 0 \right)$ yields two values.
First,
a worker will reject wage offers below $w_{R} \left( \Delta \right)$, yielding net benefit $w_{R} \left( \Delta \right) - w_{R}^{\delta} \left( 0 \right) > 0$.
(The inequality is established in lemma \ref{lemma-just-extended} in appendix \ref{sec:app-extens}.)
Second,
a worker will accept wage offers $w^{\prime}$ above $w_{R} \left( \Delta \right)$.
When benefits are not extended,
the worker gains the expected value of searching one more time,
which equals
the expected value of drawing $w^{\prime\prime}$ above $w_{R}^{\delta} \left( 0 \right)$.
Jobs are kept indefinitely. 
Their values, like perpetuities that pay off starting from the following period,
equal the flow value divided by the interest rate, $r$, where $\beta = \left( 1-r \right)^{-1}$.\footnote{Expression \eqref{eq:text:wRdelta0} 
  is a generalization of equation (6.3.3) in \citep[][163]{ljungqvist_sargent_2018}.
More details are provided in appendix \ref{sec:app-extens}.}



\subsection{Optimal Search}

Optimal decision rules are characterized by a sequence of reservation wages that are increasing in $n$.
To establish this characterization,
I first use the feature that once benefits are extended the environment coincides with the environment considered in section \ref{sec:job-search-basic}.
In other words,
the upper part of figure \ref{fig:days-left} is well defined and solved.
This allows me to define the function $\mathcal{T}^{\delta}$ on $\left[\underline{w},\overline{w}\right]$ as
\begin{align*}
  \mathcal{T}^{\delta}\left(x\right) & \equiv c\left(1-\beta\right)
                                       + \beta\delta
                                       \left[\int_{0}^{w_{R}\left(\Delta\right)}w_{R}\left(\Delta\right)dF\left(w\right)+\int_{w_{R}\left(\Delta\right)}^{\overline{w}}wdF\left(w\right)\right] \\
 & \quad+\beta\left(1-\delta\right)\left[\int_{0}^{x}xdF\left(w\right)+\int_{x}^{\overline{w}}wdF\left(w\right)\right]. 
\end{align*}
Appendix \ref{sec:proof-pos-ext} establishes that $\mathcal{T}^{\delta}$ is a self map and
$0< \left( \mathcal{T}^{\delta} \right)^{\prime} \left( x \right) = \beta \left( 1-\delta \right) F \left( x \right) <1$
for all $x \in \left(\underline{w},\overline{w} \right)$.
This fact implies that $\mathcal{T}^{\delta}$ is a contraction that admits one and only one solution, $w_{R}^{\delta} \left( 0 \right)$.\footnote{Appendix \ref{sec:proof-pos-ext} also
shows that the support of wages can be extended to $\left[0,\infty\right)$.}

Given $w_{R}^{\delta} \left( 0 \right)$,
an induction argument establishes that optimal sequential search is a sequence of reservation wages that is increasing in $n$.
This is summarized in proposition \ref{prop:1}.

\begin{prop}
  \label{prop:1}
Assume a risk-neutral, infinitely lived
worker searching for a job.
The worker is initially entitled to $N$ periods
of UI benefits and regularly receives wage offers
from the known offer distribution $F$ with support $\left[\underline{w},\overline{w}\right]$
and mean $\mu_{w}$.
A single independently and identically distributed offer is received each period.
In addition,
between each job offer,
there is a chance that benefits are extended by $\Delta$ periods with probability $\delta$.

The solution to the worker's sequential-search problem
is a sequence of reservation wages that increases
in the number of remaining periods of UI compensation benefits:
\begin{equation}
\overline{w}>w_{R}^{\delta}\left(N\right)>\cdots>w_{R}^{\delta}\left(n+1\right)>w_{R}^{\delta}\left(n\right)>\cdots>w_{R}^{\delta}\left(0\right)>\underline{w}.
\end{equation}
In state $n$ when benefits have not been extended,
the worker accepts any wage $w \geq w_{R}^{\delta}\left(n\right)$.
In state $n$ when benefits have been extended,
the worker accepts any wage $w \geq w_{R} \left( n \right)$.
For offer distributions where $\underline{w}<\left(1-\beta\right)z+\beta\mu_{w}$,
the reservation wages exist and are unique.
The support of wages can extend to $\left[0,\infty\right)$.
\end{prop}

\subsection{How Beliefs Affect Optimal Search}
\label{sec:how-beliefs-affect-optimal-search}

As is often the case,
a worker does not know whether benefits will be extended or the extension's length.
A worker's subjective beliefs about a future extension
are summarized by two parameters: $\delta$ and $\Delta$.
The parameter $\delta$ summarizes beliefs about the probability that benefits are extended.
The parameter $\Delta$ summarizes beliefs about the length of an extension.
Both $\delta$ and $\Delta$ affect search behavior in an intuitive way.
Extending benefits through $\delta$ or $\Delta$ provides relief from earning only the flow benefit of nonwork.
Which allows the worker to be more selective about the jobs they are willing to accept.
This idea is summarized in proposition \ref{prop:extend-pr-and-length}.

\begin{prop}
\label{prop:extend-pr-and-length}
Assume a worker searching for a job described in proposition \ref{prop:1}.
Each reservation wage $w_{R}^{\delta} \left( n \right)$ is increasing in $\delta$ and $\Delta$. 
That is,
the worker is more selective when accepting job offers for at least two reasons.
First, 
the worker is more selective
when they think there is a greater chance that benefits will be extended.
Second,
the worker is more selective
when they think benefits will be extended by more periods.
\end{prop}
Appendix \ref{sec:proof-app-search-behavior} provides a detailed proof.

Propositions \ref{prop:text:basic-incr-reservation-wages}, \ref{prop:1}, and \ref{prop:extend-pr-and-length} are illustrated by figure \ref{fig:seq-res-wages},
which shows sequences of reservation wages for different parameter values.
The horizontal axis shows the remaining periods of UI compensation.
As established,
optimal choices imply a worker is less selective as benefits expire---all the sequences are increasing in remaining periods of UI compensation.
The solid, dark-blue line depicts the sequence of reservation wages after an extension has occurred.
Alternatively, because extensions are made only once,
this sequence could be interpreted as the sequence of reservation wages when there is no chance of an extension.

\begin{figure}[htbp]
  \centerline{\includegraphics[width=0.8\textwidth]{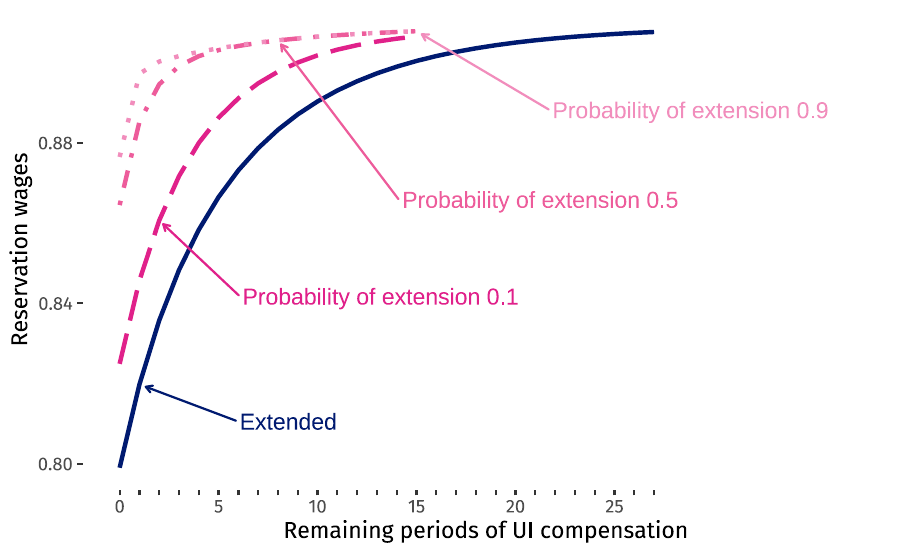}}
\caption[]{\label{fig:seq-res-wages} Sequences of reservation wages.}
  \begin{figurenotes}[Note]
    The dark-blue line depicts the sequence of reservation wages once an extension has occurred.
    These reservation wages lie below the sequences of reservation wages when there is a chance of extension,
    meaning the worker is less selective when there is no chance of extension.
    Sequences were generated using $F \left( w \right) = w$, $c = 0.42$, $z = 0.42$, $N = 15$, $\Delta = 13$, $\beta = 0.95$, and
    $\delta \in \left\{ 0.1, 0.5, 0.9 \right\}$.
\end{figurenotes} 
\end{figure}

The possibility that benefits are extended---or the perception---raises the entire sequence of reservation wages.
These types of upward shifts are depicted in figure \ref{fig:seq-res-wages} by the pink, dashed lines in the upper-left corner.
These sequences are computed for different values of $\delta$.
For example,
when $\delta = 0.1$, an extension is perceived to be unlikely relative to the case where $\delta = 0.9$.
Reservation wages associated with $\delta = 0.9$ lie above
reservation wages associated with $\delta = 0.1$.


Besides illustrating propositions \ref{prop:text:basic-incr-reservation-wages}, \ref{prop:1}, and \ref{prop:extend-pr-and-length},
figure \ref{fig:seq-res-wages} foreshadows the welfare results.
Consider a worker who perceives the probability of extension to be $0.1$ when the true probability of an extension is $0.5$.
In order for the worker to make suboptimal decisions,
they need to receive a wage offer between the long-dashed sequence and the dash--dot sequence in figure \ref{fig:seq-res-wages}.
There isn't much room to make suboptimal decisions.
Initially, at $n = N$,
the worker is rejecting similar offers, as the two lines converge
towards the case where benefits are available indefinitely.
And even at $n=0$,
a suboptimal decision will occur only when a wage offer falls between $0.864$ and $0.825$.
Put another way,
figure \ref{fig:seq-res-wages} suggests that the welfare costs of misperceiving an extension are potentially small.

Nevertheless,
policymakers often extend UI benefits temporarily by discretion.
Workers are forced to make crucial forecasts about $\delta$ and $\Delta$.
The next section considers the welfare consequences of misperceiving the probability and length of an extension.


\section{Welfare}
\label{sec:welfare}

In this section I explore a particular calibration of the model.
The calibration is not meant to be definitive.
Rather,
it is meant to exhibit the features of propositions \ref{prop:1} and \ref{prop:extend-pr-and-length}.




To do so,
I adopt a uniform wage distribution with support $\left[ 0,1 \right]$.\footnote{This assumption offers the opportunity to compare numerical work with closed-form expressions.
  Replication files for the paper compare
    reservation wages computed using closed-form expressions with
    those computed using Monte Carlo integration.
  Closed-form expressions are shared in appendix \ref{sec:closed-form-uniform}.}
For the simulations,
the discount factor, $\beta$, is set to $0.95$.
In the absence of UI compensation and the possibility of an extension,
I find the value of nonwork so that the worker expects to be unemployed for 10 periods.
I then take both the value of nonwork, $z$, and the value of UI compensation, $c$, to be half that value.
I set $N = 10$.
The true probability of extension is set to $0.5$ and
the true length of an extension is set to $25$.
I consider what happens when $\delta$ and $\Delta$ vary from these values.

Figure \ref{fig:welfare:pr-extension} compares the expected welfare of misperceiving the probability that benefits are extended.
I first compute the expected welfare associated with the baseline case where benefits are extended by $\Delta = 25$ periods with probability $0.5$.
To do this computation,
I simulate the model 500 million times and take the average of the computed welfares.\footnote{The simulations suggest that welfare is relatively flat,
  requiring a large number of runs.}
The welfare calculations add up and discount
the flow values of nonwork, any UI benefits, and the value of accepting a job offer.

\begin{figure}[htbp]
  \centerline{\includegraphics[]{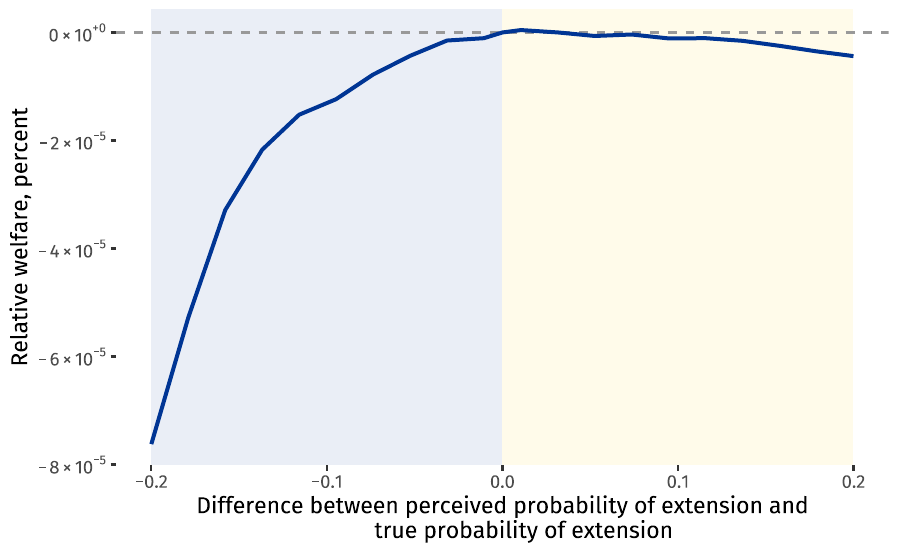}}
\caption[]{\label{fig:welfare:pr-extension} Welfare loss when a worker misperceives the likelihood of extension.}
  \begin{figurenotes}[Note]
    The horizontal axis is the subjective belief a worker holds when computing
    the sequence of reservation wages
    less the true probability that benefits are extended.
    The blue line depicts the loss in welfare associated with misperceiving the probability that benefits are extended.
    When the worker's subjective beliefs line up with the true probability at $0$ along the horizontal axis,
    then the welfare loss is zero.
  The shaded blue   region denotes cases where a worker is too pessimistic about an extension.
  The shaded yellow region denotes cases where a worker is too optimistic  about an extension.
  The true probability of an extension is $0.5$.
\end{figurenotes}
\end{figure}

A worker is then allowed to hold different beliefs about the possibility of extension.
In this scenario, $\Delta$ does not vary.
While the true probability of extension is held at $0.5$, 
the worker believes benefits are extended with probability $\delta$.
This misperception causes the worker to compute a sequence of reservation wages that differs from
the optimal sequence they would compute if they knew the true probability of extension.
When $\delta = 0.1$, for example,
the worker believes the extension is unlikely and they are therefore more likely to accept offers they should reject.
On average, the worker should remain unemployed more often to search for a better wage.
In contrast, when $\delta = 0.9$, 
the worker believes the extension is likely and they are therefore more likely to reject offers they should accept.
On average, the worker should accept more job offers instead of searching for job offers that do not arrive and remaining unemployed.

These losses are depicted in figure \ref{fig:welfare:pr-extension}.
When the worker's belief coincides with the true probability of extension,
which corresponds to $0$ along the horizontal axis,
the welfare loss is also $0$ (theoretically).
This is noted in figure \ref{fig:welfare:pr-extension} by a horizontal gray, dotted line.
As a worker misperceives the probability, they
mistakenly reject offers they would optimally accept and accept jobs they would optimally reject.
The vertical axis shows the relative welfare loss,
computed as the percent away from welfare associated with the true probability that an extension is made.

The loss is asymmetric.
A pessimistic worker, on average,
experiences more loss than an equally optimistic worker.
This feature can be seen by comparing the blue region,
where the perceived probability of extension is less than $0.5$,
to the yellow region,
where the perceived probability of extension is greater than $0.5$.

Yet, as the vertical axis in figure \ref{fig:welfare:pr-extension} indicates, the loss is very small.
What accounts for the small costs of misperception?
As suggested by figure \ref{fig:seq-res-wages},
reservation wages when $\delta = 0.5$ are close to reservation wages when $\delta = 0.9$ or even when $\delta = 0.1$.
Receiving wage offers in the gap between the two sequences is how suboptimal decisions are costly.
But the reservation-wage sequences show that
even when a worker misperceives the probability of an extension by $+0.4$ or $-0.4$,
they are nearly optimally making job-acceptance decisions.
In addition,
most spells of unemployment are short.
Because sequences of reservation wags quickly converge to the reservation wage that would prevail if benefits were paid indefinitely
(seen in figure \ref{fig:seq-res-wages} when the remaining periods of UI compensation equals 15), again,
many job acceptance decisions are made nearly optimally.

A comparison of reservation wages when $\delta = 0.5$ to cases when $\delta = 0.1$ and $\delta = 0.9$
also reveals the source of the asymmetry.
When $\delta = 0.9$, reservation wages are nearly indistinguishable in figure \ref{fig:seq-res-wages} to the case where $\delta = 0.5$.
In contrast,
there is a distinguishable gap when $\delta = 0.1$.
To be clear, though, welfare loss is small whether a worker is optimistic or pessimistic about an extension.  




The same pattern holds when workers misperceive the length of extension,
which can be seen by comparing figure \ref{fig:welfare:pr-extension} to figure \ref{fig:welfare:length-extension}.
To create figure \ref{fig:welfare:length-extension},
I first compute the expected welfare associated with a baseline.
In the baseline,
compensation benefits are extended with probability $\delta = 0.5$ for $25$ periods.
While the true length of an extension is held at $25$,
the worker holds a different belief about the length of extension,
which they use to compute a sequence of reservation wages.
Based on this sequence,
the worker accepts and rejects offers---suboptimally,
as the worker would have computed a different sequence of reservation wages to maximize their expected present value of welfare if
they knew the true length of extension.
In this scenario, $\delta$ does not vary.

\begin{figure}[htbp]
  \centerline{\includegraphics[]{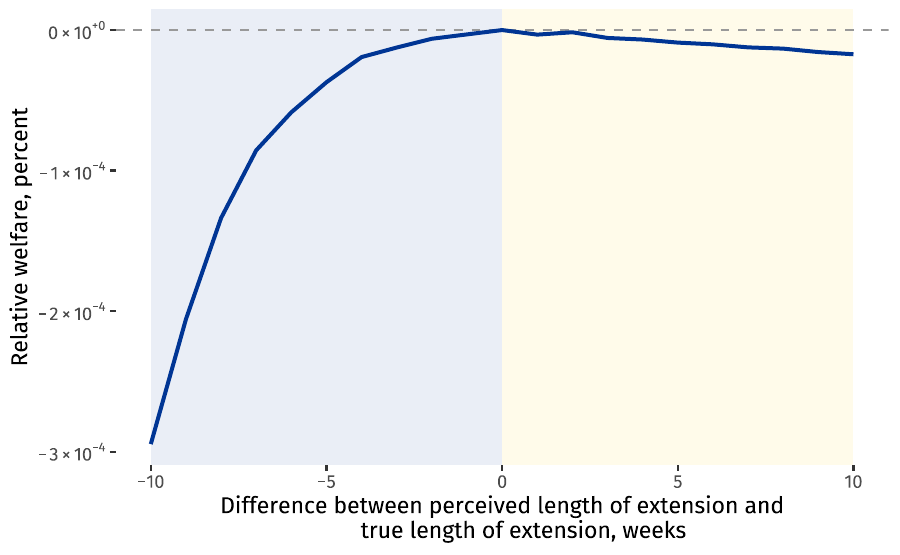}}
\caption[]{\label{fig:welfare:length-extension} Welfare loss when a worker misperceives the length of extension.}
\begin{figurenotes}[Note]
  The horizontal axis is the subjective belief about the length of an extension a worker holds
  when computing the sequence of reservation wages
  less the true length of an extension.
  The blue line depicts the loss in welfare.
  When the worker's subjective beliefs line up with the true length of an extension at $0$ along the horizontal axis,
  then the welfare loss is zero.
  The region shaded blue   denotes cases where a worker is too pessimistic about the length of the extension.
  The region shaded yellow denotes cases where a worker is too optimistic  about the length of the extension.
  The true length of an extension is 25.
\end{figurenotes}
\end{figure}


The exercise is repeated for different beliefs over the length of the extension.
Relative losses are depicted in figure \ref{fig:welfare:length-extension}.
Welfare losses are asymmetric and small in magnitude.
The magnitudes do suggest, however,
that misperceptions about the length of extension may lead to
greater welfare loss than misperceptions about the probability of an extension.


Figures \ref{fig:welfare:statistics:pr} and \ref{fig:welfare:statistics:length} illustrate the mechanisms through which welfare is lost.
Figure \ref{fig:welfare:statistics:pr} illustrates why welfare is lost when a worker misperceives the probability of an extension.
In the gray region of both panels,
looking at the horizontal axes reveals that in these cases the worker perceives that an extension is unlikely relative to the truth.
The top panel shows that, on average, workers who are too pessimistic about an extension spend
too little time unemployed and searching.
If they knew the true probability of extension was higher,
then they would reject wage offers they had accepted based on a worse forecast.
This dynamic is shown in the bottom panel of the figure,
which shows the relative accepted wage.
In contrast, in the white region,
workers are too optimistic about an extension.
They spend too much time searching for a high wage, which costs them through experienced unemployment.

\begin{figure}[htbp]
  \centerline{\includegraphics[width=0.8\textwidth]{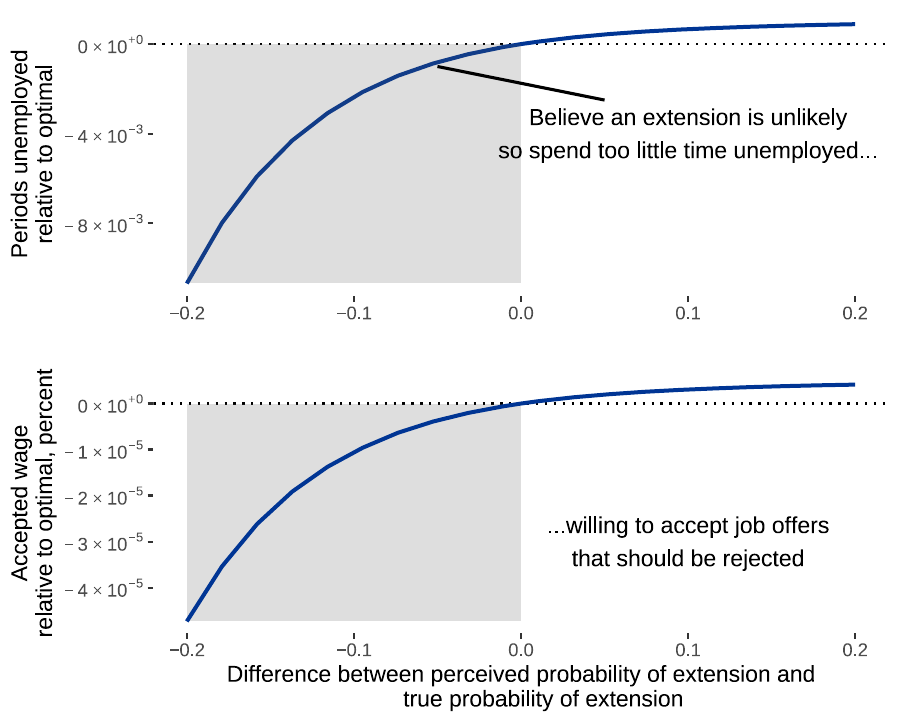}}  
\caption[]{\label{fig:welfare:statistics:pr} Simulated statistics for a worker's choices about job offers
  when the worker misperceives the true chance of an extension.}
\begin{figurenotes}[Note]
  The horizontal axis reports the subjective belief a worker holds about the chance of an extension
  when computing the sequence of reservation wages
  less the true probability that benefits are extended.
  The unit of measurement along the horizontal axis is probability.
  The shaded regions highlight cases where a worker believes an extension is unlikely.
  The top panel shows the average number of periods a worker spends unemployed relative to the
  number of periods that would be spent in unemployment if the true probability of an extension were known.
  The bottom panel shows the average accepted wage, reported as the percent away from the optimal wage.
\end{figurenotes}
\end{figure}

Both panels of figure \ref{fig:welfare:statistics:pr} show that
misperceptions lead to suboptimal job-acceptance decisions, but
these decisions are nearly optimal.
The statistics on time spent unemployed and accepted wages corroborate the welfare losses depicted in figure \ref{fig:welfare:pr-extension}.
Figure \ref{fig:welfare:statistics:length} illustrates that the same mechanisms reduce welfare when a worker misperceives the length of an extension.
Overly optimistic  workers spend too much   time unemployed  searching for a high wage.
Overly pessimistic workers spend too little time unemployed, believing their benefits will soon run out.

\begin{figure}[htbp]
\centerline{\includegraphics[width=0.8\textwidth]{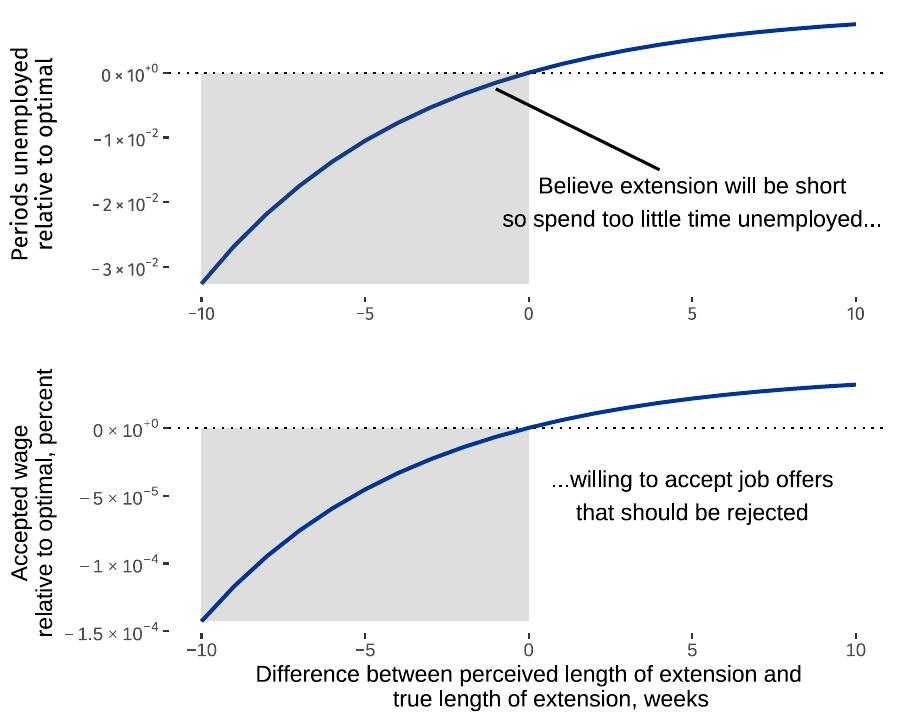}}
\caption[]{\label{fig:welfare:statistics:length} Simulated statistics for a worker's choices about job offers
  when the worker misperceives the length of an extension.}
\begin{figurenotes}[Note]
  The horizontal axis reports the subjective belief a worker holds about the length of an extension
  when computing the sequence of reservation wages
  less the true extension length.
  The unit of measurement along the horizontal axis is periods or weeks.
  The shaded regions highlight cases where a worker believes an extension will be shorter than what occurs.
  The top panel shows the average number of periods a worker spends unemployed relative to the
  number of periods that would be spent in unemployment if the true length of extension were known.
  The bottom panel shows the average accepted wage, reported as the percent away from the optimal wage.
\end{figurenotes}
\end{figure}

\section{Discussion}
\label{sec:discussion}

Propositions \ref{prop:1} and \ref{prop:extend-pr-and-length} link reservation wages to UI compensation benefits:
reservation wages decline as UI compensation expires.
This theoretical result is mildly consistent with
survey evidence presented by \citet{krueger_mueller_2016}.
They asked 6,025 respondents in New Jersey who were unemployed as of September 28, 2009:
``Suppose someone offered you a job today.
What is the lowest wage or salary you would accept (before deductions)
for the type of work you are looking for?''
Respondents' answers across 39,201 interviews indicate that
there is 
``little tendency for the reservation wage to decline over the spell of unemployment''
\citep[][158 and figure 4]{krueger_mueller_2016}.

What could account for a flat sequence of reservation wages?
My results suggest that beliefs about extensions may be an important factor.
The New Jersey respondents were unemployed as of September 28, 2009,
so many were confronted with the possibility that benefits would be extended and they could claim up to 99 weeks of UI compensation.
Many, in other words, faced a situation like Vernon's discussed in section \ref{sec:ext-thought-experiment}.
Two years later, in September 2011, the unemployment rate was around 9 percent.
A worker could interpret these data as evidence that additional extensions would be made available.
In addition,
the model shows that reservation wages converge somewhat quickly to the level where UI compensation will be available indefinitely.
Optimism about an extension and quick convergence both suggest
a rule of thumb that computes reservation wages that are close to the reservation wage in the case where benefits are available indefinitely.
If workers adopt this rule of thumb on average, then
there would indeed be ``little tendency for the reservation wage to decline over the spell of unemployment.''\footnote{Less is known about
  how psychology influences job search.
  Two examples are provided by \citet{dellavigna_paserman_2005} and \citet{paserman_2008},
  who study how impatience affects search.}


Reservation wages may well be more informative than
what can be learned in a partial-equilibrium setting like the one considered here.
\citet{shimer_werning_2007} point out
that a reservation wage makes a worker indifferent between work and nonwork (as in $W \left( w_{R} \right) = U$).
In addition, a worker's take-home pay is directly related to consumption and
therefore utility.
Thus,
the reservation wage (benefits do not expire in their model) measures the utility of unemployed workers.
Any policy that raises the single reservation wage will increase workers' well-being.
In particular,
if a marginal increase in UI compensation raises a worker's reservation wage,
then the change improves welfare.\footnote{In related work,
  \citet{shimer_werning_2008}
  establish that optimal UI compensation
  involves constant or nearly constant benefits and thus a single reservation wage.
  Both papers informatively consider cases where workers save.}
\citeauthor{krueger_mueller_2016}'s \citeyearpar{krueger_mueller_2016}
finding that reservation wages do not respond to benefits
suggests that increasing UI benefits would not be optimal.
But, as pointed out here,
that conclusion may be complicated by the uncertainty that surrounded
extended benefits after the Great Recession.
This perspective raises a number of questions for future research.
In particular: What beliefs do people hold about extensions?





\section{Conclusion}
\label{sec:conclusion}

My goal was to identify a worker's costs of misperceiving
the probability and length of an extension to UI benefits.
As discussed in section \ref{sec:descr-decis-making-ext}, 
navigating possible extensions
was an unavoidable feature of job search after
the two most recent recessions.
The model I presented in section \ref{sec:job-search-benefits}
establishes channels through which a worker's welfare can suffer.
When policymakers like congressional representatives
assure constituents that a significant extension is coming \textit{without a doubt},
then workers will believe an extension is likely and
the length of the extension will provide meaningful support.
Workers will adjust their reservation wages upward.
If the extension does not come,
then workers will have rejected offers they would have liked to accept.
On average, workers will experience too much unemployment.
The mirror channel operates
when a policymaker communicates their disapproval of UI benefits,
despite wider support among policymakers that an extension will be granted.
The numerical example, however,
found small costs to misperceiving the probability and length of an extension.
While the calibration is not definitive,
the qualitative properties of the model may help explain some limited data on reservation wages.

Going beyond this paper, 
the numerical example suggests that extending benefits when the unluckiest and most unfortunate need them to avoid misery
will alleviate some worst cases.
\citet[][208]{kahn_2011} points out that searching for work
in economic slumps is ``particularly damaging'' and
notes that the damage extends long into careers.
Small dithering costs borne by workers may be worth undertaking in order for
policymakers to extend benefits for a preferred length and amount.
Clear communication will help.
Going forward,
extensions may be relied upon more frequently as policymakers are forced to manage not only
business cycles but also natural events like
viruses, wildfires, extreme heat, and hurricanes.

\bibliographystyle{../../../bibliography/bostonfed}
\bibliography{../../../bibliography/bibliography-org-ref}

\clearpage
\pagebreak







\appendix
\appendixpage

\section{Some Additional Background on the UI System}
\label{sec:lit-review}

For most people, searching for a job is an inescapable part of life.
Its importance is reflected in a voluminous literature on the design of UI systems.
While a complete accounting is beyond the scope of this paper,
I nevertheless attempt to briefly orient the paper.



\citeauthor{shimer_werning_2007}'s \citeyearpar{shimer_werning_2007} work, mentioned in section \ref{sec:discussion}, 
derives a sufficient statistic to check optimality of the studied UI system.
For an overview of the sufficient-statistic approach, see \citet{chetty_2009} and \citet{chetty_finkelstein_2013}.
Prominent examples of this approach in the context of UI compensation are 
\citet{chetty_2006} and \citet{chetty_2008}.
\citet{landais_spinnewijn_2021} offer a related approach
to valuing UI that links structural and reduced-form models.
\citeauthor{shimer_werning_2007}'s \citeyearpar{shimer_werning_2007} analysis is primarily concerned with the level of UI benefit.
Along this train of thought,
each feature of the UI system could possibly be adjusted to improve well-being.
And there are many features to consider as the overviews by
\citet{krueger_meyer_2002}, \citet{chetty_finkelstein_2013}, and
\citet{schmieder_von-wachter_2016} attest.
The remainder of this brief section highlights different features of the UI system.

One feature is the effect UI compensation has on labor supply.
While no UI compensation seems cruel,
enough to fund lavish homes and fancy vacations may halt an economy.
Between these extremes is a suitable level that is much debated.
Part of this debate is based on labor supply around
added periods an unemployed worker can claim UI compensation
\citep{solon_1979,card_levine_2000,lalive_zweimuller_2004,card_chetty_weber_2007,lalive_2007,card_etal_2015,chodorow-reich_coglianese_karabarbounis_2019,le-barbanchon_rathelot_roulet_2019,dieterle_bartalotti_brummet_2020,marinescu_skandalis_2021}.
Labor supply is perennially held up as an answer to why a labor market may be weak,
including after the Great Recession and the COVID-19 pandemic.
Labor supply after the Great Recession is studied by
\citet{rothstein_2011},
\citet{farber_rothstein_valletta_2015},
\citet{farber_valletta_2015},
\citet{hagedorn_manovskii_mitman_2016},
\citet{johnston_mas_2018},
\citet{hagedorn_etal_2019}, and
\citet{boone_etal_2021}.
Labor supply after the COVID-19 pandemic is studied by
\citet{boar_mongey_2020},
\citet{holzer_hubbard_strain_2021},
\citet{albert_etal_2022},
\citet{petrosky-nadeau_2020},
\citet{petrosky-nadeau_valletta_2021}, and
\citet{faberman_mueller_sahin_2022}.
These are not exhaustive lists in any way.
Rather,
the listed work and the references they cite highlight the importance of the issue.

Extended UI compensation may keep workers in the labor force.
Instead of giving up their search discouraged,
a worker may continue their search to claim UI compensation and eventually find work.
This point has been made at least since \citet{solon_1979} made it.
More recently,
\citet{petrosky-nadeau_valletta_2021} consider labor-force transitions during the COVID-19 pandemic.
There is also the possibility that extended benefits
reduce claims for Social Security Disability Insurance \citep{rutledge_2011}.\footnote{\citet{autor_duggan_2003} and \citet{khemka_roberts_higgins_2017} provide background.}

In a larger context,
extended benefits affect the flow value of nonwork,
which is an important parameter in macroeconomic models.
This feature of the UI system
in the Diamond--Mortensen--Pissarides class of models is a worker's threat point
in Nash wage bargaining
\citep{chodorow-reich_karabarbounis_2016, ljungqvist_sargent_2017, jaeger_etal_2020}.
Likewise,
access to extended benefits influence efficiency-wage models,
where workers are paid above-market-clearing rates \citep{yellen_1984}.
And affect notions of fairness, which matter in practice \citep{akerlof_yellen_1990,bewley_1999}.

The flow value of nonwork in the context of general-equilibrium models of the macroeconomy
moves the focus from microeconomic choices to macroeconomic effects.
There are two objectives of the UI system:
``provide temporary and partial wage replacement to involuntarily unemployed workers and to stabilize the economy during recessions''
\citep{whittaker_isaacs_2022}.
UI benefits automatically stabilize an economy.
In periods of rising unemployment,
weekly benefit payments increase and collected payroll taxes decline.
The aggregate-demand effects of UI benefits
are considered by \citet{mitman_rabinovich_2015}, \citet{kroft_notowidigdo_2016} and \citet{landais_michaillat_saez_2018}.
Empirical evidence for these effects is provided by
\citet{marinescu_2017}  and
\citet{marinescu_skandalis_zhao_2021}.
\citet{kekre_2023} considers extended benefits in a dynamic, stochastic, general-equilibrium model that features incomplete markets.

Extended benefits may also interact with 
several other prominent features that are not in the present model.
For example,
it is easy to imagine that wishful thinking affects reservation wages \citep{caplin_leahy_2019}.
Or that extensions influence learning about wage offers \citep{burdett_vishwanath_1988}.
\citet{acemoglu_2001} and \citet{nekoei_weber_2017} 
investigate benefits and workers' subsequent wages,
an indicator of match quality.
Extensions could also affect selection issues as documented by \citet{hendren_2017},
which has implications for private UI markets.
\citet{anquandah_bogachev_2019}, for example, consider related UI pricing issues.

In summary,
the theoretical perspective I have presented raises a number of
questions about how perceptions about the probability and length of an extension
interact with features of the UI system and what optimal features look like.

\section{Canonical Job Search with Finite UI Benefits}
\label{sec:app-basic-job-search}

In this section,
I describe a canonical McCall model with expiring benefits.
I provide more details than the presentation in the main text in section \ref{sec:job-search-basic}.

In this decision-making environment,
a worker searches for a job in discrete time.
The worker seeks to maximize
\begin{equation}
\E\sum_{t=0}^{\infty}\beta^{t}x_{t}\label{eq:objective}
\end{equation}
where
$\beta\in\left(0,1\right)$ is the discount factor,
$x_{t}$ is income at time $t$, and
$\E$ denotes the expectation.
The worker is interested in decision rules that indicate whether to accept or reject job offers.
As established below,
optimal decision-making is characterized by a sequence of reservation wages.

The firm side of the economy is modeled as a collection of firms that offer productive opportunities to workers summarized by wage offers.
Wages are offered in the range $\left[\underline{w},\overline{w}\right]$ or $\left[0,\infty\right)$.
The other side of the market are workers.
A worker receives one independent and identically distributed offer each period.
In addition,
workers understand that wages are offered according to a given law, $w \sim F\left(w\right)$.

A job is kept forever.
The Bellman equation for the value of an accepted job is:
\begin{align}
W\left(w\right) & =w+\beta W\left(w\right)\quad\therefore\quad W\left(w\right)=\frac{w}{1-\beta},\label{eq:W}
\end{align}
which is repeated in equation \eqref{eq:text:W}.

Let $n$ represent the number of remaining periods of UI compensation.
This is a state variable.
The Bellman equation for the value of unemployment is
\begin{equation}
U\left(n\right)=z+c+\beta \E\left[\max\left\{ U\left(n-1\right),W\left(w\right)\right\} \right]\text{ for }n\in\left\{ 1,\dots,N\right\}, \label{eq:Un}
\end{equation}
which is the same as equation \eqref{eq:text:U-n}.
The expression $U \left( n \right)$ can be expanded as
\begin{align*}
U\left(n\right) & =z+c+\beta\int_{\underline{w}}^{\overline{w}}\max\left\{ U\left(n-1\right),W\left(w\right)\right\} dF\left(w\right)\\
 & =z+c+\beta\int_{\underline{w}}^{\overline{w}}\max\left\{ \frac{w_{R}\left(n-1\right)}{1-\beta},\frac{w}{1-\beta}\right\} dF\left(w\right)\\
\therefore\frac{w_{R}\left(n\right)}{1-\beta} & =z+c+\beta\int_{\underline{w}}^{\overline{w}}\max\left\{ \frac{w_{R}\left(n-1\right)}{1-\beta},\frac{w}{1-\beta}\right\} dF\left(w\right)\\
\therefore w_{R}\left(n\right) & =\left(z+c\right)\left(1-\beta\right)+\beta\int_{\underline{w}}^{\overline{w}}\max\left\{ w_{R}\left(n-1\right),w\right\} dF\left(w\right).
\end{align*}
Crucially, this means, for $n \in \left\{ 1,\dots,N \right\}$,
\begin{equation}
w_{R}\left(n\right)=\left(z+c\right)\left(1-\beta\right)+\beta\left\{ \int_{\underline{w}}^{w_{R}\left(n-1\right)}w_{R}\left(n-1\right)dF\left(w\right)+\int_{w_{R}\left(n-1\right)}^{\overline{w}}wdF\left(w\right)\right\}, \label{eq:wRn}
\end{equation}
which agrees with equation \eqref{eq:text:wRn}.
When UI benefits have expired,
the problem becomes 
\begin{equation}
U\left(0\right)=z+\beta \E\left[\max\left\{ U\left(0\right),W\left(w\right)\right\} \right],\label{eq:U0}
\end{equation}
where the flow value of unemployment is only $z$ instead of $z+c$.
The same procedure as above yields
\begin{equation}
  \label{eq:wR0}
w_{R}\left(0\right)=z\left(1-\beta\right)+\beta\left\{ \int_{\underline{w}}^{w_{R}\left(0\right)}w_{R}\left(0\right)dF\left(w\right)+\int_{w_{R}\left(0\right)}^{\overline{w}}wdF\left(w\right)\right\},
\end{equation}
which agrees with equation \eqref{eq:text:wR0} in the text.

The remainder of this section establishes that the reservation-wage policy is a sequence of reservation wages increasing in $n$.
This is summarized in proposition \ref{prop:text:basic-incr-reservation-wages} in the main text.

\section{Proof of Proposition \ref{prop:text:basic-incr-reservation-wages} in the Main Text}
\label{sec:proof-basic}

The worker solves 
\begin{equation}
V\left(w,n\right) = \max_{\text{accept, reject}}\left\{ W\left(w\right),U\left(n\right)\right\} .\label{eq:value-fnc}
\end{equation}
\citet{burdett_1979} uses the value function $V$ to establish the
result in \eqref{eq:text:optimal-seq-res-wages}; namely,
that $w_{R} \left( n \right) > w_{R} \left( n-1 \right)$.\footnote{Using the value function in this way is also achieved by \citet{ross_1983}.}
Another approach, which is taken here,
states the problem in terms of reservation wages, which may be more transparent.
In addition, characterizing the problem this way provides an algorithm for computing the sequence of reservation wages.


The proof proposition \ref{prop:text:basic-incr-reservation-wages} in the main text
relies on 2 lemmas.
The 2 lemmas are first stated and proved before the details of the proof of proposition \ref{prop:text:basic-incr-reservation-wages} in the main text are given.

\subsection{Proof of Lemmas \ref{lemma:upsilon} and  \ref{lemma:Psi}}

\begin{lemma}
\label{lemma:upsilon}
Define
\begin{align*}
\Upsilon\left(x\right)\equiv\int_{\underline{w}}^{x}xdF\left(w\right)+\int_{x}^{\overline{w}}wdF\left(w\right),
\end{align*}
for $x \in \left[ \underline{w}, \overline{w} \right]$.
Then $\Upsilon$ is increasing on the interval $\left(0,\overline{w}\right)$.
A minor modification allows me to replace $\overline{w}$ with $\infty$ and define $\Upsilon$ on $\left[\underline{w}, \infty \right)$.
In addition,
\begin{align*}
\Upsilon \left( x \right) < \overline{w}.
\end{align*}
\end{lemma}

\begin{proof}
Differentiation of $\Upsilon$ yields
\begin{align*}
\Upsilon^{\prime}\left(x\right) & =1xf\left(x\right)+\int_{\underline{w}}^{x}1dF\left(w\right)-xf\left(x\right)=\int_{\underline{w}}^{x}dF\left(w\right)\\
 & =F\left(x\right)>0,
\end{align*}
establishing the first result.
The second result can be established by writing $\Upsilon$ as
\begin{align*}
\Upsilon \left( x \right) &= x + \int\limits_{x}^{\overline{w}} \left( w-x \right) d F \left( w \right) \\
  &= x + \overline{w} - x - \int\limits_{x}^{\overline{w}} F \left( w \right) dw \\
  &= \overline{w} - \int\limits_{x}^{\overline{w}} F \left( w \right) dw,
\end{align*}
where the first line adds and subtracts $\int_{x}^{\overline{w}} x dF \left( w \right)$ and
the second line uses integration by parts.
These derivations are carried out in \eqref{eq:integration-by-parts-T} and \eqref{eq:integration-by-parts-T-2} below.
\end{proof}

\begin{lemma}
  \label{lemma:Psi}
  The function $\Psi\left(x\right)\equiv x-\beta\left(1-\delta\right)\Upsilon\left(x\right)$ is increasing in $x$:
\begin{align*}
\frac{\partial\Psi}{\partial x} & =1-\beta\left(1-\delta\right)\Upsilon^{\prime}\left(x\right)\\
 & =1-\beta\left(1-\delta\right)F\left(x\right)\\
 & >0,
\end{align*}
where the inequality uses lemma \ref{lemma:upsilon}.
\end{lemma}

\subsection{Details of the Proof of Proposition \ref{prop:text:basic-incr-reservation-wages} in the Main Text}

Proposition \ref{prop:text:basic-incr-reservation-wages} in the main text is established in \ref{item:basic:range} steps:
\begin{enumerate}
\item\label{item:basic:reservation-wages} The existence of reservation wages is established,
  which justifies writing the problem in terms of reservation wages.
\item\label{item:basic:contraction} Existence and uniqueness of $w_{R} \left( 0 \right)$ is established.
  This is done by an appeal to the contraction mapping theorem.
\item\label{item:basic:sequence} The sequence of reservation wages,
  $w_{R}\left(1\right),\dots,w_{R}\left(N\right)$,
  are then computed, starting from $w_{R} \left( 0 \right)$.
  It is established that the worker is less selective when there are fewer remaining periods of UI benefits,
  which is expressed in equation \eqref{eq:text:optimal-seq-res-wages} of proposition \ref{prop:text:basic-incr-reservation-wages}. 
\item\label{item:basic:range} The last step establishes that $w_{R} \left( N \right) < \overline{w}$ and $w_{R}\left( 0 \right) > \underline{w}$.
\end{enumerate}

\ul{\mbox{Step \ref{item:basic:reservation-wages}}}: A reservation wage characterizes the worker's optimal choice.
From \eqref{eq:W},
the payoff of accepting a job is $W\left(w\right)=w/\left(1-\beta\right)$.
The set of possible payoffs starts from $\underline{w}$ and goes
until $\overline{w}/\left(1-\beta\right)$.
The value $U\left( n \right)$ is constant.
In addition, the most the payoff of unemployment can be is
\begin{align*}
U\left( n \right) & \leq z+\beta W\left(\overline{w}\right) \\
 &= z+\beta\frac{\overline{w}}{1-\beta}
\end{align*}
as
\begin{align*}
z+\beta\frac{\overline{w}}{1-\beta} & <\frac{\overline{w}}{1-\beta}=W\left(\overline{w}\right)\\
\therefore z+\beta\frac{\overline{w}}{1-\beta} & <\frac{\overline{w}}{1-\beta}\\
\therefore z & <\frac{\overline{w}}{1-\beta}\left(1-\beta\right)\\
\therefore z & <\overline{w},
\end{align*}
which is true by assumption.
Likewise, from \eqref{eq:U0},
\begin{align*}
U\left(n\right) & \geq z+\beta\frac{\mu_{w}}{1-\beta}>\frac{\underline{w}}{1-\beta}
\end{align*}
as
\begin{align*}
z+\beta\frac{\mu_{w}}{1-\beta} & >\frac{\underline{w}}{1-\beta}\\
\therefore\left(1-\beta\right)z+\beta\mu_{w} & >\underline{w}
\end{align*}
by assumption.
Therefore, because $W$ is strictly increasing in $w$,
there exists a unique $w_{R}\left(n\right)$ such that
$W\left(w_{R}\left(n\right)\right)=U\left(n\right)$.
This establishes the existence of the reservation wages $w_{R} \left( n \right)$
for $n \in \left\{ 0,\dots,N \right\}$.

\ul{\mbox{Step \ref{item:basic:contraction}}}: Existence of the reservation wages offers an alternative characterization of the problem.
Starting from \eqref{eq:wR0},
define the function $\mathcal{T}$ as
\begin{equation}
\mathcal{T}\left(x\right)=z\left(1-\beta\right)+\beta\left\{ \int_{\underline{w}}^{x}xdF\left(w\right)+\int_{x}^{\overline{w}}wdF\left(w\right)\right\} .\label{eq:map-wR}
\end{equation}
The function $\mathcal{T}$ is defined on $\left[\underline{w},\overline{w}\right]$.
I am interested in $\mathcal{T} \left( w_{R} \left( 0 \right) \right) = w_{R} \left( 0 \right)$.

I first prove that $\mathcal{T}$ is a self map;
that is,
$\mathcal{T}:\left[\underline{w},\overline{w}\right]\rightarrow\left[\underline{w},\overline{w}\right]$.
First,
$\mathcal{T}\left(\underline{w}\right)=z\left(1-\beta\right)+\beta\mu_{w}>\underline{w}$
by assumption.
Second,
$\underline{w}<\mathcal{T}\left(\overline{w}\right)=z\left(1-\beta\right)+\beta\overline{w}<\overline{w}$.
Lastly, $\mathcal{T}$ is strictly increasing on $\left(\underline{w},\overline{w}\right)$:
\begin{align*}
\mathcal{T}^{\prime}\left(x\right) & =\beta\left\{ \underline{w}f\left(\underline{w}\right)+\int_{\underline{w}}^{x}1dF\left(w\right)-\underline{w}f\left(\underline{w}\right)\right\} \\
 & =\beta F\left(x\right).
\end{align*}
where the computation uses Leibniz's rule.
Hence, $\mathcal{T}:\left[\underline{w},\overline{w}\right]\rightarrow\left[\underline{w},\overline{w}\right]$. 

In addition,
because the derivative of $\mathcal{T}$ is
$\mathcal{T}^{\prime}\left(x\right)=\beta F\left(x\right)$,
it is true that $0<\mathcal{T}^{\prime}\left(x\right)<1$ for all
$x\in\left(\underline{w},\overline{w}\right)$.
Using the usual metric,
$d\left(w_{1},w_{2}\right)=\left|w_{1}-w_{2}\right|$,
$\mathcal{T}$ is a contraction on $\left[\underline{w},\overline{w}\right]$ \citep[58, theorem 4.2]{bryant_1985}.
Because $\left[\underline{w},\overline{w}\right]$ is complete,
the contraction mapping theorem establishes that $\mathcal{T}$ admits
one and only one fixed point $w_{R}\left(0\right)\in\left[\underline{w},\overline{w}\right]$
\citep[191, theorem 6.7]{acemoglu_2009}.

If instead
$\mathcal{T}:\left[\underline{w},\infty\right)\rightarrow\left[\underline{w},\infty\right)$,
which is the case for many wage-offer distributions,
then the same steps as above verify that $\mathcal{T}$ is a self-map.
In addition,
it can be verified that $\mathcal{T}$ is a contraction directly.

To do so, it will help to use an equivalent expression for $\mathcal{T}$.
I note that
\begin{align}
  \label{eq:integration-by-parts-T}
  \begin{split}
\int_{\underline{w}}^{x}xdF\left(w\right)+\int_{x}^{\overline{w}}wdF\left(w\right) &= \int_{\underline{w}}^{x}xdF\left(w\right)+\int_{x}^{\overline{w}}xdF\left(w\right) \\
 &\quad -\int_{x}^{\overline{w}}xdF\left(w\right)+\int_{x}^{\overline{w}}wdF\left(w\right)\\
 &= \int_{\underline{w}}^{\overline{w}}xdF\left(w\right)+\int_{x}^{\overline{w}}\left(w-x\right)dF\left(w\right)\\
 &= x\int_{\underline{w}}^{\overline{w}}dF\left(w\right)+\int_{x}^{\overline{w}}\left(w-x\right)dF\left(w\right)\\
 &= x+\int_{x}^{\overline{w}}\left(w-x\right)dF\left(w\right).    
  \end{split}
\end{align}
And integration by parts implies the second term can be written
\begin{align}
  \label{eq:integration-by-parts-T-2}
  \begin{split}
\int_{x}^{\overline{w}}\left(w-x\right)dF\left(w\right) & =\int_{x}^{\overline{w}}\left(w-x\right)F^{\prime}\left(w\right)dw\\
 & =\left[\left(w-x\right)F\left(w\right)\right]_{w=x}^{w=\overline{w}}-\int_{x}^{\overline{w}}F\left(w\right)dw\\
 & =\left(\overline{w}-x\right)F\left(\overline{w}\right)-\int_{x}^{\overline{w}}F\left(w\right)dw.    
  \end{split}
\end{align}
Using the fact that
\begin{align*}
\int_{x}^{\overline{w}}\left[F\left(\overline{w}\right)-F\left(w\right)\right]dw & =\int_{x}^{\overline{w}}F\left(\overline{w}\right)dw-\int_{x}^{\overline{w}}F\left(w\right)dw\\
 & =F\left(\overline{w}\right)\int_{x}^{\overline{w}}dw-\int_{w_{i}}^{\overline{w}}F\left(w\right)dw\\
 & =F\left(\overline{w}\right)\left(\overline{w}-x\right)-\int_{x}^{\overline{w}}F\left(w\right)dw
\end{align*}
the latter expression can be written
\begin{align*}
\int_{x}^{\overline{w}}\left(w-x\right)dF\left(w\right)=\int_{x}^{\overline{w}}\left[F\left(\overline{w}\right)-F\left(w\right)\right]dw.
\end{align*}
Taking the limit yields
\begin{align*}
\lim_{\overline{w}\rightarrow\infty}\int_{x}^{\overline{w}}\left(w-x\right)dF\left(w\right) & =\lim_{\overline{w}\rightarrow\infty}\int_{x}^{\overline{w}}\left[F\left(\overline{w}\right)-F\left(w\right)\right]dw\\
 & =\int_{x}^{\infty}\left[1-F\left(w\right)\right]dw.
\end{align*}
In summary, an equivalent expression for $\mathcal{T}$ is
\begin{align*}
\mathcal{T}\left(x\right)=z\left(1-\beta\right)+\beta\left\{ x+\int_{x}^{\infty}\left[1-F\left(w\right)\right]dw\right\} .
\end{align*}

I verify that $\mathcal{T}$ is a contraction directly.
Take $w_{1},w_{2}\in\left[\underline{w},\infty\right)$
and without loss of generality assume $w_{1}<w_{2}$.
Then 
\begin{align*}
\left|\mathcal{T}\left(w_{2}\right)-\mathcal{T}\left(w_{1}\right)\right| & =\beta\left|w_{2}+\int_{w_{2}}^{\infty}\left[1-F\left(w\right)\right]dw-w_{1}-\int_{w_{1}}^{\infty}\left[1-F\left(w\right)\right]dw\right|\\
                                                                         &=\beta \left|w_{2}-w_{1}+\int_{w_{2}}^{\infty}\left[1-F\left(w\right)\right]dw \right. \\
                                                                           &\quad \left. - \left\{ \int_{w_{1}}^{w_{2}}\left[1-F\left(w\right)\right]dw+\int_{w_{2}}^{\infty}\left[1-F\left(w\right)\right]dw\right\} \right|\\
 &= \beta\left|w_{2}-w_{1}-\int_{w_{1}}^{w_{2}}\left[1-F\left(w\right)\right]dw\right|\\
 &= \beta\left|w_{2}-w_{1}-\int_{w_{1}}^{w_{2}}1dw+\int_{w_{1}}^{w_{2}}F\left(w\right)dw\right|\\
 &= \beta\left|\int_{w_{1}}^{w_{2}}F\left(w\right)dw\right|\\
 & \leq\beta\left|\int_{w_{1}}^{w_{2}}1dw\right|\\
 &= \beta\left|w_{2}-w_{1}\right|,
\end{align*}
establishing that the map is contracting. Because $\left[\underline{w},\infty\right)$
is a closed subset of a complete metric space, it is complete. An
appeal to the contraction mapping theorem again establishes the existence
and uniqueness of $w_{R}\left(0\right)$. 

\ul{\mbox{Step \ref{item:basic:sequence}}}: Given $w_{R}\left(0\right)$,
the sequence of reservation wages can be computed from \eqref{eq:wRn}.

The next part of the proof establishes that the sequence of reservations is increasing in the remaining days of UI benefits.
The proof goes by induction.
I first need to check that $w_{R}\left(1\right)>w_{R}\left(0\right)$.
This is indeed the case because UI compensation benefits are included in $w_{R}\left(1\right)$
but not $w_{R}\left(0\right)$:
\begin{align*}
w_{R}\left(1\right)-w_{R}\left(0\right) & =\left(z+c\right)\left(1-\beta\right)+\beta\left\{ \int_{\underline{w}}^{w_{R}\left(0\right)}w_{R}\left(0\right)dF\left(w\right)+\int_{w_{R}\left(0\right)}^{\overline{w}}wdF\left(w\right)\right\} \\
 & \quad-z\left(1-\beta\right)-\beta\left\{ \int_{\underline{w}}^{w_{R}\left(0\right)}w_{R}\left(0\right)dF\left(w\right)+\int_{w_{R}\left(0\right)}^{\overline{w}}wdF\left(w\right)\right\} \\
 & =c\left(1-\beta\right)>0.
\end{align*}
Next, I need to establish that $w_{R}\left(n\right)>w_{R}\left(n-1\right)$
implies $w_{R}\left(n+1\right)>w_{R}\left(n\right)$. Starting from
the expression for $w_{R}\left(n+1\right)$:
\begin{align*}
w_{R}\left(n+1\right) & =\left(z+c\right)\left(1-\beta\right)+\beta\left\{ \int_{\underline{w}}^{w_{R}\left(n\right)}w_{R}\left(n\right)dF\left(w\right)+\int_{w_{R}\left(n\right)}^{\overline{w}}wdF\left(w\right)\right\} \\
 & =\left(z+c\right)\left(1-\beta\right)+\beta\Upsilon\left[w_{R}\left(n\right)\right]\\
 & >\left(z+c\right)\left(1-\beta\right)+\beta\Upsilon\left[w_{R}\left(n-1\right)\right]\\
 & =\left(z+c\right)\left(1-\beta\right)+\beta\left\{ \int_{\underline{w}}^{w_{R}\left(n-1\right)}w_{R}\left(n-1\right)dF\left(w\right)+\int_{w_{R}\left(n-1\right)}^{\overline{w}}wdF\left(w\right)\right\} \\
 & =w_{R}\left(n\right),
\end{align*}
where the inequality follows from the fact that $\Upsilon$ is a strictly
increasing function (lemma \ref{lemma:upsilon}) and the induction
hypothesis: $w_{R}\left(n\right)>w_{R}\left(n-1\right)$. Thus, the
sequence of reservation wages is increasing in $n$. Thus $w_{R}\left(n\right)>w_{R}\left(n-1\right)$
for all positive integers $n$. 

\ul{\mbox{Step \ref{item:basic:range}}}: Finally,
I verify that the elements of sequence fall between $\underline{w}$ and $\overline{w}$.
Suppose to the contrary that $w_{R}\left(0\right)=\underline{w}$.
From (\ref{eq:map-wR}), $w_{R}\left(0\right)$ satisfies
\begin{align*}
w_{R}\left(0\right) & =z\left(1-\beta\right)+\beta\left\{ \int_{\underline{w}}^{w_{R}\left(0\right)}w_{R}\left(0\right)dF\left(w\right)+\int_{w_{R}\left(0\right)}^{\overline{w}}wdF\left(w\right)\right\} \\
 & =z\left(1-\beta\right)+\beta\int_{\underline{w}}^{\overline{w}}wdF\left(w\right)\\
 & =z\left(1-\beta\right)+\beta\mu_{w}\\
 & >\underline{w},
\end{align*}
which is a contradiction.
Thus $w_{R}\left(0\right)>\underline{w}$.
The inequality also holds for the case where the support of wages is $\left[\underline{w},\infty\right)$ by the same argument. 

The proof that $w_{R}\left(N\right)<\overline{w}$ is completed by induction.
The first part of the induction argument checks that
$w_{R}\left(0\right)$ and $w_{R}\left(1\right)$ are less than $\overline{w}$ on $\left[\underline{w},\overline{w}\right]$.
Suppose not; that is, suppose $w_{R}\left(0\right)=\overline{w}$.
Then the expression for $w_{R}\left(0\right)$ implies $\overline{w}=w_{R}\left(0\right)=z\left(1-\beta\right)+\beta\overline{w}$ or $\overline{w}=z$,
which is a contradiction.
Thus, $w_{R}\left(0\right)<\overline{w}$.
Similarly,
suppose $w_{R}\left(1\right)=\overline{w}$.
Then the expression for $w_{R}\left(1\right)$ given in \eqref{eq:wRn} implies $w_{R}\left(1\right)=\left(z+c\right)\left(1-\beta\right)+\beta\overline{w}$,
which is a contradiction.

Next I want to show $w_{R}\left(N-1\right)<\overline{w}$ implies
$w_{R}\left(N\right)<\overline{w}$.
Note:
\begin{align*}
w_{R}\left(N\right) & =\left(z+c\right)\left(1-\beta\right)+\beta\left\{ \int_{\underline{w}}^{w_{R}\left(N-1\right)}w_{R}\left(N-1\right)dF\left(w\right)+\int_{w_{R}\left(N-1\right)}^{\overline{w}}wdF\left(w\right)\right\} \\
 & =\left(z+c\right)\left(1-\beta\right)+\beta\Upsilon\left(w_{R}\left(N-1\right)\right)\\
 & <\left(z+c\right)\left(1-\beta\right)+\beta\Upsilon\left(\overline{w}\right)\\
 & =\left(z+c\right)\left(1-\beta\right)+\beta\overline{w}\\
 & <\overline{w},
\end{align*}
where the first inequality uses the fact that $\Upsilon$ is increasing---lemma
\ref{lemma:upsilon}---and the second inequality comes from the fact that a weighted average of $z+c$ and the maximum wage is less than the maximum wage:
\begin{align*}
\left(z+c\right)\left(1-\beta\right)+\beta\overline{w} & <\overline{w}\\
\iff\left(z+c\right)\left(1-\beta\right) & <\overline{w}\left(1-\beta\right)\\
\therefore z+c & <\overline{w}.
\end{align*}
Thus $w\left(n\right)<\overline{w}$ for all non-negative integers.
If the the support of wages is replaced with $\left[\overline{w},\infty\right)$,
then each wage will be finite as long as the truncated wage offer
distribution has finite expected value,
which will be satisfied for many distributions.

\section{Allowing for Benefits to be Extended by Discretion}
\label{sec:app-extens}

In this section,
I provide the algebraic details behind expressions found in the main text in section \ref{sec:job-search-benefits}.

\subsection{Description of the Economic Environment}

The economic environment considered in this section
is the same as the economic environment is section \ref{sec:app-basic-job-search},
except that, each period,
there is the chance that benefits are extended.
In that sense,
the McCall model with expiring benefits considered by \citet{burdett_1979} is a
particular case of the environment considered here.

Like in the McCall model with finite benefits,
the value of a job is $W\left(w\right)=w/\left(1-\beta\right)$.
The value of unemployment
when there is no chance of extension is denoted by $U$.
The value of unemployment when there is a chance benefits are extended is denoted by $U^{\delta}$.
Both $U$ and $U^{\delta}$ depend on the remaining periods of UI compensation.
Likewise,
when there is a perceived chance that benefits will be extended,
the reservation wage differs from the case when there is no chance.
The reservation wage when there is a  chance of an extension is denoted by $w_{R}^{\delta}$.
The reservation wage when there is no chance of an extension is denoted by $w_{R}$.
Both $w_{R}^{\delta}$ and $w_{R}$ depend on the remaining periods of UI compensation.

As in section \ref{sec:app-basic-job-search},
the model can be expressed in terms of reservation wages.
The value $U$ corresponds to the basic model of job search in section \ref{sec:app-basic-job-search}.
When there is a perceived chance that benefits will be extended,
the value of unemployment is, for $n\in\left\{ 1,2,\dots\right\}$,
\begin{align*}
U^{\delta}\left(n\right) &= z+c+\beta\left\{ \delta \E\left[\max\left\{ U\left(n-1+\Delta\right),W\left(w\right)\right\} \right]+\left(1-\delta\right) \E\left[\max\left\{ U^{\delta}\left(n-1\right),W\left(w\right)\right\} \right]\right\} \\
 & =z+c+\delta\beta\left\{ \int_{\underline{w}}^{w_{R}\left(n-1+\Delta\right)}w_{R}\left(n-1+\Delta\right)dF\left(w\right)+\int_{w_{R}\left(n-1+\Delta\right)}^{\overline{w}}wdF\left(w\right)\right\} \\
 & \quad+\left(1-\delta\right)\beta\left\{ \int_{\underline{w}}^{w_{R}^{\delta}\left(n-1\right)}w_{R}^{\delta}\left(n-1\right)dF\left(w\right)+\int_{w_{R}^{\delta}\left(n-1\right)}^{\overline{w}}wdF\left(w\right)\right\} .
\end{align*}
The first component of the expression, $z+b$,
corresponds to the value of nonwork plus the unemployment benefit.
The following period, discounted by $\beta$,
corresponds to choosing to accept or reject a job when benefits have or have not been extended,
which occurs with probability $\delta$ and $1-\delta$.
When $n=0$,
UI compensation is unavailable:
\begin{align*}
  U^{\delta} \left( 0 \right) &= z + \beta \left\{ \E \left[ \max \left\{ U \left( \Delta \right) \right\}, W \left( w \right) \right] \right\}
                                + \left( 1-\delta \right) \E \left[ \max \left\{ U^{\delta} \left( 0 \right), W \left( w \right) \right\}  \right] \\
                              &= z + \delta \beta \left\{ \int\limits_{\underline{w}}^{w_{R}\left( \Delta \right)}w_{R} \left( \Delta \right)dF \left( w \right)
                                + \int\limits_{w_{R}\left( \Delta \right)}^{\overline{w}} w dF \left( w \right) \right\} \\
                              &\quad + \left( 1-\delta \right) \beta \left\{ \int\limits_{\underline{w}}^{w_{R}^{\delta} \left( 0 \right)} w_{R}^{\delta}\left( 0 \right) dF \left( w \right)
                                + \int\limits_{w_{R}^{\delta}\left( 0 \right)}^{\overline{w}} w dF \left( w \right) \right\}.
\end{align*}

Expanding the expression for $U^{\delta}\left(n\right)$ implies,
for $n \in \left\{ 1, \dots, N \right\}$,
\begin{align*}
\frac{w_{R}^{\delta}\left(n\right)}{1-\beta} & =z+c+\beta\left\{ \delta \E\left[\max\left\{ \frac{w_{R}\left(n-1+\Delta\right)}{1-\beta},\frac{w}{1-\beta}\right\} \right]+\left(1-\delta\right)\E\left[\max\left\{ \frac{w_{R}^{\delta}\left(n-1\right)}{1-\beta},\frac{w}{1-\beta}\right\} \right]\right\}.
\end{align*}
Therefore
\begin{align*}
w_{R}^{\delta}\left(n\right) = \left(z+c\right)\left(1-\beta\right)+\beta\left\{ \delta \E\left[\max\left\{ w_{R}\left(n-1+\Delta\right),w\right\} \right]+\left(1-\delta\right)\E\left[\max\left\{ w_{R}^{\delta}\left(n-1\right),w\right\} \right]\right\}
\end{align*}
or 
\begin{align}
  \label{eq:wR-delta}
  \begin{split}
w_{R}^{\delta} \left( n \right) &= \left(z+c\right)\left(1-\beta\right)+\beta\delta\int_{\underline{w}}^{w_{R}\left(n-1+\Delta\right)}w_{R}\left(n-1+\Delta\right)dF\left(w\right)+\beta\delta\int_{w_{R}\left(n-1+\Delta\right)}^{\overline{w}}wdF\left(w\right) \\
 &\quad +\beta\left(1-\delta\right)\int_{\underline{w}}^{w_{R}^{\delta}\left(n-1\right)}w_{R}^{\delta}\left(n-1\right)dF\left(w\right)+\beta\left(1-\delta\right)\int_{w_{R}^{\delta}\left(n-1\right)}^{\overline{w}}wdF\left(w\right).    
  \end{split}
\end{align}
When the worker has one remaining period of benefits,
their reservation wage satisfies
\begin{align*}
w_{R}^{\delta}\left(1\right) & =\left(z+c\right)\left(1-\beta\right)+\beta\delta\int_{\underline{w}}^{w_{R}\left(\Delta\right)}w_{R}\left(\Delta\right)dF\left(w\right)+\beta\delta\int_{w_{R}\left(\Delta\right)}^{\overline{w}}wdF\left(w\right)\\
 & \quad+\beta\left(1-\delta\right)\int_{\underline{w}}^{w_{R}^{\delta}\left(0\right)}w_{R}^{\delta}\left(0\right)dF\left(w\right)+\beta\left(1-\delta\right)\int_{w_{R}^{\delta}\left(0\right)}^{\overline{w}}wdF\left(w\right);
\end{align*}
and when there are no remaining periods of benefits,
their reservation wage satisfies
\begin{align*}
w_{R}^{\delta}\left(0\right)=z\left(1-\beta\right)+\beta\left\{ \delta \E\left[\max\left\{ w_{R}\left(\Delta\right),w\right\} \right]+\left(1-\delta\right)\E\left[\max\left\{ w_{R}^{\delta}\left(0\right),w\right\} \right]\right\} 
\end{align*}
or equivalently
\begin{align}
  \label{eq:wR-delta-0}
  \begin{split}
w_{R}^{\delta}\left(0\right) & =z\left(1-\beta\right)+\beta\delta\int_{\underline{w}}^{w_{R}\left(\Delta\right)}w_{R}\left(\Delta\right)dF\left(w\right)+\beta\delta\int_{w_{R}\left(\Delta\right)}^{\overline{w}}wdF\left(w\right) \\
 & \quad+\beta\left(1-\delta\right)\int_{\underline{w}}^{w_{R}^{\delta}\left(0\right)}w_{R}^{\delta}\left(0\right)dF\left(w\right)+\beta\left(1-\delta\right)\int_{w_{R}^{\delta}\left(0\right)}^{\overline{w}}wdF\left(w\right).    
  \end{split}
\end{align}
Optimal decision rules are characterized by a sequence of reservation wages that are increasing in $n$. 

\subsection{Interpreting Search in terms of Reservation Wages}

For an interpretation of the search problem in terms of reservation wages,
I follow \citet{ljungqvist_sargent_2018} and write $w_{R}^{\delta} \left( 0 \right)$ as
\begin{align*}
  w_{R}^{\delta}\left(0\right)\int_{\underline{w}}^{\overline{w}}dF\left(w\right) &= z\left(1-\beta\right)+\beta\delta\int_{\underline{w}}^{w_{R}\left(\Delta\right)}w_{R}\left(\Delta\right)dF\left(w\right)+\beta\delta\int_{w_{R}\left(\Delta\right)}^{\overline{w}}wdF\left(w\right) \\
 &\quad +\beta\left(1-\delta\right)\int_{\underline{w}}^{w_{R}^{\delta}\left(0\right)}w_{R}^{\delta}\left(0\right)dF\left(w\right)+\beta\left(1-\delta\right)\int_{w_{R}^{\delta}\left(0\right)}^{\overline{w}}wdF\left(w\right).
\end{align*}
Therefore
\begin{multline*}
  \beta w_{R}^{\delta}\left(0\right)\int_{\underline{w}}^{\overline{w}}dF\left(w\right)+\left(1-\beta\right)w_{R}^{\delta}\left(0\right)\int_{\underline{w}}^{\overline{w}}dF\left(w\right)-z\left(1-\beta\right) \\
  =\beta\delta\left[\int_{\underline{w}}^{w_{R}\left(\Delta\right)}w_{R}\left(\Delta\right)dF\left(w\right)+\int_{w_{R}\left(\Delta\right)}^{\overline{w}}wdF\left(w\right)\right]\\
  +\beta\left(1-\delta\right)\left[\int_{\underline{w}}^{w_{R}^{\delta}\left(0\right)}w_{R}^{\delta}\left(0\right)dF\left(w\right)+\int_{w_{R}^{\delta}\left(0\right)}^{\overline{w}}wdF\left(w\right)\right] .
\end{multline*}
Subtracting $\beta w_{R}^{\delta} \left( 0 \right)$ from both sides and collecting terms yields
\begin{multline*}
\left(1-\beta\right)\left[w_{R}^{\delta}\left(0\right)-z\right] \\
= \beta\delta\left\{ \int_{\underline{w}}^{w_{R}\left(\Delta\right)}\left[w_{R}\left(\Delta\right)-w_{R}^{\delta}\left(0\right)\right]dF\left(w\right)+\int_{w_{R}\left(\Delta\right)}^{\overline{w}}\left[w-w_{R}^{\delta}\left(0\right)\right]dF\left(w\right)\right\} \\
 +\beta\left(1-\delta\right)\left\{ \int_{w_{R}^{\delta}\left(0\right)}^{\overline{w}}\left[w-w_{R}^{\delta}\left(0\right)\right]dF\left(w\right)\right\}.
\end{multline*}
This expression implies
\begin{align*}
w_{R}^{\delta}\left(0\right)-z & =\frac{\beta}{1-\beta}\delta\left\{ \int_{\underline{w}}^{w_{R}\left(\Delta\right)}\left[w_{R}\left(\Delta\right)-w_{R}^{\delta}\left(0\right)\right]dF\left(w\right)+\int_{w_{R}\left(\Delta\right)}^{\overline{w}}\left[w-w_{R}^{\delta}\left(0\right)\right]dF\left(w\right)\right\} \\
 & \quad+\frac{\beta}{1-\beta}\left(1-\delta\right)\left\{ \int_{w_{R}^{\delta}\left(0\right)}^{\overline{w}}\left[w-w_{R}^{\delta}\left(0\right)\right]dF\left(w\right)\right\},
\end{align*}
which is equation \eqref{eq:search-reservation-interpretation} in the main text.
This expression generalizes equation (6.3.3) in \citet[163]{ljungqvist_sargent_2018}.
A similar interpretation is available for $w_{R}^{\delta} \left( n \right)$.

\section{Proof of Proposition \ref{prop:1} in the Main Text}
\label{sec:proof-pos-ext}

Proposition \ref{prop:1} in the main text is established by a series of lemmas.
\begin{description}
\item[Lemma \ref{lemma:wR0-delta-exists}: Existence and uniqueness of $w_{0}^{\delta} \left( 0 \right)$.] Lemma \ref{lemma:wR0-delta-exists} establishes
  the existence and uniqueness of $w_{R}^{\delta} \left( 0 \right)$ by appealing to the contraction mapping theorem.  
\item[Lemma \ref{lemma:optimal-search-delta}: Reservation wages increase in the remaining periods of UI compensation benefits] Lemma \ref{lemma:optimal-search-delta}
  establishes monotonicity: $w_{R}^{\delta} \left( n+1 \right) > w_{R}^{\delta} \left( n \right)$.
\item[Lemma \ref{lemma-just-extended}: Auxiliary result.] Lemma \ref{lemma-just-extended} establishes an auxiliary result used in later proofs.  
\item[Lemma \ref{lemma:bounds}: Reservation wages are interior] Lemma \ref{lemma:bounds} establishes that
  reservation wages fall within the interior of the support of wage offers:
  $\overline{w} >  w_{R}^{\delta} \left( N \right) > \cdots > w_{R}^{\delta} \left( 0 \right) > \underline{w}$.
\end{description}
Lemmas \ref{lemma:wR0-delta-exists}, \ref{lemma:optimal-search-delta}, and \ref{lemma:bounds} are established in order.
Together, the lemmas directly establish \ref{prop:1}.

\subsection{Proof of Lemma \ref{lemma:wR0-delta-exists}}

\begin{lemma}
\label{lemma:wR0-delta-exists}
The reservation wage $w_{R}^{\delta}\left(0\right)$ exists and is unique. 
\end{lemma}
\begin{proof}
Define the function $\mathcal{T}^{\delta}$ on $\left[\underline{w},\overline{w}\right]$ as
\begin{align}
  \label{eq:map-wR-delta}
  \begin{split}
\mathcal{T}^{\delta}\left(x\right) & \equiv z\left(1-\beta\right)+\beta\delta\left[\int_{\underline{w}}^{w_{R}\left(\Delta\right)}w_{R}\left(\Delta\right)dF\left(w\right)+\int_{w_{R}\left(\Delta\right)}^{\overline{w}}wdF\left(w\right)\right] \\
 &\quad + \beta\left(1-\delta\right)\left[\int_{\underline{w}}^{x}xdF\left(w\right)+\int_{x}^{\overline{w}}wdF\left(w\right)\right] \\
 &= z\left(1-\beta\right)+\beta\delta\Upsilon\left(w_{R}\left(\Delta\right)\right)+\beta\left(1-\delta\right)\left[\int_{\underline{w}}^{x}xdF\left(w\right)+\int_{x}^{\overline{w}}wdF\left(w\right)\right]    
  \end{split}
\end{align}
I want to use the contraction mapping theorem to establish the existence and uniqueness of $w_{R}^{\delta} \left( 0 \right)$.
I therefore need to show that $\mathcal{T}^{\delta}$ is a self map and it contracts.

I first establish that $\mathcal{T}^{\delta}$ is a self-map.
First,
\begin{align*}
\mathcal{T}^{\delta}\left(\underline{w}\right) & =z\left(1-\beta\right)+\beta\left[\delta\Upsilon\left(w_{R}\left(\Delta\right)\right)+\left(1-\delta\right)\mu_{w}\right]\\
 & >z\left(1-\beta\right)+\beta\left[\delta\Upsilon\left(\underline{w}\right)+\left(1-\delta\right)\mu_{w}\right]\\
 & =z\left(1-\beta\right)+\beta\left[\delta\mu_{w}+\left(1-\delta\right)\mu_{w}\right]\\
 & =z\left(1-\beta\right)+\beta\mu_{w}\\
 & >\underline{w},
\end{align*}
where the first inequality uses the fact that $\Upsilon$ is increasing (lemma \ref{lemma:upsilon}) and
the second inequality follows by assumption.
Second,
\begin{align*}
\mathcal{T}^{\delta}\left(\overline{w}\right) & =z\left(1-\beta\right)+\beta\left[\delta\Upsilon\left(w_{R}\left(\Delta\right)\right)+\left(1-\delta\right)\overline{w}\right]\\
 & <z\left(1-\beta\right)+\beta\left[\delta\Upsilon\left(\overline{w}\right)+\left(1-\delta\right)\overline{w}\right]\\
 & =z\left(1-\beta\right)+\beta\left[\delta\overline{w}+\left(1-\delta\right)\overline{w}\right]\\
 & =z\left(1-\beta\right)+\beta\overline{w}\\
 & <\overline{w}
\end{align*}
as $z\left(1-\beta\right)+\beta\overline{w}<\overline{w}$
if and only if $z<\overline{w}$, which is true by assumption.
Third, lemma \ref{lemma:upsilon} establishes that $\mathcal{T}^{\delta}$ is increasing.
Thus, $\underline{w} < \mathcal{T}^{\delta}\left(x\right) < \overline{w}$ for all $x\in\left[\underline{w},\overline{w}\right]$,
establishing that $\mathcal{T}^{\delta}$ is a self-map. 

Investigating the derivative further yields
\begin{align*}
\left( \mathcal{T}^{\delta} \right)^{\prime} \left(x\right) &= \beta\left(1-\delta\right)xf\left(x\right)+\beta\left(1-\delta\right)\int_{0}^{x}1dF\left(w\right)-\beta\left(1-\delta\right)xf\left(x\right)\\
 & =\beta\left(1-\delta\right)\int_{0}^{x}1dF\left(w\right)\\
 & =\beta\left(1-\delta\right)F\left(w\right)\bigg\vert_{w=0}^{w=x}\\
 & =\beta\left(1-\delta\right)F\left(x\right).
\end{align*}
Therefore $0< \left( \mathcal{T}^{\delta} \right)^{\prime} \left( x \right)<1$ for all $x\in\left(\underline{w},\overline{w}\right)$
and thus $\mathcal{T}^{\delta}$ contracts.
The contraction mapping
theorem implies the existence and uniqueness of the fixed point $w_{R}^{\delta}\left(0\right)$.
A similar proof to the one given above could extend the support of wages to $\left[\underline{w},\infty\right)$.
\end{proof}

The next step establishes that
the optimal policy in the environment with a perceived extension to UI compensation benefits involves a sequence of reservation wages.
Reservation wages increase in the remaining periods of UI benefits.
The following proposition establishes this result. 

\subsection{Proof of Lemma \ref{lemma:optimal-search-delta}}

\begin{lemma}[Optimal policy]
  \label{lemma:optimal-search-delta}
  Assume the search environment in proposition \ref{prop:1} in the main text holds.
  In the environment,
  there is a chance that UI compensation benefits can be extended.
  A worker's optimal policy is a sequence of reservation wages
  that are increasing in the remaining periods of UI compensation.
\end{lemma}

\begin{proof}
  First, lemma \ref{lemma:wR0-delta-exists} establishes that
  the reservation wage $w_{R}^{\delta}\left(0\right)$ exists and is unique.
  Next I want to show that  
\begin{align*}
w_{R}^{\delta} \left( N \right) > \cdots > w_{R}^{\delta}\left(n+1\right)>w_{R}^{\delta}\left(n\right)>\cdots>w_{R}^{\delta}\left(1\right)>w_{R}^{\delta}\left(0\right).
\end{align*}
This is accomplished by induction.

It is true that $w_{R}^{\delta}\left(1\right)>w_{R}^{\delta}\left(0\right)$:
\begin{align*}
w_{R}^{\delta}\left(1\right)-w_{R}^{\delta}\left(0\right) & =\left(z+c\right)\left(1-\beta\right)+\beta\delta\left[\int_{\underline{w}}^{w_{R}\left(\Delta\right)}w_{R}\left(\Delta\right)dF\left(w\right)+\int_{w_{R}\left(\Delta\right)}^{\overline{w}}wdF\left(w\right)\right]\\
 & \quad+\beta\left(1-\delta\right)\left[\int_{\underline{w}}^{w_{R}^{\delta}\left(0\right)}w_{R}^{\delta}\left(0\right)dF\left(w\right)+\int_{w_{R}^{\delta}\left(0\right)}^{\overline{w}}wdF\left(w\right)\right]\\
 & \quad-z\left(1-\beta\right)-\beta\delta\left[\int_{\underline{w}}^{w_{R}\left(\Delta\right)}w_{R}\left(\Delta\right)dF\left(w\right)+\int_{w_{R}\left(\Delta\right)}^{\overline{w}}wdF\left(w\right)\right]\\
 & \quad-\beta\left(1-\delta\right)\left[\int_{\underline{w}}^{w_{R}^{\delta}\left(0\right)}w_{R}^{\delta}\left(0\right)dF\left(w\right)+\int_{w_{R}^{\delta}\left(0\right)}^{\overline{w}}wdF\left(w\right)\right]\\
 & =c\left(1-\beta\right)>0.
\end{align*}
Then, using $w_{R}^{\delta}\left(n\right)>w_{R}^{\delta}\left(n-1\right)$,
I want to show $w_{R}^{\delta}\left(n+1\right)>w_{R}^{\delta}\left(n\right)$.
Using the expression for $w_{R}^{\delta}$ in \eqref{eq:wR-delta},
\begin{align*}
w_{R}^{\delta}\left(n+1\right) & =\left(z+c\right)\left(1-\beta\right)+\beta\delta\left[\int_{\underline{w}}^{w_{R}\left(n+\Delta\right)}w_{R}\left(n+\Delta\right)dF\left(w\right)+\int_{w_{R}\left(n+\Delta\right)}^{\overline{w}}wdF\left(w\right)\right]\\
 & \quad+\beta\left(1-\delta\right)\left[\int_{\underline{w}}^{w_{R}^{\delta}\left(n\right)}w_{R}^{\delta}\left(n\right)dF\left(w\right)+\int_{w_{R}^{\delta}\left(n\right)}^{\overline{w}}wdF\left(w\right)\right]\\
 & =\left(z+c\right)\left(1-\beta\right)+\beta\delta\Upsilon\left[w_{R}\left(n+\Delta\right)\right]+\beta\left(1-\delta\right)\Upsilon\left[w_{R}^{\delta}\left(n\right)\right]\\
 & >\left(z+c\right)\left(1-\beta\right)+\beta\delta\Upsilon\left[w_{R}\left(n-1+\Delta\right)\right]+\beta\left(1-\delta\right)\Upsilon\left[w_{R}^{\delta}\left(n-1\right)\right]\\
 & =w_{R}^{\delta}\left(n\right),
\end{align*}
where the inequality uses the fact that $\Upsilon$ is increasing
(lemma \ref{lemma:upsilon}), the induction hypothesis, and proposition \ref{prop:text:basic-incr-reservation-wages} in the main text,
which implies $w_{R}\left(n+\Delta\right)>w_{R}\left(n-1+\Delta\right)$.
\end{proof}

\subsection{Proof of Lemma \ref{lemma-just-extended}}

\begin{lemma}
  \label{lemma-just-extended}
  It is true that $w_{R}\left(n+\Delta\right)>w_{R}^{\delta}\left(n\right)$ for $n\in\left\{ 0,\dots,N\right\}$.
\end{lemma}
\begin{proof}
  The lemma is established by induction.
  I first establish that $w_{R}\left(\Delta\right)>w_{R}^{\delta}\left(0\right)$.

The expression for $w_{R} \left( \Delta \right)$ satisfies
\begin{align*}
w_{R}\left(\Delta\right) = \left(z+c\right)\left(1-\beta\right)+\beta\left\{ \int_{\underline{w}}^{w_{R}\left(\Delta-1\right)}w_{R}\left(\Delta-1\right)dF\left(w\right)+\int_{w_{R}\left(\Delta-1\right)}^{\overline{w}}wdF\left(w\right)\right\}.
\end{align*}
The expression for $w_{R}^{\delta} \left( 0 \right)$ satisfies
\begin{align}
  \label{eq:wR0:delta-01}
  \begin{split}
    w_{R}^{\delta}\left(0\right) & =z\left(1-\beta\right)+\beta\delta\int_{\underline{w}}^{w_{R}\left(\Delta\right)}w_{R}\left(\Delta\right)dF\left(w\right)+\beta\delta\int_{w_{R}\left(\Delta\right)}^{\overline{w}}wdF\left(w\right) \\
 &\quad +\beta\left(1-\delta\right)\int_{\underline{w}}^{w_{R}^{\delta}\left(0\right)}w_{R}^{\delta}\left(0\right)dF\left(w\right)+\beta\left(1-\delta\right)\int_{w_{R}^{\delta}\left(0\right)}^{\overline{w}}wdF\left(w\right).
  \end{split}
\end{align}
I establish $w_{R}\left(\Delta\right)>w_{R}^{\delta}\left(0\right)$ by contradiction. 

There are two cases to consider: the case where $w_{R}\left(\Delta\right)=w_{R}^{\delta}\left(0\right)$
and the case where $w_{R}\left(\Delta\right)<w_{R}^{\delta}\left(0\right)$.
First, I suppose $w_{R}\left(\Delta\right)=w_{R}^{\delta}\left(0\right)$.
This equality implies that, using the expression for $w_{R}^{\delta}\left(0\right)$ in \eqref{eq:wR0:delta-01},
\begin{align*}
w_{R}\left(\Delta\right)=z\left(1-\beta\right)+\beta\left[\int_{\underline{w}}^{w_{R}\left(\Delta\right)}w_{R}\left(\Delta\right)dF\left(w\right)+\int_{w_{R}\left(\Delta\right)}^{\overline{w}}wdF\left(w\right)\right].
\end{align*}
Which implies $w_{R}\left(\Delta\right)$ is a fixed point of $\mathcal{T}$ defined in (\ref{eq:map-wR}).
The properties of $\mathcal{T}$ imply $w_{R}\left(\Delta\right)=w_{R}\left(0\right)$,
which contradicts the results in proposition \ref{prop:text:basic-incr-reservation-wages} in the main text. 

Second, I consider the case where $w_{R}\left(\Delta\right)<w_{R}^{\delta}\left(0\right)$.
Then
\begin{align*}
w_{R}^{\delta}\left(0\right) & =z\left(1-\beta\right)+\beta\delta\left[\int_{\underline{w}}^{w_{R}\left(\Delta\right)}w_{R}\left(\Delta\right)dF\left(w\right)+\int_{w_{R}\left(\Delta\right)}^{\overline{w}}wdF\left(w\right)\right]\\
 & \quad+\beta\left(1-\delta\right)\left[\int_{\underline{w}}^{w_{R}^{\delta}\left(0\right)}w_{R}^{\delta}\left(0\right)dF\left(w\right)+\int_{w_{R}^{\delta}\left(0\right)}^{\overline{w}}wdF\left(w\right)\right]\\
 & =z\left(1-\beta\right)+\beta\delta\Upsilon\left[w_{R}\left(\Delta\right)\right]+\beta\left(1-\delta\right)\Upsilon\left[w_{R}^{\delta}\left(0\right)\right]\\
 & <z\left(1-\beta\right)+\beta\Upsilon\left[w_{R}^{\delta}\left(0\right)\right],
\end{align*}
where the inequality uses the fact that $\Upsilon$ is increasing,
which is established in lemma \ref{lemma:upsilon}.
The inequality, using the definition of $\Upsilon$, can be expressed as
\begin{align*}
w_{R}^{\delta}\left(0\right)<z\left(1-\beta\right)+\beta\left[\int_{\underline{w}}^{w_{R}^{\delta}\left(0\right)}w_{R}^{\delta}\left(0\right)dF\left(w\right)+\int_{w_{R}^{\delta}\left(0\right)}^{\overline{w}}wdF\left(w\right)\right].
\end{align*}
Developing this expression yields
\begin{align}
  \label{eq:wR-wRdelta}    
  \begin{split}
w_{R}^{\delta}\left(0\right)-\beta w_{R}^{\delta}\left(0\right) &< z\left(1-\beta\right) + \beta\left[\int_{\underline{w}}^{w_{R}^{\delta}\left(0\right)}w_{R}^{\delta}\left(0\right)dF\left(w\right)+\int_{w_{R}^{\delta}\left(0\right)}^{\overline{w}}wdF\left(w\right)\right] \\
  &\quad- \beta w_{R}^{\delta}\left(0\right)\int_{\underline{w}}^{\overline{w}}dF\left(w\right) \\
 &= z\left(1-\beta\right)+\beta\int_{w_{R}^{\delta}\left(0\right)}^{\overline{w}}\left[w-w_{R}^{\delta}\left(0\right)\right]dF\left(w\right) \\
\therefore w_{R}^{\delta}\left(0\right) & <z+\frac{\beta}{1-\beta}\int_{w_{R}^{\delta}\left(0\right)}^{\overline{w}}\left[w-w_{R}^{\delta}\left(0\right)\right]dF\left(w\right) \\
\therefore 0 & <-w_{R}^{\delta}\left(0\right)+z+\frac{\beta}{1-\beta}\int_{w_{R}^{\delta}\left(0\right)}^{\overline{w}}\left[w-w_{R}^{\delta}\left(0\right)\right]dF\left(w\right).
  \end{split}
\end{align}
I consider the function 
\begin{equation}
\Xi\left(x\right) \equiv -x + z + \frac{\beta}{1-\beta}\int_{x}^{\overline{w}} \left( w-x \right)  dF\left(w\right).\label{eq:Xi}
\end{equation}
Comparison of $\Xi$ with $\mathcal{T}$ in \eqref{eq:map-wR} establishes that $w_{R}\left(0\right)$ solves $\Xi \left( w_{R}\left(0\right) \right) = 0$.
Indeed, if $\chi$ solves $\Xi\left(\chi\right)=0$, then
\begin{align*}
0 & =-\chi+z+\frac{\beta}{1-\beta}\int_{\chi}^{\overline{w}}\left(w-\chi\right)dF\left(w\right)\\
\iff\chi & =z+\frac{\beta}{1-\beta}\int_{\chi}^{\overline{w}}\left(w-\chi\right)dF\left(w\right)\\
\iff\chi\left(1-\beta\right) & =z\left(1-\beta\right)+\beta\int_{\chi}^{\overline{w}}\left(w-\chi\right)dF\left(w\right)\\
\iff\chi & =z\left(1-\beta\right)+\beta\int_{\chi}^{\overline{w}}\left(w-\chi\right)dF\left(w\right)+\beta\chi\int_{\underline{w}}^{\overline{w}}dF\left(w\right)\\
\iff\chi & =z\left(1-\beta\right)+\beta\left[\int_{\chi}^{\overline{w}}\left(w-\chi\right)dF\left(w\right)+\chi\int_{\underline{w}}^{\overline{w}}dF\left(w\right)\right]\\
\iff\chi & =z\left(1-\beta\right)+\beta\left[\int_{\underline{w}}^{\chi}\chi dF\left(w\right)+\int_{\chi}^{\overline{w}}wdF\left(w\right)\right]
\end{align*}
and $\chi$ is a fixed point of $\mathcal{T}$.
Additionally,
\begin{align*}
\Xi^{\prime}\left(x\right) &= -1-\frac{\beta}{1-\beta}\left(x-x\right)f\left(x\right)-\frac{\beta}{1-\beta}\int_{x}^{\overline{w}}f\left(w\right)dw\\
 &= -1-\frac{\beta}{1-\beta}\left[1-F\left(x\right)\right]\\
 &< 0.
\end{align*}
Because $w_{R}\left(0\right)$ solves $\Xi \left( w_{R}\left(0\right) \right) = 0$ and
$\Xi$ is strictly decreasing, any $x$ that satisfies $\Xi\left(x\right)>0$
must be less than $w_{R}\left(0\right)$.
Figure \ref{fig:fun-Chi} illustrates this idea.

\begin{figure}[htbp]
\centerline{\includegraphics[width=0.8\textwidth]{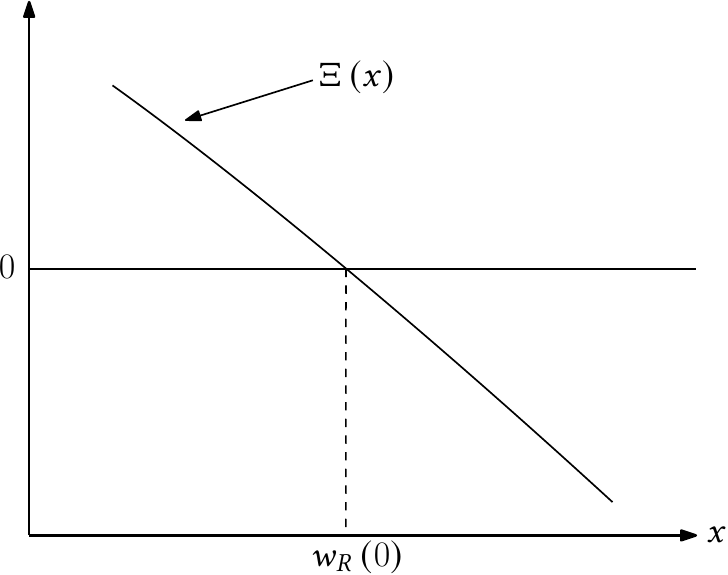}}
\caption[]{\label{fig:fun-Chi} The function $\Xi$ where $\Xi\left[w_{R}\left(0\right)\right]=0$.}
\end{figure}

The inequality in \eqref{eq:wR-wRdelta} establishes that $\Xi \left( w_{R}^{\delta} \left( 0 \right) \right) > 0$.
The properties of $\Xi$ in figure \ref{fig:fun-Chi} imply $w_{R}^{\delta} \left( 0 \right) < w_{R} \left( 0 \right)$.
This inequality along with the assumption that $w_{R} \left( \Delta \right) < w_{R}^{\delta} \left( 0 \right)$ implies that
$w_{R}\left(\Delta\right) < w_{R}^{\delta}\left(0\right) < w_{R}\left(0\right)$,
which contradicts the result in equation \eqref{eq:text:optimal-seq-res-wages} found in proposition \ref{prop:text:basic-incr-reservation-wages} of the main text.
Therefore, $w_{R}\left(\Delta\right)>w_{R}\left(0\right)$.

I have established the base case.
Now I want to show that $w_{R}\left(n+1+\Delta\right)>w_{R}^{\delta}\left(n+1\right)$
using $w_{R}\left(n+\Delta\right)>w_{R}^{\delta}\left(n\right)$.
To do so, I suppose not; that is, I suppose $w_{R}\left(n+1+\Delta\right)\leq w_{R}^{\delta}\left(n+1\right)$.
The expression for $w_{R}^{\delta}$ in \eqref{eq:wR-delta} implies
\begin{align*}
w_{R}\left(n+1+\Delta\right) \leq w_{R}^{\delta}\left(n+1\right) & =\left(z+c\right)\left(1-\beta\right)+\beta\delta\left[\int_{\underline{w}}^{w_{R}\left(n+\Delta\right)}w_{R}\left(n+\Delta\right)dF\left(w\right)+\int_{w_{R}\left(n+\Delta\right)}^{\overline{w}}wdF\left(w\right)\right]\\
 & \quad+\beta\left(1-\delta\right)\left[\int_{\underline{w}}^{w_{R}^{\delta}\left(n\right)}w_{R}^{\delta}\left(n\right)dF\left(w\right)+\int_{w_{R}^{\delta}\left(n\right)}^{\overline{w}}wdF\left(w\right)\right],
\end{align*}
where the inequality uses the supposition that $w_{R}\left(n+1+\Delta\right)\leq w_{R}^{\delta}\left(n+1\right)$.
Using $w_{R}^{\delta}\left(n\right)<w_{R}\left(n+\Delta\right)$ on
the right side of the latter implies
\begin{align*}
w_{R}\left(n+1+\Delta\right) & <\left(z+c\right)\left(1-\beta\right)+\beta\delta\left[\int_{\underline{w}}^{w_{R}\left(n+\Delta\right)}w_{R}\left(n+\Delta\right)dF\left(w\right)+\int_{w_{R}\left(d+\Delta\right)}^{\overline{w}}wdF\left(w\right)\right]\\
 & \quad+\beta\left(1-\delta\right)\left[\int_{\underline{w}}^{w_{R}\left(n+\Delta\right)}w_{R}\left(n+\Delta\right)dF\left(w\right)+\int_{w_{R}\left(d+\Delta\right)}^{\overline{w}}wdF\left(w\right)\right]\\
 & =\left(z+c\right)\left(1-\beta\right)+\beta\left[\int_{\underline{w}}^{w_{R}\left(n+\Delta\right)}w_{R}\left(n+\Delta\right)dF\left(w\right)+\int_{w_{R}\left(n+\Delta\right)}^{\overline{w}}wdF\left(w\right)\right].
\end{align*}
This contradicts the definition of $w_{R}$ in \eqref{eq:text:wRdelta},
which requires the latter expression to hold with equality.
This establishes what was set out to be shown;
namely, $w_{R}\left(n+\Delta\right)>w_{R}^{\delta}\left(n\right)$ for all $n\in\left\{ 0,\dots,N\right\}$.
\end{proof}

\subsection{Proof of Lemma \ref{lemma:bounds}}

\begin{lemma}[Bounds]
  \label{lemma:bounds}
  The sequence of reservation wages is bounded by $\underline{w}$ and $\overline{w}$. 
\end{lemma}

\begin{proof}
To show $w_{R}^{\delta}\left(0\right)>\underline{w}$, suppose not;
that is, suppose $w_{R}^{\delta}\left(0\right)= \underline{w}$.
This implies, using \eqref{eq:wR-delta-0}, that
\begin{align*}
w_{R}^{\delta}\left(0\right) &= z\left(1-\beta\right)+\beta\delta\Upsilon\left(w_{R}\left(\Delta\right)\right)+\left(1-\delta\right)\Upsilon\left(w_{R}^{\delta}\left(0\right)\right) \\
  &> z\left(1-\beta\right)+\beta\delta\Upsilon\left(w_{R}^{\delta}\left(0\right)\right)+\left(1-\delta\right)\Upsilon\left(w_{R}^{\delta}\left(0\right)\right) \\
 &=z\left(1-\beta\right)+\beta\delta\mu_{w}+\left(1-\delta\right)\mu_{w}\\
 & =z\left(1-\beta\right)+\beta\mu_{w}\\
 & >\underline{w},
\end{align*}
where
the first inequality uses lemma \ref{lemma-just-extended} and
the insertion of $\mu_{w}$ uses the fact that integration is being done over the entire support of wages,
establishing a contradiction.
(The case where $w_{R}^{\delta} < \underline{w}$ is vacuous.)
Thus $w_{R}^{\delta}\left(0\right) > \underline{w}$. 

Both $w_{R}^{\delta}\left(0\right)$ and $w_{R}^{\delta}\left(1\right)$
are less than $\overline{w}$. To establish that this is the case,
suppose that $w_{R}^{\delta}\left(0\right)=\overline{w}$.
This implies, using \eqref{eq:wR-delta-0}, that
\begin{align*}
w_{R}^{\delta}\left(0\right)=\overline{w}=z\left(1-\beta\right)+\beta\left[\delta+\left(1-\delta\right)\right]\overline{w}=z\left(1-\beta\right)+\beta\overline{w}<\overline{w}
\end{align*}
where the inequality follows by the assumption that $z < \overline{w}$,
which establishes a contradiction.
If it is supposed that $w_{R}^{\delta}\left(1\right)=\overline{w}$,
then \eqref{eq:wR-delta} implies
\begin{align*}
w_{R}^{\delta}\left(1\right)=\overline{w}=\left(z+c\right)\left(1-\beta\right)+\beta\left[\delta+\left(1-\delta\right)\right]\overline{w}=\left(z+c\right)\left(1-\beta\right)+\beta\overline{w}<\overline{w},
\end{align*}
where the inequality follows by the assumption that $z+c < \overline{w}$,
which establishes a contradiction.
Next I want to show $w_{R}^{\delta}\left(N-1\right)<\overline{w}$ implies $w_{R}^{\delta}\left(N\right)<\overline{w}$.
From \eqref{eq:wR-delta},
\begin{align*}
w_{R}^{\delta} \left( N \right) &= \left(z+c\right)\left(1-\beta\right)+\beta\left[\delta\Upsilon\left(w_{R}\left(N-1+\Delta\right)\right)+\left(1-\delta\right)\Upsilon\left(w_{R}^{\delta}\left(N-1\right)\right)\right] \\
                                &< \left(z+c\right)\left(1-\beta\right)+\beta\left[\delta\Upsilon\left(w_{R}\left(N-1+\Delta\right)\right)
                                  +\left(1-\delta\right)\Upsilon\left( w_{R}\left(N-1+\Delta\right) \right)\right] \\
 &< \left(z+c\right)\left(1-\beta\right)+\beta\left[\delta\overline{w}+\left(1-\delta\right)\overline{w}\right]\\
 &= \left(z+c\right)\left(1-\beta\right)+\beta\overline{w} \\
 &< \overline{w},
\end{align*}
where
the first inequality uses lemma \ref{lemma-just-extended};
the second inequality uses the induction hypothesis and lemma \ref{lemma:upsilon}; and
the final inequality uses the assumption that $z+c<\overline{w}$.
If the support of wages is $\left[\underline{w},\infty\right)$,
then $w_{R}^{\delta}\left(N\right)$ will be finite as long as the truncated distribution of wage offers is finite. 

In summary, $\underline{w}<w_{R}^{\delta}\left(0\right)<\cdots<w_{R}^{\delta}\left(N\right)<\overline{w}$.
\end{proof}

\section{Proof of Proposition \ref{prop:extend-pr-and-length} in the Main Text}
\label{sec:proof-app-search-behavior}

Proposition \ref{prop:extend-pr-and-length} in the main text is broken into 2 lemmas.
\begin{description}
\item[Increasing in $\delta$.] Lemma \ref{lemma:increading-delta} establishes that $w_{R}^{\delta} \left( n \right)$ is increasing in $\delta$,
  holding constant $\Delta$.
\item[Increasing in $\Delta$.] Lemma \ref{prop:increasing-Delta} establishes that $w_{R}^{\delta} \left( n \right)$ is increasing in $\Delta$,
  holding constant $\delta$.
\end{description}
The proofs are different in nature.
The probability that benefits are extended is modeled as a number that takes on all values in $\left[ 0,1 \right]$.
The length of the extension, however, takes on discrete values.

\subsection{Proof of Lemma \ref{lemma:increading-delta}}

\begin{lemma}[Increasing in $\delta$]
  \label{lemma:increading-delta}
  Assume a worker in proposition \ref{prop:1} has computed their sequence of reservation wages.
  Each reservation wage $w_{R}^{\delta}$ is increasing in $\delta$.
  In other words,
  the worker is more selective when it comes to accepting job offers when they perceive an extension to be more likely.
\end{lemma}

\begin{proof}
  The proof, again, goes by induction.
  I first establish that $w_{R}^{\delta}\left(0\right)$ is increasing in the belief that benefits are extended.
  Using the expression for $w_{R}^{\delta}\left(0\right)$ in \ref{eq:wR-delta-0}, define the function 
\begin{align*}
G\left[w_{R}^{\delta}\left(0\right);\delta\right] &=-w_{R}^{\delta}\left(0\right)+z\left(1-\beta\right)+\beta\delta\left[\int_{\underline{w}}^{w_{R}\left(\Delta\right)}w_{R}\left(\Delta\right)dF\left(w\right)+\int_{w_{R}\left(\Delta\right)}^{\overline{w}}wdF\left(w\right)\right]\\
 &\quad+\beta\left(1-\delta\right)\left[\int_{\underline{w}}^{w_{R}^{\delta}\left(0\right)}w_{R}^{\delta}\left(0\right)dF\left(w\right)+\int_{w_{R}^{\delta}\left(0\right)}^{\overline{w}}wdF\left(w\right)\right]\\
 &= -w_{R}^{\delta}\left(0\right)+z\left(1-\beta\right)+\beta\delta\Upsilon\left(w_{R}\left(\Delta\right)\right)+\beta\left(1-\delta\right)\Upsilon\left(w_{R}^{\delta}\left(0\right)\right)
\end{align*}
The implicit function theorem implies
\begin{align*}
\frac{\partial w_{R}^{\delta}\left(0\right)}{\partial\delta} & =-\frac{\partial G/\partial\delta}{\partial G/\partial w_{R}^{\delta}}\\
 & =-\frac{\beta\Upsilon\left(w_{R}\left(\Delta\right)\right)-\beta\Upsilon\left(w_{R}\left(0\right)\right)}{-1+\beta\left(1-\delta\right)\Upsilon^{\prime}\left(w_{R}^{\delta}\left(0\right)\right)}\\
 & =\beta\frac{\Upsilon\left(w_{R}\left(\Delta\right)\right)-\Upsilon\left(w_{R}\left(0\right)\right)}{1-\beta\left(1-\delta\right)F\left(w_{R}^{\delta}\left(0\right)\right)},
\end{align*}
where the last equality uses lemma \ref{lemma:upsilon}.
Because $\Upsilon$ is increasing (lemma \ref{lemma:upsilon}) and
$w_{R}\left(\Delta\right)>w_{R}^{\delta}\left(0\right)$ (lemma \ref{lemma-just-extended}),
the numerator of the expression is positive.
In addition, $1>\beta\left(1-\delta\right)F\left(w_{R}^{\delta}\left(0\right)\right)$ as each term on the right-hand side is less than $1$,
making the denominator positive.
Thus, $w_{R}^{\delta}\left(0\right)$ is increasing in $\delta$; that is, increasing in the perceived likelihood of benefits being extended. 

It follows that $w_{R}^{\delta}\left(1\right)$ is increasing in $\delta$.
Expressing $w_{R}^{\delta} \left( 1 \right)$ as
\begin{align*}
w_{R}^{\delta}\left(1\right) & =\left(z+c\right)\left(1-\beta\right)+\beta\delta\left[\int_{\underline{w}}^{w_{R}\left(\Delta\right)}w_{R}\left(\Delta\right)dF\left(w\right)+\int_{w_{R}\left(\Delta\right)}^{\overline{w}}wdF\left(w\right)\right]\\
 & \quad+\beta\left(1-\delta\right)\left[\int_{\underline{w}}^{w_{R}^{\delta}\left(0\right)}w_{R}^{\delta}\left(0;\delta\right)dF\left(w\right)+\int_{w_{R}^{\delta}\left(0;\delta\right)}^{\overline{w}}wdF\left(w\right)\right]\\
 & =\left(z+c\right)\left(1-\beta\right)+\beta\delta\Upsilon\left(w_{R}\left(\Delta\right)\right)+\beta\left(1-\delta\right)\Upsilon\left(w_{R}^{\delta}\left(0\right)\right)
\end{align*}
and differentiating the latter expression with respect to $\delta$ implies 
\begin{align*}
\frac{\partial w_{R}^{\delta}\left(1\right)}{\partial\delta} & =\beta\Upsilon\left(w_{R}\left(\Delta\right)\right)\\
 & \quad+\beta\left\{ -\Upsilon\left(w_{R}^{\delta}\left(0\right)\right)+\left(1-\delta\right)\Upsilon^{\prime}\left(w_{R}^{\delta}\left(0\right)\right)\frac{\partial w_{R}^{\delta}\left(0\right)}{\partial\delta}\right\} \\
 & =\beta\left[\Upsilon\left(w_{R}\left(\Delta\right)\right)-\Upsilon\left(w_{R}^{\delta}\left(0\right)\right)\right]+\beta\left(1-\delta\right)F\left(w_{R}^{\delta}\left(0\right)\right)\frac{\partial w_{R}^{\delta}\left(0\right)}{\partial\delta},
\end{align*}
using lemma \ref{lemma:upsilon}. Because $w_{R}\left(\Delta\right)>w_{R}^{\delta}\left(0\right)$
from lemma \ref{lemma-just-extended} and the fact that $\Upsilon$
is increasing,
$\Upsilon \left( w_{R}\left(\Delta\right) \right) > \Upsilon \left( w_{R}^{\delta}\left(0\right) \right)$.
Thus, $\partial w_{R}^{\delta}\left(1\right)/\partial\delta>0$. 

The next step uses induction.
I want to establish that $\partial w_{R}^{\delta}\left(n+1\right)/\partial\delta>0$ using $\partial w_{R}^{\delta}\left(n\right)/\partial\delta>0$.
Indeed,
\begin{align*}
w_{R}^{\delta}\left(n+1\right) & =\left(z+c\right)\left(1-\beta\right)+\beta\left[\delta\Upsilon\left(w_{R}\left(n+\Delta\right)\right)+\left(1-\delta\right)\Upsilon\left(w_{R}^{\delta}\left(n\right)\right)\right].
\end{align*}
Evaluating the derivative yields
\begin{align*}
\frac{\partial w_{R}^{\delta}\left(d+1\right)}{\partial\delta} & =\beta\left\{ \Upsilon\left(w_{R}\left(n+\Delta\right)\right)-\Upsilon\left(w_{R}^{\delta}\left(n\right)\right)+\left(1-\delta\right)\Upsilon^{\prime}\left(w_{R}^{\delta}\left(n\right)\right)\frac{\partial w_{R}^{\delta}\left(n\right)}{\partial\delta}\right\} \\
 & =\beta\left[\Upsilon\left(w_{R}\left(n+\Delta\right)\right)-\Upsilon\left(w_{R}^{\delta}\left(n\right)\right)\right]+\beta\left(1-\delta\right)F\left(w_{R}^{\delta}\left(n\right)\right)\frac{\partial w_{R}^{\delta}\left(n\right)}{\partial\delta},
\end{align*}
where the last line uses the expression for the derivative in lemma
\ref{lemma:upsilon}.
Lemma \ref{lemma-just-extended} implies $w_{R}\left(n+\Delta\right)>w_{R}^{\delta} \left( n \right)$ and
lemma \ref{lemma:upsilon} implies the first term is positive.
The second term is positive by the induction hypothesis.
Therefore $\partial w_{R}^{\delta}\left(n+1\right)/\partial\delta>0$. 
\end{proof}

\subsection{Proof of Lemma \ref{prop:increasing-Delta}}

\begin{lemma}[Increasing in $\Delta$]
  \label{prop:increasing-Delta}
  Assume a worker in proposition \ref{prop:1} has computed their sequence of reservation wages.
  Each reservation wage $w_{R}^{\delta}$ is increasing in $\Delta$.
  In other words,
  the worker is more selective when it comes to accepting job offers when they perceive the length of an extension to increase,
  holding constant the chance of an extension.  
\end{lemma}

\begin{proof}
The proof uses induction. 

First, $w_{R}^{\delta}\left(0\right)$ is increasing in $\Delta$.
To see that this is the case, I take $\Delta^{\bullet},\Delta^{\bullet\bullet}\in\left\{ 1,2,\dots\right\}$
with $\Delta^{\bullet}<\Delta^{\bullet\bullet}$.
The difference evaluates to 
\begin{align*}
w_{R}^{\delta}\left(0;\Delta^{\bullet\bullet}\right)-w_{R}^{\delta}\left(0;\Delta^{\bullet}\right) & =\left(z+c\right)\left(1-\beta\right)+\beta\left[\delta\Upsilon\left(w_{R}\left(\Delta^{\bullet\bullet}\right)\right)+\left(1-\delta\right)\Upsilon\left(w_{R}^{\delta}\left(0;\Delta^{\bullet\bullet}\right)\right)\right]\\
 &\quad-\left(z+c\right)\left(1-\beta\right)-\beta\left[\delta\Upsilon\left(w_{R}\left(\Delta^{\bullet}\right)\right)+\left(1-\delta\right)\Upsilon\left(w_{R}^{\delta}\left(0;\Delta^{\bullet}\right)\right)\right]\\
 &=\beta\delta\left[\Upsilon\left(w_{R}\left(\Delta^{\bullet\bullet}\right)\right)-\Upsilon\left(w_{R}\left(\Delta^{\bullet}\right)\right)\right] \\
  &\quad +\beta\left(1-\delta\right)\left[\Upsilon\left(w_{R}^{\delta}\left(0;\Delta^{\bullet\bullet}\right)\right)-\Upsilon\left(w_{R}^{\delta}\left(0;\Delta^{\bullet}\right)\right)\right].
\end{align*}
This expression, using the definition of $\Psi$ given in lemma \ref{lemma:Psi}, can be re-expressed as
\begin{align*}
\Psi\left(w_{R}^{\delta}\left(0;\Delta^{\bullet\bullet}\right)\right)-\Psi\left(w_{R}^{\delta}\left(0;\Delta^{\bullet}\right)\right)=\beta\delta\left[\Upsilon\left(w_{R}\left(\Delta^{\bullet\bullet}\right)\right)-\Upsilon\left(w_{R}\left(\Delta^{\bullet}\right)\right)\right]>0,
\end{align*}
where the inequality follows from the fact that $w_{R}\left(\Delta^{\bullet\bullet}\right)>w_{R}\left(\Delta^{\bullet}\right)$
(proposition \ref{prop:text:basic-incr-reservation-wages} in the main text) and
the fact that $\Upsilon$ is increasing (lemma \ref{lemma:upsilon}).
Because
$\Psi\left(w_{R}^{\delta}\left(0;\Delta^{\bullet\bullet}\right)\right)>\Psi\left(w_{R}^{\delta}\left(0;\Delta^{\bullet}\right)\right)$ and
$\Psi$ is increasing, as established in lemma \ref{lemma:Psi},
it follows that $w_{R}^{\delta}\left(0;\Delta^{\bullet}\right)<w_{R}^{\delta}\left(0;\Delta^{\bullet\bullet}\right)$.

The remainder of the proof goes by induction.
I first show that $w_{R}^{\delta}\left(1;\Delta^{\bullet}\right)<w_{R}^{\delta}\left(1;\Delta^{\bullet\bullet}\right)$.
I then show that $w_{R}^{\delta}\left(n;\Delta^{\bullet}\right)<w_{R}^{\delta}\left(n;\Delta^{\bullet\bullet}\right)$
implies $w_{R}^{\delta}\left(n;\Delta^{\bullet}\right)<w_{R}^{\delta}\left(n;\Delta^{\bullet\bullet}\right)$. 

It is true that $w_{R}^{\delta}\left(1;\Delta^{\bullet}\right)<w_{R}^{\delta}\left(1;\Delta^{\bullet\bullet}\right)$.
Using expression for $w_{R}^{\delta}$ in \eqref{eq:wR-delta}, 
\begin{align*}
w_{R}^{\delta}\left(1;\Delta^{\bullet\bullet}\right)-w_{R}^{\delta}\left(1;\Delta^{\bullet}\right) & =\left(z+c\right)\left(1-\beta\right)+\beta\delta\Upsilon\left(w_{R}\left(\Delta^{\bullet\bullet}\right)\right)+\beta\left(1-\delta\right)\Upsilon\left(w_{R}^{\delta}\left(0;\Delta^{\bullet\bullet}\right)\right)\\
 & \quad-\left(z+c\right)\left(1-\beta\right)-\beta\delta\Upsilon\left(w_{R}\left(\Delta^{\bullet}\right)\right)-\beta\left(1-\delta\right)\Upsilon\left(w_{R}^{\delta}\left(0;\Delta^{\bullet}\right)\right)\\
 &= \beta\delta\left[\Upsilon\left(w_{R}\left(\Delta^{\bullet\bullet}\right)\right)-\Upsilon\left(w_{R}\left(\Delta^{\bullet}\right)\right)\right] \\
 &\quad +\beta\left(1-\delta\right)\left[\Upsilon\left(w_{R}^{\delta}\left(0;\Delta^{\bullet\bullet}\right)\right)-\Upsilon\left(w_{R}^{\delta}\left(0;\Delta^{\bullet}\right)\right)\right].
\end{align*}
The first part of this proof establishes $w_{R}^{\delta}\left(0;\Delta^{\bullet\bullet}\right)>w_{R}^{\delta}\left(0;\Delta^{\bullet}\right)$.
Proposition \ref{prop:text:basic-incr-reservation-wages} in the main text establishes that $w_{R}\left(\Delta^{\bullet\bullet}\right)>w_{R}\left(\Delta^{\bullet}\right)$.
From lemma \ref{lemma:upsilon}, $\Upsilon$ is increasing.
These facts imply that both terms in square brackets are positive,
making the right side of the latter positive.
Therefore $w_{R}^{\delta}\left(1;\Delta^{\bullet\bullet}\right)>w_{R}^{\delta}\left(1;\Delta^{\bullet}\right)$.

The next step of the proof uses the fact that $w_{R}^{\delta}\left(n;\Delta^{\bullet\bullet}\right)>w_{R}^{\delta}\left(n;\Delta^{\bullet}\right)$
to prove that $w_{R}^{\delta}\left(n+1;\Delta^{\bullet\bullet}\right)>w_{R}^{\delta}\left(n+1;\Delta^{\bullet}\right)$.
Similar steps, using the expression for $w_{R}^{\delta}$ in \eqref{eq:wR-delta},
establish the result:
\begin{align*}
w_{R}^{\delta}\left(n+1;\Delta^{\bullet\bullet}\right) & =\left(z+c\right)\left(1-\beta\right)+\beta\delta\Upsilon\left(w_{R}\left(n+\Delta^{\bullet\bullet}\right)\right)+\beta\left(1-\delta\right)\Upsilon\left(w_{R}^{\delta}\left(n;\Delta^{\bullet\bullet}\right)\right)
\end{align*}
and thus
\begin{align*}
w_{R}^{\delta}\left(n+1;\Delta^{\bullet\bullet}\right)-w_{R}^{\delta}\left(n+1;\Delta^{\bullet}\right) & =\beta\delta\left[\Upsilon\left(w_{R}\left(n+\Delta^{\bullet\bullet}\right)\right)-\Upsilon\left(w_{R}\left(n+\Delta^{\bullet}\right)\right)\right]\\
 & \quad+\beta\left(1-\delta\right)\left[\Upsilon\left(w_{R}^{\delta}\left(n;\Delta^{\bullet\bullet}\right)\right)-\Upsilon\left(w_{R}^{\delta}\left(n;\Delta^{\bullet}\right)\right)\right].
\end{align*}
Proposition \ref{prop:text:basic-incr-reservation-wages} in the main text
establishes that $w_{R}\left(n+\Delta^{\bullet\bullet}\right)>w_{R}\left(n+\Delta^{\bullet}\right)$.
The induction hypothesis assumes $w_{R}^{\delta}\left(n;\Delta^{\bullet\bullet}\right)>w_{R}^{\delta}\left(n;\Delta^{\bullet}\right)$.
In addition, $\Upsilon$ is increasing by lemma \ref{lemma:upsilon}.
Therefore the right side is positive.
Thus, $w_{R}^{\delta}\left(n;\Delta^{\bullet\bullet}\right)>w_{R}^{\delta}\left(n;\Delta^{\bullet}\right)$ for positive $n$.
This completes the proof. 
\end{proof}

\section{How $\beta$, $c$, and $z$ Affect Reservation Wages}
\label{sec:prop-reserv-wages}

The main text shows how the probability of an extension and the extension's length affect reservation wages.
  How a worker discounts the future and the value of nonwork affect a worker's sequence of reservation wages.
  The effect is intuitive,
  as figure \ref{fig:seq-reservation-wages-properties} illustrates.

For a comparison,
figure \ref{fig:seq-reservation-wages-properties} reports in purple reservation wages for two cases shown in figure \ref{fig:seq-res-wages}.
The purple solid line reports the sequence of reservation wages upon an extension and
the purple broken line reports the sequence of reservation wages when the probability of extension is perceived to be $0.5$.
When $\beta$ increases, workers are more patient.
They are more selective throughout unemployment spells and
their sequences of reservation wages shift upward.
The higher-$\beta$ case is depicted in dark blue.
When $c$ is lower, workers receive less in UI compensation each period.
They are less selective selective through unemployment spells and
their sequences of reservation wages shift downward.
This is depicted in light blue.
Because workers have the same $\beta$ and the same $z$ in the baseline and low-$c$ parameterizations,
the reservation wage upon an extension when there are no remaining periods of UI compensation are the same.
In contrast,
when $z$ is lower, which is depicted in green,
sequences of reservation wages are lower throughout workers' unemployment spells.

\begin{figure}[htbp]
\centerline{\includegraphics[width=0.8\textwidth]{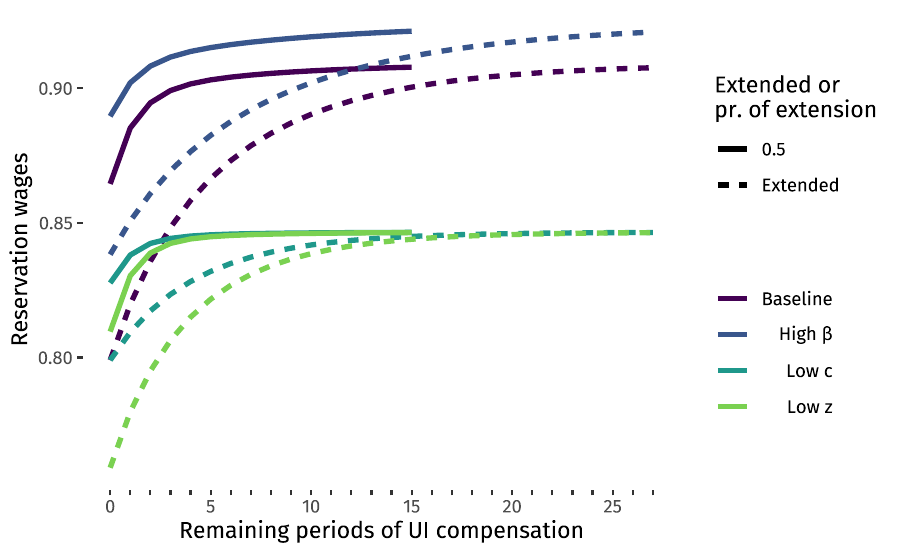}}
\caption[]{\label{fig:seq-reservation-wages-properties} How $\beta$, $c$, and $z$ affect reservation wages.}
  \begin{figurenotes}[Note]
    The baseline reports the values from figure \ref{fig:seq-res-wages} in purple.
    Relative to the baseline,
    the effect of increasing the discount factor and lowering UI compensation are illustrated in blue and green.
    Solid lines correspond to instances when benefits have been extended.
    Broken lines correspond to instances where a worker believes the probability of extension is $0.5$.
    Parameter changes affect the entire sequence of reservation wages.
\end{figurenotes} 
\end{figure}

\section{Expressions When Wages Offers are Characterized by a Uniform Distribution}
\label{sec:closed-form-uniform}

When job offers arrive from a uniform distribution,
optimal policy can be expressed in closed form.
I start with the characterization of job search under expiring benefits when there is no chance of an extension.

\subsection{Basic Job Search with Finite UI Benefits}
\label{sec:closed-form-uniform:finite}

As the proof of proposition \ref{prop:text:basic-incr-reservation-wages} in section \ref{sec:proof-basic} suggests,
the characterization starts with $w_{R}\left(0\right)$.
An implicit expression for $w_{R}\left(0\right)$ is given in equation \eqref{eq:wR0}.
When wage offers are distributed uniformly over $\left[0,1\right]$,
the expression for $w_{R}\left(0\right)$ can be developed as 
\begin{align*}
w_{R}\left(0\right) & =z\left(1-\beta\right)+\beta\left\{ \int_{\underline{w}}^{w_{R}\left(0\right)}w_{R}\left(0\right)dF\left(w\right)+\int_{w_{R}\left(0\right)}^{\overline{w}}wdF\left(w\right)\right\} \\
 & =z\left(1-\beta\right)+\beta\left\{ w_{R}\left(0\right)\int_{0}^{w_{R}\left(0\right)}dF\left(w\right)+\int_{w_{R}\left(0\right)}^{1}wdF\left(w\right)\right\} \\
 & =z\left(1-\beta\right)+\beta\left\{ w_{R}\left(0\right)F\left(w_{R}\left(0\right)\right)+\frac{w^{2}}{2}\Bigg\vert_{w=w_{R}\left(0\right)}^{w=1}\right\} \\
 & =z\left(1-\beta\right)+\beta\left\{ \left[w_{R}\left(0\right)\right]^{2}+\frac{1}{2}-\frac{\left[w_{R}\left(0\right)\right]^{2}}{2}\right\} .
\end{align*}
Collecting terms yields a quadratic equation in terms of $w_{R}\left(0\right)$:
\begin{align*}
0=\frac{\beta}{2}\left[w_{R}\left(0\right)\right]^{2}-w_{R}\left(0\right)+z\left(1-\beta\right)+\frac{\beta}{2}.
\end{align*}
The two roots of this expression are
\begin{align*}
w_{R}\left(0\right) & =\frac{1\pm\sqrt{1-4\left(\beta/2\right)\left[z\left(1-\beta\right)+\beta/2\right]}}{\beta}\\
 & =\frac{1\pm\sqrt{\left(1-\beta\right)\left(1+\beta-2\beta z\right)}}{\beta},
\end{align*}
where the second equality uses the fact that
\begin{align*}
\left(1-\beta\right)\left(1+\beta-2\beta z\right) & =1-\beta+\left(1-\beta\right)\left(\beta-2\beta z\right)\\
 & =1-\beta^{2}-2\beta z+2\beta^{2}z\\
 & =1-2\beta\left[z\left(1-\beta\right)+\beta/2\right].
\end{align*}
One of the roots yields a value for $w_{R}\left(0\right)$ that is
above one and this root is not of economic interest. The root that
is of economic interest is
\begin{equation}
w_{R}\left(0\right)=\frac{1-\sqrt{\left(1-\beta\right)\left(1+\beta-2\beta z\right)}}{\beta}.\label{eq:wR0:sol-uniform}
\end{equation}
For $\beta\in\left(0,1\right)$,
the reservation wages $w_{R}\left(0\right)$ will fall in the interior of the support of wage offers if $0<z<1$.
The restriction on the flow value of nonwork follows from the requirement that $0<w_{R}\left(0\right)<1$.
The restriction that $0<w_{R}\left(0\right)$ requires
\begin{align*}
0 & <\frac{1}{\beta}-\frac{\left[\left(1-\beta\right)\left(1+\beta-2\beta z\right)\right]^{1/2}}{\beta}\\
\iff1 & >\left(1-\beta\right)\left(1+\beta-2\beta z\right)\\
\iff-2\beta z & <\frac{1}{1-\beta}-\frac{\left(1+\beta\right)\left(1-\beta\right)}{1-\beta}\\
\iff z & >-\frac{1}{2}\frac{\beta}{1-\beta},
\end{align*}
which is guaranteed if $z>0$.
The restriction that $w_{R}\left(0\right)<1$ requires
\begin{align*}
\frac{1}{\beta}-\frac{\left[\left(1-\beta\right)\left(1+\beta-2\beta z\right)\right]^{1/2}}{\beta} & <1\\
\iff \left[\left(1-\beta\right)\left(1+\beta-2\beta z\right)\right]^{1/2} & >1-\beta\\
\iff 1+\beta-2\beta z & >1-\beta\\
\iff 1 & >z.
\end{align*}

Expressions for $w_{R}\left(n\right)$ are arrived at recursively.
Starting from \eqref{eq:wRn}, for $n=1,\dots,N$:
\begin{align*}
w_{R}\left(n\right) & =\left(z+c\right)\left(1-\beta\right)+\beta\left\{ \int_{\underline{w}}^{w_{R}\left(n-1\right)}w_{R}\left(n-1\right)dF\left(w\right)+\int_{w_{R}\left(n-1\right)}^{\overline{w}}wdF\left(w\right)\right\} \\
 & =\left(z+c\right)\left(1-\beta\right)+\beta\left\{ w_{R}\left(n-1\right)\int_{0}^{w_{R}\left(n-1\right)}dF\left(w\right)+\int_{w_{R}\left(n-1\right)}^{1}wdF\left(w\right)\right\} \\
 & =\left(z+c\right)\left(1-\beta\right)+\beta\left\{ \left[w_{R}\left(n-1\right)\right]^{2}+\frac{w^{2}}{2}\bigg\vert_{w=w_{R}\left(n-1\right)}^{w=1}\right\} ,
\end{align*}
which evaluates to
\begin{equation}
w_{R}\left(n\right)=\left(z+c\right)\left(1-\beta\right)+\frac{\beta}{2}\left\{ 1+\left[w_{R}\left(n-1\right)\right]^{2}\right\}. \label{eq:wRn:sol-uniform}
\end{equation}

\subsection{Allowing for an Extension to Benefits}
\label{sec:closed-form-uniform:chance-ext}

Reservation wages when the probability that benefits are extended with probability $\delta$ are implicitly given in
\eqref{eq:wR-delta-0} and \eqref{eq:map-wR-delta}.
Using the same techniques as those in section \ref{sec:closed-form-uniform:finite},
the expression for $w_{R}^{\delta}\left(0\right)$ can be developed as 
\begin{align*}
w_{R}^{\delta}\left(0\right) & =z\left(1-\beta\right)+\beta\delta\left[\int_{\underline{w}}^{w_{R}\left(\Delta\right)}w_{R}\left(\Delta\right)dF\left(w\right)+\int_{w_{R}\left(\Delta\right)}^{\overline{w}}wdF\left(w\right)\right]\\
 & \quad+\beta\left(1-\delta\right)\left[\int_{\underline{w}}^{w_{R}^{\delta}\left(0\right)}w_{R}^{\delta}\left(0\right)dF\left(w\right)+\int_{w_{R}^{\delta}\left(0\right)}^{\overline{w}}wdF\left(w\right)\right]\\
 & =z\left(1-\beta\right)+\frac{\beta\delta}{2}\left\{ 1+\left[w_{R}\left(\Delta\right)\right]^{2}\right\} +\frac{\beta\left(1-\delta\right)}{2}\left\{ 1+\left[w_{R}^{\delta}\left(0\right)\right]^{2}\right\} .
\end{align*}
The value $w_{R}\left(\Delta\right)$ is known to the worker. Thus
$w_{R}^{\delta}\left(0\right)$ is the root to the quadratic equation
\begin{align*}
0=\frac{\beta\left(1-\delta\right)}{2}\left[w_{R}^{\delta}\left(0\right)\right]^{2}-w_{R}^{\delta}\left(0\right)+z\left(1-\beta\right)+\frac{\beta\delta}{2}\left\{ 1+\left[w_{R}\left(\Delta\right)\right]^{2}\right\} +\frac{\beta\left(1-\delta\right)}{2}.
\end{align*}
The economically relevant root is 
\begin{equation}
w_{R}^{\delta}\left(0\right)=\frac{1-\sqrt{1-\beta\left(1-\delta\right)\left\{ \beta+2z\left(1-\beta\right)+\beta\delta\left[w_{R}\left(\Delta\right)\right]^{2}\right\} }}{\beta\left(1-\delta\right)}.\label{eq:wR0:sol-uniform:ext}
\end{equation}
When $\delta=0$, the expressions \eqref{eq:wR0:sol-uniform} and \eqref{eq:wR0:sol-uniform:ext} agree. 

As in section \ref{sec:closed-form-uniform:finite} above,
expressions for $w_{R}^{\delta}\left(n\right)$ are arrived at recursively.
Starting from \eqref{eq:wR-delta}, for $n \in \left\{ 1,\dots,N \right\}$:
\begin{align*}
w_{R}^{\delta}\left(n\right) & =\left(z+c\right)\left(1-\beta\right)+\beta\delta\left[\int_{\underline{w}}^{w_{R}\left(n-1+\Delta\right)}w_{R}\left(n-1+\Delta\right)dF\left(w\right)+\int_{w_{R}\left(n-1+\Delta\right)}^{\overline{w}}wdF\left(w\right)\right]\\
 & \quad+\beta\left(1-\delta\right)\left[\int_{\underline{w}}^{w_{R}^{\delta}\left(n-1\right)}w_{R}^{\delta}\left(n-1\right)dF\left(w\right)+\int_{w_{R}^{\delta}\left(n-1\right)}^{\overline{w}}wdF\left(w\right)\right].
\end{align*}
Using the assumption about uniform wage offers,
the latter evaluates to
\begin{equation}
w_{R}^{\delta}\left(n\right)=\left(z+c\right)\left(1-\beta\right)+\frac{\beta}{2}\left\{ 1+\delta\left[w_{R}\left(n-1+\Delta\right)\right]^{2}+\left(1-\delta\right)\left[w_{R}^{\delta}\left(n-1\right)\right]^{2}\right\} .\label{eq:wRn:sol-uniform:ext}
\end{equation}

\subsection{Welfare}

Computer code that conducts the numerical experiments makes use of expressions for expected welfare.

The value function for a worker entitled to $n$ periods of UI compensation
and wage offer $w$ to accept or reject is
\begin{align*}
V\left(w,n\right)=\max_{\text{reject, accept}}\left\{ U\left(n-1\right),W\right\} .
\end{align*}
The value function can be written
\begin{align*}
V\left(w,n\right)=\begin{cases}
\frac{w_{R}\left(n\right)}{1-\beta}=z+c+\beta\int V\left(w^{\prime},n-1\right)dF\left(w^{\prime}\right) & \text{if }w\leq w_{R}\left(n\right)\\
\frac{w}{1-\beta} & \text{if }w\geq w_{R}\left(n\right).
\end{cases}
\end{align*}
Expected welfare is
\begin{align*}
\E\left[V\left(w,n\right)\right]=\int V\left(w,n\right)dF\left(w\right).
\end{align*}

If wages are uniformly distributed over $\left[0,1\right]$, then
expected welfare evaluates to
\begin{align*}
\E\left[V\left(w,n\right)\right] & =\int_{0}^{w_{R}\left(n\right)}\frac{w_{R}\left(n\right)}{1-\beta}dF\left(w\right)+\int_{w_{R}\left(n\right)}^{1}\frac{w}{1-\beta}dF\left(w\right)\\
 & =\frac{\left[w_{R}\left(n\right)\right]^{2}}{1-\beta}+\frac{1}{1-\beta}\frac{w^{2}}{2}\bigg\vert_{w=w_{R}\left(n\right)}^{1}\\
 & =\frac{\left[w_{R}\left(n\right)\right]^{2}}{1-\beta}+\frac{1-\left[w_{R}\left(n\right)\right]^{2}}{1-\beta}\frac{1}{2}\\
 & =\frac{1}{1-\beta}\frac{1+\left[w_{R}\left(n\right)\right]^{2}}{2}.
\end{align*}

\section{Documentation of EB and EUC Programs}
\label{sec:documentation-eb-euc}

The US Department of Labor publishes reports titled
``Extended Benefits Trigger Notice'' and
``Emergency Unemployment Compensation Trigger Notice.''
The Extended Benefits Trigger Notice reports whether a state has triggered on extended benefits and
when the extended-benefits-period began.
The Emergency Unemployment Compensation Trigger Notice reports
whether different tiers of the Emergency Unemployment Compensation program were active.

Background on the Emergency Unemployment Compensation program enacted in June 2008 is provided by
\citet{fujita_2010}, \citet{rothstein_2011}, and \citet{whittaker_isaacs_2014}.

These reports are accessible through the Unemployment Insurance Data dashboard (\url{https://oui.doleta.gov/unemploy/DataDashboard.asp}) by
clicking on ``Weekly Claims and Trigger Notice Archives,''
which brings up
\url{https://oui.doleta.gov/unemploy/claims_arch.asp}.



\section*{Supplementary material}

TRIGGER NOTICE NO. 2009 - 51
STATE EXTENDED BENEFIT (E.B.) INDICATORS UNDER P.L. 102-318
Effective January 03, 2010

| | | | | INDICATORS | | | | | STATUS |
|---|---|---|---|---|---|---|---|---|---|
| | | | | 13 Weeks Insured Unemployment Rate | Percent of Prior 2 Years | 3 months S.A. T.U.R. | Percent of prior | | Periods |
| | | | | | | | Year | Second Year | Available Weeks | Begin Date(B) End Date(E) |
| @ | | | Alabama | 3.39 | 176 | 10.7 | 184 | 297 | 20 | B 03-29-2009 |
| @ | | | Alaska | 5.30 | 154 | 8.6 | 126 | 136 | 20 | B 01-25-2009 |
| @ | | | Arizona | 3.83 | 229 | 9.1 | 146 | 227 | 20 | B 02-22-2009 |
| | | & | Arkansas | 4.70 | 156 | 7.4 | 137 | 145 | | E 09-26-2009 |
| @ | | | California | 4.85 | 160 | 12.4 | 153 | 217 | 20 | B 02-22-2009 |
| @ | | | Colorado | 3.05 | 243 | 7.0 | 134 | 170 | 13 | B 04-12-2009 |
| @ | | | Connecticut | 4.20 | 163 | 8.5 | 139 | 177 | 20 | B 02-15-2009 |
| @ | * | | Delaware | 3.19 | 141 | 8.5 | 157 | 236 | 20 | B 06-28-2009 |
| @ | | | District of Col | 1.50 | 151 | 11.7 | 151 | 208 | 20 | B 04-05-2009 |
| | * | & | Florida | 3.63 | 166 | 11.3 | 163 | 251 | | E 01-02-2010 |
| @ | * | | Georgia | 3.77 | 182 | 10.1 | 146 | 210 | 20 | B 02-22-2009 |
| | | & | Hawaii | 3.39 | 179 | 7.1 | 154 | 253 | | E 01-24-1981 |
| @ | | | Idaho | 4.30 | 160 | 9.0 | 160 | 272 | 20 | B 02-08-2009 |
| @ | | | Illinois | 4.13 | 168 | 10.8 | 158 | 200 | 20 | B 03-15-2009 |
| @ | | | Indiana | 3.26 | 136 | 9.7 | 149 | 210 | 20 | B 03-15-2009 |
| | * | & | Iowa | 2.98 | 177 | 6.7 | 155 | 176 | | E 06-04-1983 |
| @ | | | Kansas | 3.31 | 210 | 6.6 | 140 | 165 | 13 | B 07-05-2009 |
| @ | * | | Kentucky | 3.13 | 150 | 10.9 | 155 | 201 | 20 | B 02-22-2009 |
| | | & | Louisiana | 3.33 | 238 | 7.2 | 130 | 200 | | E 02-25-2006 |
| @ | | | Maine | 2.84 | 162 | 8.2 | 141 | 170 | 20 | B 03-29-2009 |
| | | & | Maryland | 3.13 | 169 | 7.3 | 148 | 202 | | E 07-31-1982 |
| @ | * | | Massachusetts | 3.82 | 143 | 9.0 | 155 | 204 | 20 | B 03-22-2009 |
| @ | | | Michigan | 4.83 | 137 | 15.0 | 163 | 205 | 20 | B 01-25-2009 |
| @ | | | Minnesota | 3.34 | 172 | 7.5 | 131 | 159 | 13 | B 03-29-2009 |
| | | & | Mississippi | 3.54 | 166 | 9.5 | 128 | 153 | | E 07-16-1983 |
| | | & | Missouri | 3.45 | 168 | 9.4 | 144 | 174 | | E 01-02-2010 |
| | | & | Montana | 4.09 | 205 | 6.5 | 135 | 180 | | E 07-18-2009 |
| | | & | Nebraska | 2.03 | 184 | 4.8 | 133 | 154 | | E 01-24-1981 |
| @ | | | Nevada | 5.27 | 173 | 12.9 | 167 | 258 | 20 | B 02-22-2009 |
| @ | * | | New Hampshire | 2.71 | 196 | 6.9 | 172 | 202 | 13 | B 08-02-2009 |
| @ | | | New Jersey | 4.32 | 140 | 9.7 | 161 | 225 | 20 | B 03-15-2009 |
| @ | | | New Mexico | 3.18 | 189 | 7.8 | 173 | 222 | 13 | B 09-06-2009 |
| @ | | | New York | 3.53 | 154 | 8.8 | 146 | 191 | 20 | B 03-29-2009 |
| @ | | | North Carolina | 4.67 | 176 | 10.8 | 152 | 220 | 20 | B 10-05-2008 |
| | * | & | North Dakota | 1.33 | 156 | 4.2 | 127 | 135 | | E 06-11-1983 |
| @ | | | Ohio | 3.43 | 158 | 10.4 | 150 | 182 | 20 | B 03-15-2009 |
| | | & | Oklahoma | 2.73 | 245 | 7.0 | 166 | 179 | | E 01-24-1981 |
| @ | | | Oregon | 5.65 | 163 | 11.2 | 153 | 211 | 20 | B 12-07-2008 |
| @ | | | Pennsylvania | 5.04 | 159 | 8.7 | 150 | 193 | 20 | B 02-15-2009 |
| # | | & | Puerto Rico | 6.22 | 160 | 15.9 | 129 | 140 | 13 | B 01-25-2009 |
| @ | | | Rhode Island | 3.78 | 132 | 12.9 | 146 | 230 | 20 | B 07-06-2008 |
| @ | | | South Carolina | 4.53 | 155 | 12.0 | 153 | 214 | 20 | B 03-08-2009 |
| | * | & | South Dakota | 1.25 | 231 | 4.9 | 148 | 175 | | E 01-24-1981 |
| @ | | | Tennessee | 2.86 | 162 | 10.4 | 148 | 200 | 20 | B 02-22-2009 |
| @ | | | Texas | 2.51 | 198 | 8.2 | 154 | 186 | 20 | B 05-03-2009 |
| | * | & | Utah | 2.74 | 220 | 6.3 | 180 | 217 | | E 06-25-1983 |
| @ | | | Vermont | 3.40 | 151 | 6.5 | 130 | 162 | 13 | B 03-22-2009 |
| | | & | Virgin Islands | 2.78 | 144 | 8.0 | 119 | 166 | | E 08-06-1983 |
| @ | | | Virginia | 2.02 | 176 | 6.6 | 153 | 206 | 13 | B 05-03-2009 |
| @ | * | | Washington | 4.71 | 202 | 9.2 | 158 | 200 | 20 | B 02-15-2009 |
| @ | | | West Virginia | 3.55 | 185 | 8.6 | 200 | 195 | 20 | B 06-07-2009 |

| | | | | | | | | | | |
|---|---|---|---|---|---|---|---|---|---|---|
| @ | | | **Wisconsin** | 5.10 | 178 | 8.3 | 166 | 180 | 20 | B 02-22-2009 |
| | * | & | **Wyoming** | 2.83 | 293 | 7.1 | 221 | 253 | | E 06-13-1987 |

| | | | | | |
|---|---|---|---|---|---|
| Total Number "ON":  37 | | | 1 | | 36 |

\* - State does not have 6 % I.U.R option in law  
& - State does not have T.U.R option in law                                                                                         December 31, 2009  
@ - State "ON" by 3-month average TUR  
# - State "ON" by 13-week IUR  
I.U.R reflects 13-week period ending December 19, 2009  
T.U.R reflects avg. seasonally adjusted T.U.R for 3-month period ending November, 2009

EUC 2008 TRIGGER NOTICE NO. 2009 - 51
THIRD AND FOURTH TIER EUC 2008 TRIGGERS UNDER P.L. 111-92 Effective January 3, 2010

| State | 13 Weeks IUR | 3 months SA TUR | Tier Three Status | Tier Three Effective Date | Tier Four Status | Tier Four Effective Date |
|---|---|---|---|---|---|---|
| Alabama | 3.39 | 10.7 | ON | 11/8/2009 | ON | 11/8/2009 |
| Alaska | 5.30 | 8.6 | ON | 11/8/2009 | ON | 1/3/2010 |
| Arizona | 3.83 | 9.1 | ON | 11/8/2009 | ON | 11/8/2009 |
| Arkansas | 4.70 | 7.4 | ON | 11/8/2009 | | |
| California | 4.85 | 12.4 | ON | 11/8/2009 | ON | 11/8/2009 |
| Colorado | 3.05 | 7.0 | ON | 11/8/2009 | | |
| Connecticut | 4.20 | 8.5 | ON | 11/8/2009 | ON | 1/3/2010 |
| Delaware | 3.19 | 8.5 | ON | 11/8/2009 | ON | 1/3/2010 |
| District of Columbia | 1.50 | 11.7 | ON | 11/8/2009 | ON | 11/8/2009 |
| Florida | 3.63 | 11.3 | ON | 11/8/2009 | ON | 11/8/2009 |
| Georgia | 3.77 | 10.1 | ON | 11/8/2009 | ON | 11/8/2009 |
| Hawaii | 3.39 | 7.1 | ON | 11/8/2009 | | |
| Idaho | 4.30 | 9.0 | ON | 11/8/2009 | ON | 11/8/2009 |
| Illinois | 4.13 | 10.8 | ON | 11/8/2009 | ON | 11/8/2009 |
| Indiana | 3.26 | 9.7 | ON | 11/8/2009 | ON | 11/8/2009 |
| Iowa | 2.98 | 6.7 | ON | 11/8/2009 | | |
| Kansas | 3.31 | 6.6 | ON | 11/8/2009 | | |
| Kentucky | 3.13 | 10.9 | ON | 11/8/2009 | ON | 11/8/2009 |
| Louisiana | 3.33 | 7.2 | ON | 11/8/2009 | | |
| Maine | 2.84 | 8.2 | ON | 11/8/2009 | OFF | 1/2/2010 |
| Maryland | 3.13 | 7.3 | ON | 11/8/2009 | | |
| Massachusetts | 3.82 | 9.0 | ON | 11/8/2009 | ON | 11/8/2009 |
| Michigan | 4.83 | 15.0 | ON | 11/8/2009 | ON | 11/8/2009 |
| Minnesota | 3.34 | 7.5 | ON | 11/8/2009 | | |
| Mississippi | 3.54 | 9.5 | ON | 11/8/2009 | ON | 11/8/2009 |
| Missouri | 3.45 | 9.4 | ON | 11/8/2009 | ON | 11/8/2009 |
| Montana | 4.09 | 6.5 | ON | 11/8/2009 | | |
| Nebraska | 2.03 | 4.8 | | | | |
| Nevada | 5.27 | 12.9 | ON | 11/8/2009 | ON | 11/8/2009 |
| New Hampshire | 2.71 | 6.9 | ON | 11/8/2009 | | |
| New Jersey | 4.32 | 9.7 | ON | 11/8/2009 | ON | 11/8/2009 |
| New Mexico | 3.18 | 7.8 | ON | 11/8/2009 | | |
| New York | 3.53 | 8.8 | ON | 11/8/2009 | ON | 11/8/2009 |
| North Carolina | 4.67 | 10.8 | ON | 11/8/2009 | ON | 11/8/2009 |
| North Dakota | 1.33 | 4.2 | | | | |
| Ohio | 3.43 | 10.4 | ON | 11/8/2009 | ON | 11/8/2009 |

| State | I.U.R | T.U.R | Tier 2 | Tier 2 Eff. | Tier 3/4 | Tier 3/4 Eff. |
|---|---|---|---|---|---|---|
| **Oklahoma** | 2.73 | 7.0 | ON | 11/8/2009 | | |
| **Oregon** | 5.65 | 11.2 | ON | 11/8/2009 | ON | 11/8/2009 |
| **Pennsylvania** | 5.04 | 8.7 | ON | 11/8/2009 | ON | 11/8/2009 |
| **Puerto Rico** | 6.22 | 15.9 | ON | 11/8/2009 | ON | 11/8/2009 |
| **Rhode Island** | 3.78 | 12.9 | ON | 11/8/2009 | ON | 11/8/2009 |
| **South Carolina** | 4.53 | 12.0 | ON | 11/8/2009 | ON | 11/8/2009 |
| **South Dakota** | 1.25 | 4.9 | | | | |
| **Tennessee** | 2.86 | 10.4 | ON | 11/8/2009 | ON | 11/8/2009 |
| **Texas** | 2.51 | 8.2 | ON | 11/8/2009 | | |
| **Utah** | 2.74 | 6.3 | ON | 11/8/2009 | | |
| **Vermont** | 3.40 | 6.5 | ON | 11/8/2009 | | |
| **Virgin Islands** | 2.78 | 8.0 | ON | 11/8/2009 | | |
| **Virginia** | 2.02 | 6.6 | ON | 11/8/2009 | | |
| **Washington** | 4.71 | 9.2 | ON | 11/8/2009 | ON | 11/8/2009 |
| **West Virginia** | 3.55 | 8.6 | ON | 11/8/2009 | ON | 11/8/2009 |
| **Wisconsin** | 5.10 | 8.3 | ON | 11/8/2009 | WILL END | 1/9/2010 |
| **Wyoming** | 2.83 | 7.1 | ON | 11/8/2009 | | |

All states are eligible for up to 34 weeks of First and Second Tier benefits.
50 states are eligible for up to 13 weeks of Third Tier benefits as of January 3, 2010.
31 states are eligible for up to 6 weeks of Fourth Tier benefits as of January 3, 2010.
I.U.R reflects 13-week period ending December 19, 2009.
T.U.R reflects avg. seasonally adjusted T.U.R for 3-month period ending November, 2009.
* The effective date represents when states could begin making payments to claimants, and not when the state met
trigger criteria to be "on" in any tier of benefits.

# TRIGGER NOTICE NO. 2010 - 24
## STATE EXTENDED BENEFIT (E.B.) INDICATORS UNDER P.L. 102-318
### Effective June 27, 2010

| | | | | INDICATORS | | | | | STATUS |
| | | | 13 Weeks Insured Unemployment Rate | Percent of Prior 2 Years | 3 months S.A. T.U.R. | Percent of prior | | | Periods |
| | | | | | | Year | Second Year | Available Weeks | Begin Date(B) End Date(E) |
|---|---|---|---|---|---|---|---|---|---|
| | | Alabama | 2.76 | 97 | 10.9 | 113 | 247 | | E 06-05-2010 |
| @ | | Alaska | 6.12 | 131 | 8.4 | 109 | 133 | 20 | B 01-25-2009 |
| | & | Arizona | 3.36 | 120 | 9.6 | 107 | 192 | | E 06-12-2010 |
| | & | Arkansas | 3.84 | 97 | 7.8 | 109 | 162 | | E 09-26-2009 |
| | & | California | 4.73 | 111 | 12.5 | 114 | 195 | | E 06-12-2010 |
| | & | Colorado | 3.24 | 142 | 8.0 | 100 | 181 | | E 06-05-2010 |
| @ | | Connecticut | 4.35 | 114 | 9.0 | 112 | 176 | 20 | B 02-15-2009 |
| | & | Delaware | 3.20 | 106 | 9.0 | 115 | 214 | | E 06-05-2010 |
| | & | District of Col | 2.69 | 230 | 11.0 | 117 | 186 | | E 06-12-2010 |
| | & | Florida | 3.22 | 110 | 12.0 | 121 | 218 | | E 06-05-2010 |
| | & | Georgia | 2.88 | 101 | 10.3 | 111 | 187 | | E 06-12-2010 |
| | & | Hawaii | 3.14 | 119 | 6.7 | 98 | 197 | | E 01-24-1981 |
| | & | Idaho | 4.43 | 105 | 9.2 | 124 | 219 | | E 06-05-2010 |
| | & | Illinois | 4.01 | 100 | 11.2 | 116 | 189 | | E 06-05-2010 |
| | & | Indiana | 2.90 | 79 | 10.0 | 96 | 196 | | E 06-05-2010 |
| * | & | Iowa | 2.79 | 100 | 6.8 | 121 | 170 | | E 06-04-1983 |
| @ | | Kansas | 2.87 | 108 | 6.5 | 97 | 162 | 13 | B 07-05-2009 |
| | & | Kentucky | 3.02 | 91 | 10.6 | 101 | 179 | | E 06-05-2010 |
| | & | Louisiana | 3.10 | 161 | 6.8 | 104 | 178 | | E 02-25-2006 |
| | & | Maine | 3.42 | 107 | 8.1 | 100 | 168 | | E 06-12-2010 |
| | & | Maryland | 3.00 | 110 | 7.5 | 108 | 197 | | E 07-31-1982 |
| | & | Massachusetts | 3.77 | 97 | 9.2 | 115 | 195 | 0 | E 06-26-2010 |
| @ * | | Michigan | 4.42 | 79 | 13.9 | 106 | 187 | 20 | B 01-25-2009 |
| @ | | Minnesota | 3.29 | 102 | 7.2 | 86 | 144 | 13 | B 03-29-2009 |
| | & | Mississippi | 3.40 | 111 | 11.5 | 126 | 182 | | E 07-16-1983 |
| | & | Missouri | 3.20 | 105 | 9.4 | 103 | 170 | | E 06-05-2010 |
| | & | Montana | 4.48 | 120 | 7.2 | 122 | 171 | | E 06-19-2010 |
| | & | Nebraska | 2.06 | 124 | 5.0 | 108 | 166 | | E 01-24-1981 |
| | & | Nevada | 4.96 | 110 | 13.7 | 124 | 244 | | E 06-12-2010 |
| @ * | | New Hampshire | 3.31 | 118 | 6.7 | 111 | 186 | 13 | B 08-02-2009 |
| @ | | New Jersey | 4.26 | 101 | 9.8 | 111 | 200 | 20 | B 03-15-2009 |
| @ | | New Mexico | 3.60 | 132 | 8.6 | 130 | 215 | 20 | B 09-06-2009 |
| | & | New York | 3.46 | 104 | 8.5 | 104 | 177 | | E 06-05-2010 |
| @ | | North Carolina | 4.52 | 117 | 10.8 | 101 | 200 | 20 | B 10-05-2008 |
| * | & | North Dakota | 1.48 | 99 | 3.8 | 86 | 131 | | E 06-11-1983 |
| | & | Ohio | 3.20 | 88 | 10.9 | 109 | 181 | | E 06-05-2010 |
| | & | Oklahoma | 2.37 | 120 | 6.7 | 109 | 203 | | E 01-24-1981 |
| @ | | Oregon | 5.57 | 105 | 10.6 | 92 | 192 | 20 | B 12-07-2008 |
| | & | Pennsylvania | 4.84 | 102 | 9.1 | 116 | 185 | | E 06-05-2010 |
| # | & | Puerto Rico | 6.15 | 117 | 16.8 | 112 | 161 | 13 | B 01-25-2009 |
| @ | | Rhode Island | 4.28 | 98 | 12.4 | 119 | 179 | 20 | B 07-06-2008 |
| | & | South Carolina | 4.11 | 106 | 11.6 | 100 | 200 | | E 06-05-2010 |
| * | & | South Dakota | 1.23 | 111 | 4.7 | 95 | 167 | | E 01-24-1981 |
| | & | Tennessee | 2.67 | 92 | 10.5 | 100 | 175 | | E 06-05-2010 |
| | & | Texas | 2.17 | 109 | 8.3 | 113 | 184 | | E 06-05-2010 |
| * | & | Utah | 2.83 | 124 | 7.3 | 110 | 221 | | E 06-25-1983 |
| @ | | Vermont | 4.01 | 99 | 6.4 | 88 | 152 | 13 | B 03-22-2009 |
| | & | Virgin Islands | 1.85 | 86 | 5.2 | 100 | 152 | | E 08-06-1983 |
| | & | Virginia | 1.86 | 110 | 7.2 | 109 | 205 | | E 06-12-2010 |
| @ | | Washington | 4.23 | 117 | 9.3 | 105 | 193 | 20 | B 02-15-2009 |
| | & | West Virginia | 3.24 | 104 | 9.2 | 124 | 230 | | E 06-05-2010 |

|   |   | State | | | | | | |
|---|---|---|---|---|---|---|---|---|
|   | & | **Wisconsin** | 4.79 | 104 | 8.5 | 98 | 197 | E 06-12-2010 |
| * | & | **Wyoming** | 2.89 | 129 | 7.1 | 126 | 253 | E 06-13-1987 |

| Total Number 'ON': 14 | 1 | 13 |
|---|---|---|

\* - State does not have 6 % I.U.R option in law

& - State does not have T.U.R option in law         6/24/2010

@ - State "ON" by 3-month average T.U.R
# - State "ON" by 13-week I.U.R
IUR reflects 13-week period ending June 12, 2010
TUR reflects avg. seasonally adjusted TUR for 3-month period ending May, 2010

EUC 2008 TRIGGER NOTICE NO. 2010 - 24
THIRD AND FOURTH TIER EUC 2008 TRIGGERS UNDER P.L. 111-92
Effective June 27, 2010

| State | 13 Weeks IUR | 3 months SA TUR | Tier Three Status | Tier Three Effective Date | Tier Four Status | Tier Four Effective Date |
|---|---|---|---|---|---|---|
| Alabama | 2.76 | 10.9 | OFF | 5/29/2010 | OFF | 5/29/2010 |
| Alaska | 6.12 | 8.4 | OFF | 5/29/2010 | OFF | 5/29/2010 |
| Arizona | 3.36 | 9.6 | OFF | 5/29/2010 | OFF | 5/29/2010 |
| Arkansas | 3.84 | 7.8 | OFF | 5/29/2010 | OFF | 5/29/2010 |
| California | 4.73 | 12.5 | OFF | 5/29/2010 | OFF | 5/29/2010 |
| Colorado | 3.24 | 8.0 | OFF | 5/29/2010 | OFF | 5/29/2010 |
| Connecticut | 4.35 | 9.0 | OFF | 5/29/2010 | OFF | 5/29/2010 |
| Delaware | 3.20 | 9.0 | OFF | 5/29/2010 | OFF | 5/29/2010 |
| District of Columbia | 2.69 | 11.0 | OFF | 5/29/2010 | OFF | 5/29/2010 |
| Florida | 3.22 | 12.0 | OFF | 5/29/2010 | OFF | 5/29/2010 |
| Georgia | 2.88 | 10.3 | OFF | 5/29/2010 | OFF | 5/29/2010 |
| Hawaii | 3.14 | 6.7 | OFF | 5/29/2010 | OFF | 5/29/2010 |
| Idaho | 4.43 | 9.2 | OFF | 5/29/2010 | OFF | 5/29/2010 |
| Illinois | 4.01 | 11.2 | OFF | 5/29/2010 | OFF | 5/29/2010 |
| Indiana | 2.90 | 10.0 | OFF | 5/29/2010 | OFF | 5/29/2010 |
| Iowa | 2.79 | 6.8 | OFF | 5/29/2010 | OFF | 5/29/2010 |
| Kansas | 2.87 | 6.5 | OFF | 5/29/2010 | OFF | 5/29/2010 |
| Kentucky | 3.02 | 10.6 | OFF | 5/29/2010 | OFF | 5/29/2010 |
| Louisiana | 3.10 | 6.8 | OFF | 5/29/2010 | OFF | 5/29/2010 |
| Maine | 3.42 | 8.1 | OFF | 5/29/2010 | OFF | 5/29/2010 |
| Maryland | 3.00 | 7.5 | OFF | 5/29/2010 | OFF | 5/29/2010 |
| Massachusetts | 3.77 | 9.2 | OFF | 5/29/2010 | OFF | 5/29/2010 |
| Michigan | 4.42 | 13.9 | OFF | 5/29/2010 | OFF | 5/29/2010 |
| Minnesota | 3.29 | 7.2 | OFF | 5/29/2010 | OFF | 5/29/2010 |
| Mississippi | 3.40 | 11.5 | OFF | 5/29/2010 | OFF | 5/29/2010 |
| Missouri | 3.20 | 9.4 | OFF | 5/29/2010 | OFF | 5/29/2010 |
| Montana | 4.48 | 7.2 | OFF | 5/29/2010 | OFF | 5/29/2010 |
| Nebraska | 2.06 | 5.0 | OFF | 5/29/2010 | OFF | 5/29/2010 |
| Nevada | 4.96 | 13.7 | OFF | 5/29/2010 | OFF | 5/29/2010 |
| New Hampshire | 3.31 | 6.7 | OFF | 5/29/2010 | OFF | 5/29/2010 |
| New Jersey | 4.26 | 9.8 | OFF | 5/29/2010 | OFF | 5/29/2010 |
| New Mexico | 3.60 | 8.6 | OFF | 5/29/2010 | OFF | 5/29/2010 |
| New York | 3.46 | 8.5 | OFF | 5/29/2010 | OFF | 5/29/2010 |
| North Carolina | 4.52 | 10.8 | OFF | 5/29/2010 | OFF | 5/29/2010 |
| North Dakota | 1.48 | 3.8 | OFF | 5/29/2010 | OFF | 5/29/2010 |

| | | | | | | |
|---|---|---|---|---|---|---|
| **Ohio** | 3.20 | 10.9 | OFF | 5/29/2010 | OFF | 5/29/2010 |
| **Oklahoma** | 2.37 | 6.7 | OFF | 5/29/2010 | OFF | 5/29/2010 |
| **Oregon** | 5.57 | 10.6 | OFF | 5/29/2010 | OFF | 5/29/2010 |
| **Pennsylvania** | 4.84 | 9.1 | OFF | 5/29/2010 | OFF | 5/29/2010 |
| **Puerto Rico** | 6.15 | 16.8 | OFF | 5/29/2010 | OFF | 5/29/2010 |
| **Rhode Island** | 4.28 | 12.4 | OFF | 5/29/2010 | OFF | 5/29/2010 |
| **South Carolina** | 4.11 | 11.6 | OFF | 5/29/2010 | OFF | 5/29/2010 |
| **South Dakota** | 1.23 | 4.7 | OFF | 5/29/2010 | OFF | 5/29/2010 |
| **Tennessee** | 2.67 | 10.5 | OFF | 5/29/2010 | OFF | 5/29/2010 |
| **Texas** | 2.17 | 8.3 | OFF | 5/29/2010 | OFF | 5/29/2010 |
| **Utah** | 2.83 | 7.3 | OFF | 5/29/2010 | OFF | 5/29/2010 |
| **Vermont** | 4.01 | 6.4 | OFF | 5/29/2010 | OFF | 5/29/2010 |
| **Virgin Islands** | 1.85 | 5.2 | OFF | 5/29/2010 | OFF | 5/29/2010 |
| **Virginia** | 1.86 | 7.2 | OFF | 5/29/2010 | OFF | 5/29/2010 |
| **Washington** | 4.23 | 9.3 | OFF | 5/29/2010 | OFF | 5/29/2010 |
| **West Virginia** | 3.24 | 9.2 | OFF | 5/29/2010 | OFF | 5/29/2010 |
| **Wisconsin** | 4.79 | 8.5 | OFF | 5/29/2010 | OFF | 5/29/2010 |
| **Wyoming** | 2.89 | 7.1 | OFF | 5/29/2010 | OFF | 5/29/2010 |

The EUC law expired June 2, 2010. Under EUC law, the last date eligibility could be established for new participants was
May 22, 2010. Because EUC 2008 Trigger Notice No. 2010-24 is effective with respect to the week beginning on June 27 and
ending July 03, it now shows that states have triggered "off" Third and Fourth Tier EUC since claimants ending a tier may not
advance to the next tier. Under the phase-out provisions, claimants may receive their remaining entitlement for the tier in which
they are collecting.


I.U.R reflects 13-week period ending June 12, 2010.
T.U.R reflects avg. seasonally adjusted T.U.R for 3-month period ending May, 2010.

# TRIGGER NOTICE NO. 2010 - 28
## STATE EXTENDED BENEFIT (E.B.) INDICATORS UNDER P.L. 102-318
### Effective July 25, 2010

| | | | | INDICATORS | | | | | STATUS |
| | | | 13 Weeks Insured Unemployment Rate | Percent of Prior 2 Years | 3 months S.A. T.U.R. | Percent of prior | | | Periods |
| | | | | | | Year | Second Year | Available Weeks | Begin Date(B) End Date(E) |
|---|---|---|---|---|---|---|---|---|---|
| @ | | Alabama | 2.72 | 92 | 10.7 | 107 | 232 | 20 | B 03-29-2009 |
| @ | | Alaska | 5.46 | 130 | 8.2 | 105 | 128 | 20 | B 01-25-2009 |
| @ | | Arizona | 3.44 | 115 | 9.6 | 104 | 181 | 20 | B 02-22-2009 |
| | & | Arkansas | 3.76 | 94 | 7.7 | 106 | 157 | | E 09-26-2009 |
| @ | | California | 4.59 | 109 | 12.4 | 109 | 185 | 20 | B 02-22-2009 |
| @ | | Colorado | 3.20 | 141 | 8.0 | 97 | 177 | 20 | B 04-12-2009 |
| @ | | Connecticut | 4.24 | 111 | 8.9 | 108 | 167 | 20 | B 02-15-2009 |
| @ * | | Delaware | 3.10 | 103 | 8.8 | 111 | 200 | 20 | B 06-28-2009 |
| @ | | District of Col | 2.67 | 219 | 10.4 | 107 | 170 | 20 | B 04-05-2009 |
| @ * | | Florida | 3.27 | 93 | 11.7 | 114 | 205 | 20 | B 02-22-2009 |
| @ * | | Georgia | 2.95 | 99 | 10.2 | 108 | 178 | 20 | B 02-22-2009 |
| | & | Hawaii | 3.16 | 116 | 6.5 | 94 | 180 | | E 01-24-1981 |
| @ | | Idaho | 3.91 | 105 | 9.0 | 116 | 200 | 20 | B 02-08-2009 |
| @ | | Illinois | 3.82 | 99 | 10.8 | 108 | 174 | 20 | B 03-15-2009 |
| @ | | Indiana | 2.74 | 77 | 10.0 | 95 | 188 | 20 | B 03-15-2009 |
| * | & | Iowa | 2.54 | 97 | 6.8 | 119 | 165 | | E 06-04-1983 |
| @ | | Kansas | 2.95 | 109 | 6.5 | 92 | 154 | 13 | B 07-05-2009 |
| @ * | | Kentucky | 2.87 | 88 | 10.3 | 97 | 168 | 20 | B 02-22-2009 |
| | & | Louisiana | 3.09 | 147 | 6.9 | 101 | 172 | | E 02-25-2006 |
| @ | | Maine | 2.99 | 107 | 8.0 | 97 | 163 | 20 | B 03-29-2009 |
| | & | Maryland | 2.94 | 109 | 7.3 | 104 | 182 | | E 07-31-1982 |
| @ | | Massachusetts | 3.68 | 98 | 9.1 | 110 | 185 | 20 | B 03-22-2009 |
| @ * | | Michigan | 4.02 | 75 | 13.6 | 100 | 176 | 20 | B 01-25-2009 |
| @ | | Minnesota | 2.90 | 98 | 7.0 | 83 | 137 | 13 | B 03-29-2009 |
| | & | Mississippi | 3.58 | 111 | 11.3 | 121 | 171 | | E 07-16-1983 |
| @ | | Missouri | 3.10 | 103 | 9.3 | 100 | 163 | 20 | B 02-22-2009 |
| | & | Montana | 4.01 | 121 | 7.2 | 120 | 167 | | E 06-12-2010 |
| | & | Nebraska | 1.95 | 120 | 4.9 | 104 | 158 | | E 01-24-1981 |
| @ | | Nevada | 4.78 | 104 | 13.9 | 120 | 235 | 20 | B 02-22-2009 |
| @ * | | New Hampshire | 3.01 | 112 | 6.3 | 101 | 170 | 13 | B 08-02-2009 |
| @ | | New Jersey | 4.13 | 99 | 9.7 | 106 | 190 | 20 | B 03-15-2009 |
| @ | | New Mexico | 3.61 | 130 | 8.4 | 123 | 204 | 20 | B 09-06-2009 |
| @ | | New York | 3.35 | 103 | 8.3 | 98 | 169 | 20 | B 03-29-2009 |
| @ | | North Carolina | 4.24 | 110 | 10.4 | 96 | 185 | 20 | B 10-05-2008 |
| * | & | North Dakota | 1.19 | 97 | 3.7 | 84 | 123 | | E 06-11-1983 |
| @ | | Ohio | 3.22 | 91 | 10.7 | 103 | 172 | 20 | B 03-15-2009 |
| | & | Oklahoma | 2.35 | 114 | 6.7 | 104 | 197 | | E 01-24-1981 |
| @ | | Oregon | 5.30 | 104 | 10.6 | 91 | 185 | 20 | B 12-07-2008 |
| @ | | Pennsylvania | 4.71 | 100 | 9.1 | 113 | 182 | 20 | B 02-15-2009 |
| # | & | Puerto Rico | 6.45 | 116 | 16.8 | 113 | 152 | 13 | B 01-25-2009 |
| @ | | Rhode Island | 4.04 | 99 | 12.3 | 114 | 170 | 20 | B 07-06-2008 |
| @ | | South Carolina | 4.03 | 102 | 11.1 | 94 | 181 | 20 | B 03-08-2009 |
| * | & | South Dakota | 1.01 | 105 | 4.6 | 93 | 158 | | E 01-24-1981 |
| @ | | Tennessee | 2.64 | 90 | 10.3 | 96 | 163 | 20 | B 02-22-2009 |
| @ | | Texas | 2.19 | 106 | 8.2 | 109 | 178 | 20 | B 05-03-2009 |
| * | & | Utah | 2.67 | 122 | 7.3 | 107 | 214 | | E 06-25-1983 |
| | | Vermont | 3.60 | 99 | 6.2 | 86 | 144 | | E 07-10-2010 |
| | & | Virgin Islands | 2.06 | 90 | 5.2 | 89 | 126 | | E 08-06-1983 |
| @ | | Virginia | 1.81 | 109 | 7.1 | 104 | 197 | 13 | B 05-03-2009 |
| @ * | | Washington | 3.91 | 112 | 9.1 | 101 | 182 | 20 | B 02-15-2009 |
| @ | | West Virginia | 3.05 | 98 | 8.8 | 111 | 220 | 20 | B 06-07-2009 |

| | | | | | | | | | | |
|---|---|---|---|---|---|---|---|---|---|---|
| @ | | | **Wisconsin** | 4.41 | 102 | 8.2 | 93 | 186 | 20 | B 02-22-2009 |
| | * | & | **Wyoming** | 2.53 | 120 | 7.0 | 114 | 233 | | E 06-13-1987 |

| | | |
|---|---|---|
| Total Number 'ON': 38 | 1 | 37 |

\* - State does not have 6 % I.U.R option in law

& - State does not have T.U.R option in law  7/23/2010

@ - State "ON" by 3-month average T.U.R
\# - State "ON" by 13-week I.U.R
IUR reflects 13-week period ending July 10, 2010
TUR reflects avg. seasonally adjusted TUR for 3-month period ending June, 2010

EUC 2008 TRIGGER NOTICE NO. 2010 - 28
THIRD AND FOURTH TIER EUC 2008 TRIGGERS UNDER P.L. 111-92
Effective July 25, 2010

| State | 13 Weeks IUR | 3 months SA TUR | Tier Three Status | Tier Three Effective Date | Tier Four Status | Tier Four Effective Date |
|---|---|---|---|---|---|---|
| Alabama | 2.72 | 10.7 | ON | 11/8/2009 | ON | 11/8/2009 |
| Alaska | 5.62 | 8.2 | ON | 11/8/2009 | OFF | 6/12/2010 |
| Arizona | 3.42 | 9.6 | ON | 11/8/2009 | ON | 11/8/2009 |
| Arkansas | 3.77 | 7.7 | ON | 11/8/2009 | | |
| California | 4.64 | 12.4 | ON | 11/8/2009 | ON | 11/8/2009 |
| Colorado | 3.21 | 8.0 | ON | 11/8/2009 | | |
| Connecticut | 4.24 | 8.9 | ON | 11/8/2009 | ON | 1/3/2010 |
| Delaware | 3.13 | 8.8 | ON | 11/8/2009 | ON | 1/3/2010 |
| District of Columbia | 2.67 | 10.4 | ON | 11/8/2009 | ON | 11/8/2009 |
| Florida | 3.27 | 11.7 | ON | 11/8/2009 | ON | 11/8/2009 |
| Georgia | 2.93 | 10.2 | ON | 11/8/2009 | ON | 11/8/2009 |
| Hawaii | 3.18 | 6.5 | ON | 11/8/2009 | | |
| Idaho | 4.04 | 9.0 | ON | 11/8/2009 | ON | 11/8/2009 |
| Illinois | 3.82 | 10.8 | ON | 11/8/2009 | ON | 11/8/2009 |
| Indiana | 2.73 | 10.0 | ON | 11/8/2009 | ON | 11/8/2009 |
| Iowa | 2.58 | 6.8 | ON | 11/8/2009 | | |
| Kansas | 2.92 | 6.5 | ON | 11/8/2009 | | |
| Kentucky | 2.91 | 10.3 | ON | 11/8/2009 | ON | 11/8/2009 |
| Louisiana | 3.09 | 6.9 | ON | 11/8/2009 | | |
| Maine | 3.08 | 8.0 | ON | 11/8/2009 | | 1/2/2010 |
| Maryland | 2.94 | 7.3 | ON | 11/8/2009 | | |
| Massachusetts | 3.65 | 9.1 | ON | 11/8/2009 | ON | 11/8/2009 |
| Michigan | 4.08 | 13.6 | ON | 11/8/2009 | ON | 11/8/2009 |
| Minnesota | 2.98 | 7.0 | ON | 11/8/2009 | | |
| Mississippi | 3.52 | 11.3 | ON | 11/8/2009 | ON | 11/8/2009 |
| Missouri | 3.11 | 9.3 | ON | 11/8/2009 | ON | 11/8/2009 |
| Montana | 4.12 | 7.2 | ON | 11/8/2009 | OFF | 7/3/2010 |
| Nebraska | 1.97 | 4.9 | | | | |
| Nevada | 4.84 | 13.9 | ON | 11/8/2009 | ON | 11/8/2009 |
| New Hampshire | 3.07 | 6.3 | ON | 11/8/2009 | | |
| New Jersey | 4.13 | 9.7 | ON | 11/8/2009 | ON | 11/8/2009 |
| New Mexico | 3.62 | 8.4 | ON | 11/8/2009 | ON** | 4/11/2010 |
| New York | 3.35 | 8.3 | ON | 11/8/2009 | ON** | 11/8/2009 |
| North Carolina | 4.31 | 10.4 | ON | 11/8/2009 | ON | 11/8/2009 |
| North Dakota | 1.23 | 3.7 | | | | |

| State | I.U.R | T.U.R | Tier 3 | Tier 3 Date | Tier 4 | Tier 4 Date |
|---|---|---|---|---|---|---|
| **Ohio** | 3.21 | 10.7 | ON | 11/8/2009 | ON | 11/8/2009 |
| **Oklahoma** | 2.35 | 6.7 | ON | 11/8/2009 | | |
| **Oregon** | 5.38 | 10.6 | ON | 11/8/2009 | ON | 11/8/2009 |
| **Pennsylvania** | 4.73 | 9.1 | ON | 11/8/2009 | ON | 11/8/2009 |
| **Puerto Rico** | 6.42 | 16.8 | ON | 11/8/2009 | ON | 11/8/2009 |
| **Rhode Island** | 4.06 | 12.3 | ON | 11/8/2009 | ON | 11/8/2009 |
| **South Carolina** | 4.04 | 11.1 | ON | 11/8/2009 | ON | 11/8/2009 |
| **South Dakota** | 1.05 | 4.6 | | | | |
| **Tennessee** | 2.64 | 10.3 | ON | 11/8/2009 | ON | 11/8/2009 |
| **Texas** | 2.19 | 8.2 | ON | 11/8/2009 | | |
| **Utah** | 2.71 | 7.3 | ON | 11/8/2009 | | |
| **Vermont** | 3.70 | 6.2 | ON | 11/8/2009 | | |
| **Virgin Islands** | 1.88 | 5.2 | 0 | 1/0/1900 | | |
| **Virginia** | 1.82 | 7.1 | ON | 11/8/2009 | | |
| **Washington** | 3.99 | 9.1 | ON | 11/8/2009 | ON | 11/8/2009 |
| **West Virginia** | 3.09 | 8.8 | ON | 11/8/2009 | ON | 11/8/2009 |
| **Wisconsin** | 4.50 | 8.2 | ON | 11/8/2009 | ON** | 4/11/2010 |
| **Wyoming** | 2.62 | 7.0 | ON | 11/8/2009 | | |

All states are eligible for up to 34 weeks of First and Second Tier benefits.
49 states are eligible for up to 13 weeks of Third Tier benefits as of July 25, 2010.
28 states are eligible for up to 6 weeks of Fourth Tier benefits as of July 25, 2010.
I.U.R reflects 13-week period ending July 10, 2010.
T.U.R reflects avg. seasonally adjusted T.U.R for 3-month period ending June, 2010.
* The effective date represents the date a state can begin making payments to claimants,
and not the date a state met trigger criteria to be "on" in any tier of benefits.
** these states will end their period August 14, 2010



\end{document}
